\documentclass[11pt]{article}
\usepackage{liyang}
 
\usepackage{tikz}
\usepackage{verbatim}
\usepackage{cancel}
\usepackage{setspace}

\makeatletter
\let\amsmath@bigm\bigm
\renewcommand{\bigm}[1]{%
  \ifcsname fenced@\string#1\endcsname
    \expandafter\@firstoftwo
  \else
    \expandafter\@secondoftwo
  \fi
  {\expandafter\amsmath@bigm\csname fenced@\string#1\endcsname}%
  {\amsmath@bigm#1}%
}

\newcommand{\DeclareFence}[2]{\@namedef{fenced@\string#1}{#2}}
\makeatother

\DeclareFence{\mid}{|}

\newcommand{\rnote}[1]{\footnote{{\bf \color{red}Rocco}: {#1}}}

\newcommand{\xnote}[1]{\footnote{{\bf \color{red}Xi}: {#1}}}

\usepackage{todonotes}

\makeatletter
\newtheorem*{rep@theorem}{\rep@title}
\newcommand{\newreptheorem}[2]{
\newenvironment{rep#1}[1]{
 \def\rep@title{#2 \ref{##1}}
 \begin{rep@theorem}\itshape}
 {\end{rep@theorem}}}
\makeatother
\theoremstyle{plain}

\makeatletter
\newenvironment{proofof}[1]{\par
  \pushQED{\qed}%
  \normalfont \topsep6\p@\@plus6\p@\relax
  \trivlist
  \item[\hskip\labelsep
\emph{    Proof of #1\@addpunct{.}}]\ignorespaces
}{%
  \popQED\endtrivlist\@endpefalse
}
\makeatother

\newcommand{\ignore}[1]{}

\def\colorful{1}

\ifnum\colorful=1

\newcommand{\blue}[1]{{{\color{blue}#1}}}
\newcommand{\red}[1]{{\color{brown} {#1}}}

\newcommand{\gray}[1]{{\color{gray}{#1}}}

\fi
\ifnum\colorful=0

\newcommand{\blue}[1]{{{#1}}}
\newcommand{\red}[1]{{{#1}}}

\newcommand{\gray}[1]{{{#1}}}

\fi

\usepackage{boxedminipage}

\newreptheorem{theorem}{Theorem}
\newtheorem*{theorem*}{Theorem}
\newreptheorem{lemma}{Lemma}
\newreptheorem{proposition}{Proposition}
\newreptheorem{corollary}{Corollary}
\newtheorem*{noclaim*}{Claim}

\newcommand{\Bin}{\mathrm{Bin}}

\def\curr{\mathrm{current}}
\def\dist{\mathsf{distance}}
\def\pos{\mathsf{position}}
\def\upper{\mathsf{upper}}
\def\ideal{\mathsf{ideal}}

\newcommand{\Del}{\mathrm{Del}}

\def\dedit{{d_{\mathrm{edit}}}}

\newcommand{\source}{\mathrm{source}}
\newcommand{\image}{\mathrm{image}}
\newcommand{\bssx}{\text{\boldmath${\mathsf{x}}$}}
\newcommand{\bssy}{\text{\boldmath${\mathsf{y}}$}}
\newcommand{\bssz}{\text{\boldmath${\mathsf{z}}$}}
\newcommand{\bssw}{\text{\boldmath${\mathsf{w}}$}}

\newcommand{\sss}{{\mathsf{s}}}
\newcommand{\ssx}{\mathsf{x}}
\newcommand{\ssy}{{\mathsf{y}}}
\newcommand{\ssz}{{\mathsf{z}}}
\newcommand{\ssw}{{\mathsf{w}}}

\def\BMA{\textup{{\tt BMA}}\xspace}

\def\Align{\textup{{\tt Align}}\xspace}
\def\SBMA{\BMA^\ast}
\def\MAIN{\textup{{\tt Reconstruct}}\xspace}
\def\AR{\MAIN}

\newcommand{\current}{\textnormal{current}\xspace}

\def\pos{\textup{\textnormal{last}}\xspace}

\usepackage{dsfont}

\renewcommand{\N}{\mathds{N}}

\def\tsx{\tilde{\mathsf{x}}} \def\sx{\mathsf{x}} \def\sw{\mathsf{w}}
\def\dist{\textnormal{dist}} \def\bsw{\boldsymbol{\mathsf{w}}}
\def\idist{\textnormal{dist}_\textnormal{\textsf{ideal}}} \def\ipos{\pos_\textnormal{\textsf{ideal}}}
\def\icurr{\textnormal{current}_\textsf{ideal}}

\begin{document}

\title{Near-Optimal Average-Case Approximate Trace Reconstruction\\ from Few Traces}

\author{
Xi Chen\thanks{Supported by NSF grants CCF-1703925 and IIS-1838154.} 
\\
Columbia University\\
xichen@cs.columbia.edu
\and
Anindya De\thanks{Supported by NSF grants CCF-1926872 and CCF-1910534.}
\\
University of Pennsylvania\\
anindyad@cis.upenn.edu
\and
Chin Ho Lee\thanks{Supported by the Croucher Foundation and the Simons Collaboration on Algorithms and Geometry.}
\\
Columbia University\\
c.h.lee@columbia.edu
\and
Rocco A. Servedio\thanks{Supported by NSF grants CCF-1814873, IIS-1838154, CCF-1563155, and by the Simons Collaboration on Algorithms and Geometry.} 
\\
Columbia University\\
rocco@cs.columbia.edu
\and
Sandip Sinha\thanks{Supported by NSF grants  CCF-1714818, CCF-1822809, IIS-1838154, CCF-1617955, CCF-1740833, and by the Simons Collaboration on Algorithms and Geometry.}
\\
Columbia University\\
sandip@cs.columbia.edu
}

\maketitle

\thispagestyle{empty}


\begin{abstract} 

In the standard trace reconstruction problem, the goal is to \emph{exactly} reconstruct an unknown source string $\ssx \in \zo^n$ from independent ``traces'', which are copies of $\ssx$ that have been corrupted by a $\delta$-deletion channel which independently deletes each bit of $\ssx$ with probability $\delta$ and concatenates the surviving bits. We study the \emph{approximate} trace reconstruction problem, in which the goal is only to obtain a high-accuracy approximation of $\ssx$ rather than an exact reconstruction. 

We give an efficient algorithm, and a near-matching lower bound, for approximate reconstruction of a random source string $\ssx \in \zo^n$ from few traces.  
Our main algorithmic result is a polynomial-time algorithm with the following property: for any deletion rate $0 < \delta < 1$ (which may depend on $n$), for almost every source string $\ssx \in \zo^n$, given any number $M \leq \Theta(1/\delta)$ of traces from $\Del_\delta(\ssx)$, the algorithm constructs a hypothesis string $\widehat{\ssx}$ that has edit distance at most $n \cdot (\delta M)^{\Omega(M)}$ from $\ssx$. 
We also prove a near-matching information-theoretic lower bound showing that given $M \leq \Theta(1/\delta)$ traces from $\Del_\delta(\ssx)$ for a random $n$-bit string $\ssx$, the smallest possible expected edit distance that any algorithm can achieve, regardless of its running time, is $n \cdot (\delta M)^{O(M)}$. 
\end{abstract}

\newpage

\setcounter{page}{1}


\section{Introduction} \label{sec:intro}

\subsection{Background and prior work}

In the \emph{trace reconstruction problem} \cite{Kalashnik73,Lev01a,Lev01b,BKKM04}, there is an unknown $n$-bit source string $\ssx \in \zo^n$, and a reconstruction algorithm that has access to independent \emph{traces} of $\ssx$, where a trace of $\ssx$ is a draw from $\Del_\delta(\ssx)$.  Here $\Del_\delta(\cdot)$ is the \emph{deletion channel}, which independently deletes each bit of $\ssx$ with probability $\delta$ and outputs the concatenation of the surviving bits. The goal of the reconstruction algorithm is to correctly reconstruct the source string $\ssx$.

\subsubsection{Exact trace reconstruction}

Much research effort has been dedicated to different aspects of the trace reconstruction problem in recent years \cite{MPV14,DOS17,NazarovPeres17,PeresZhai17,HPP18,HHP18,BCFSSfocs19,BCSSrandom19,Chase19,KMMP19,HPPZ20,Chase20,NarayananRen20,CDLSS20smoove,CDLSS20lowdeletion}.  In the ``worst-case'' version of trace reconstruction, the source string $\ssx$ may be an arbitrary $n$-bit string.  This is a challenging problem, with the best known information theoretic lower bound on the number of traces required for trace reconstruction (for constant deletion rates $\delta$) being $\tilde{\Omega}(n^{3/2})$ traces \cite{Chase19} and the best known information theoretic upper bound being $\exp(\tilde{O}(n^{1/5}))$ traces \cite{Chase20} (improving on earlier $\exp(\tilde{O}(n^{1/2}))$ and $\exp(\tilde{O}(n^{1/3}))$-time and sample algorithms due to \cite{HMPW08} and \cite{NazarovPeres17,DOS17} respectively). In the subconstant deletion rate regime, a $\poly(n)$-time and sample algorithm for worst-case source strings was recently given in \cite{CDLSS20lowdeletion} for deletion rate $\delta = O(1/n^{1/3 + \eps})$, improving on an earlier result for deletion rate $\delta = O(1/n^{1/2 + \eps})$ \cite{BKKM04}.
Turning to the ``average-case'' variant of trace reconstruction, the goal is to give algorithms (and lower bounds) that hold for most possible source strings $\ssx$ (equivalently, hold with high probability for a uniform random source string $\bssx \in \zo^n$). For the average-case problem, at constant deletion rate $\delta$ the current best known lower bound is $\tilde{\Omega}((\log n)^{5/2})$ traces \cite{Chase19} and the best known upper bound is $\exp(O(\log n)^{1/3})$ traces \cite{HPP18,HPPZ20}.  In \cite{BKKM04} an $O(\log n)$-trace, $\poly(n)$-time algorithm is given for the average case problem when the deletion rate is $\delta = O(1/\log n).$

\medskip

\subsubsection{Approximate trace reconstruction}

\noindent {\bf Motivation.}
In this paper we study a relaxation of the exact trace reconstruction problem in which the goal is only to obtain an \emph{approximation} of the unknown source string $\ssx$. {Of course this immediately raises the question of what distance measure to use; throughout this paper we use edit distance as our distance measure between strings. 
We remark that edit distance is a natural metric to consider in the context of trace reconstruction: in particular, trace reconstruction is motivated by problems such as ancestral DNA reconstruction where the natural corruption process  includes synchronization errors such as insertion and deletion. Indeed, edit distance is the distance measure used in all of the works discussed below under ``Prior work.''} 

The study of approximate trace reconstruction has several natural motivations; first, in some applications a high-accuracy reconstruction of $\ssx$ may be all that is required rather than exact reconstruction. Second, it is of interest to obtain algorithmic results for trace reconstruction in settings where insufficiently many traces are available for exact reconstruction (because of known lower bounds mentioned above); approximate trace reconstruction offers a potential avenue for obtaining rigorous results in such settings.  Finally, as sketched above, there is a frustrating exponential gap between the known upper and lower bounds for exact trace reconstruction in both the worst-case and average-case problem variants. Hence it is natural to wonder whether sharper bounds can be achieved for approximate versions of the problem.

\ignore{
}

\medskip

\noindent {\bf Prior work.}
Several authors have quite recently considered the approximate trace reconstruction problem and related questions.

Davies, Ra\'cz, Rashtchian, and Schiffer~\cite{DRRS20} gave several algorithms that use $\polylog(n)$ traces and achieve edit distance $n/\polylog(n)$ for certain classes of source strings defined by various run-length assumptions.   They also give other algorithms which, under stronger run-length assumptions, succeed in performing approximate reconstruction using only a single trace.  In another recent work, Srinivasavaradhan, Du, Diggavi, and Fragouli~\cite{SDDF18} proposed heuristics for approximate reconstruction based on a few traces. 


Sima and Bruck \cite{SB21} have recently studied exact trace reconstruction under an edit distance constraint.  They showed that $n^{O(k)}$ traces suffice to distinguish between two (known) worst-case $n$-bit strings that are promised to have edit distance at most $k$ from each other. In a related but incomparable result, Grigorescu, Sudan, and Zhu~\cite{GSZ20} have given lower bounds on ``mean-based'' algorithms for distinguishing between worst-case pairs of strings that have small edit distance.


Summarizing the prior results on approximate trace reconstruction, we are not aware of either algorithms or lower bounds in the previous literature that apply to typical source strings $\bssx \sim \zo^n$ (though see below for a discussion of the recent work of \cite{ChasePeres:21} that is simultaneous to ours).  It is easy to see that simply outputting a single trace gives expected edit distance $\delta n$ (for any source string), and also that given $M$ traces no algorithm can achieve expected edit distance better than $\Theta(\delta^M n)$ for random source strings (since in expectation $\delta^M n$ bits of the $n$-bit source string will have been deleted from all $M$ traces).  Other than these simple observations, to the best of our knowledge no prior results were known, either in terms of algorithms or lower bounds, for approximate trace reconstruction of random strings.  We describe our algorithms and lower bounds for this setting below.

\ignore{


}

\subsection{Our results}

\noindent {\bf Matching upper and lower bounds on approximate reconstruction of random strings from few traces.}  Our main contribution is the following algorithmic result: 

\begin{theorem} [Approximate average-case trace reconstruction algorithm] \label{thm:main}
There is a $\poly(n)$ time algorithm~\MAIN~with the following property: Let $0 < \delta < 1$, and let $\bssx$ be an unknown source string that is uniform random over $\zo^n$.  
Let $\bssy^{(1)},\dots,\bssy^{(M)}$ be $M \leq \Theta(1/\delta)$ independent traces drawn from $\Del_\delta(\bssx)$.
Then with probability at least $1-1/\poly(n)$ over $\bssx \sim \zo^n$ and $\bssy^{(1)},\dots,\bssy^{(M)} \sim \Del_\delta(\bssx)$, the output of \MAIN~on input $\delta$ and $\bssy^{(1)},\dots,\bssy^{(M)}$~is a string $\widehat{\ssx} \in \zo^\ast$ that has $\dedit(\bssx,\widehat{\ssx}) \leq  n \cdot (\delta M)^{\Omega(M)}.$
\end{theorem}

\ignore{
}

An interesting special case of \Cref{thm:main} is obtained when the number of available traces $M$ is $\Theta(1/\delta).$ In this case the \AR~algorithm achieves edit distance $n/2^{\Omega(1/\delta)}$, which is exponentially better than the benchmark of $\delta n$ edit distance that is trivially achievable using a single trace.

To complement \Cref{thm:main}, we prove an information-theoretic lower bound on approximate trace reconstruction of random strings from $M \leq \Theta(1/\delta)$ traces. This lower bound shows that the accuracy achieved by \AR~is essentially the best possible:

\begin{theorem} [Lower bound on approximate average-case trace reconstruction] \label{thm:mainlower}
Let $0 < \delta < 1$, and let $\bssx$ be an unknown source string that is uniform random over $\zo^n$.  
Let $\bssy^{(1)},\dots,\bssy^{(M)}$ be $M$ independent traces drawn from $\Del_\delta(\bssx)$, where $M \leq \Theta(1/\delta)$.  Let $\tt{A}$ be any algorithm which, given $\delta$ and $\bssy^{(1)},\dots,\bssy^{(M)}$ as input, outputs a hypothesis string $\widehat{\ssx}$ for $\bssx$.  Then the expected edit distance between
$\widehat{\ssx}$ and $\bssx$ is at least 
$n \cdot (\delta M)^{O(M)}$.
\end{theorem}

We observe that for natural parameter settings, \Cref{thm:mainlower} improves on the simple $\Omega(\delta^M n)$ expected edit distance lower bound mentioned earlier; for example, taking $M = \Theta(1/\delta)$, \Cref{thm:mainlower} proves that the best possible accuracy is $n/2^{O(M)}$ rather than $\delta^M n.$

\ignore{
}

\begin{remark}
In simultaneous and independent work to ours, Chase and Peres \cite{ChasePeres:21} have also considered the problem of approximate trace reconstruction of random source strings $\bssx$. Their main result is that for any constant deletion rate $\delta$ (bounded away from 1) and any constant $\eps$ (bounded away from 0), there is an algorithm that uses $O_{\delta,\eps}(1)$ traces and, with high probability over a random source string $\bssx \sim \zo^n$, succeeds in reconstructing a hypothesis string $\widehat{x}$ with edit distance at most $\eps n$ from $\bssx$. 

The work of \cite{ChasePeres:21} and the current paper focus on different parameter settings, in particular different regimes for the number of traces available to the algorithm, and establish complementary results.  
The results of \cite{ChasePeres:21} apply in the regime where ``many traces '' (significantly more than $\Theta(1/\delta)$) are available, and give high-accuracy reconstruction in this regime. 
In contrast, our results apply in the ``few traces'' regime where only some number $1 \leq M \leq \Theta(1/\delta)$ of traces are available, and give essentially optimal reconstruction for any such small number of traces.

\ignore{




}
\end{remark}

\subsection{Discussion and future work}
A number of directions suggest themselves for future work on approximate trace reconstruction; we close this introduction by briefly mentioning a few of these.

One natural goal is to obtain results for average-case approximate trace reconstruction which generalize both the results of the current paper and the results of \cite{ChasePeres:21}, by establishing sharp bounds on approximate average-case trace reconstruction in the regime where more than $\Theta(1/\delta)$ many traces are available.  It is clear that the $n \cdot (\delta M)^{\Theta(M)}$ form of our edit distance bound no longer holds once $M$ is $\omega(1/\delta)$; it would be interesting to understand the best achievable edit distance, as a function of $\delta$ and $M$, in this regime.

\ignore{
}

Another natural goal is to obtain algorithmic results for approximate trace reconstruction of \emph{worst-case} rather than random strings.  Here we observe that the current state of the art for worst-case exact trace reconstruction places significant limitations on how much better than edit distance $\delta n$ (trivially achievable by simply outputting a random trace) it is possible to do for approximate reconstruction of worst-case strings.  
As noted earlier, until quite recently the best result known for the low deletion rate regime was that of \cite{BKKM04}, which gave an algorithm using $O(n \log n)$ traces to reconstruct an arbitrary source string $\ssx$ at deletion rate $\delta = n^{-(1/2 + \eps)}$. 
This was recently strengthened to a $\poly(n)$-trace algorithm that reconstructs at rate $\delta = n^{-(1/3 + \eps)}$ \cite{CDLSS20lowdeletion}.
For the worst case approximate trace reconstruction problem, achieving edit distance $\delta^3 n$
for all $\delta$, even using $\poly(n)$ traces, would require improving the recently established state of the art from \cite{CDLSS20lowdeletion} for the low-deletion-rate regime of the exact reconstruction problem.


\section{Our approach} \label{sec:approach}

\subsection{Overview of our algorithmic approach (\Cref{thm:main})}

\subsubsection{Some preliminary observations and simplifications} \label{sec:simplifications}



We begin by observing that to prove \Cref{thm:main} it suffices to prove it under the assumptions that
\begin{equation} \label{eq:assumptions-initial}
  {\frac 1 {n^2}} \leq \delta < \frac 1 {KM},
\quad \quad
K^2 \leq M \leq {\frac 1 {K \delta}}, \quad \quad
\text{and} \quad \quad
(\delta M)^{M/K} \geq 1/n^2,
\end{equation} 
for a sufficiently large absolute constant $K$.
The upper bounds on $\delta$ and $M$ follow directly from our assumption $M \le \Theta(1/\delta)$ in \Cref{thm:main}.
For the lower bound on $\delta$, note that if $\delta < 1/n^2$, then with probability at least $1 - 1/n$ a single input trace will have no bits deleted and hence will trivially yield a string $\widehat{\ssx}$ that has $\dedit(\bssx,\widehat{\ssx})=0$.


For the lower bound on $M$, we observe that if $M < K^2$, then a single trace would satisfy the claimed edit distance bound in \Cref{thm:main}.
Indeed, since a single trace has edit distance from $\bssx$ distributed as $\Bin(n,\delta)$, and the probability that a draw from $\Bin(n,\delta)$ exceeds $n \cdot \delta^{0.1}$ is at most $n^{-\Omega(1)}$ (by a standard multiplicative Chernoff bound, using that $1/n^2 \leq \delta \leq 1/K$), the trivial algorithm that simply outputs a single input trace would satisfy edit distance
\[
  n \delta^{0.1} \le n (\delta M)^{0.1} \le n (\delta M)^{0.1 M/K^2} = n (\delta M)^{\Omega(M)} .
\]

For our final simplifying observation 
that $M$ and $\delta$ jointly satisfy $(\delta M)^{M/K} \geq 1/n^2$,
note that if $(\delta M)^{M/K}$ is less than $1/n^2$, then the claimed high-probability edit distance bound $n \cdot (\delta M)^{\Omega(M)}$ of \Cref{thm:main} is less than 1 (for a suitable choice of the hidden constants), and hence the claim of \Cref{thm:main} is that with high probability the edit distance achieved is zero.  In this case since $(\delta M)^{M/K}$ is decreasing for $M \in [0,\delta/e]$, we can simply
use $M' < M$ traces so that $1/n^2 \leq (\delta M')^{M'/K} < 1/n$, and achieve edit distance $n \cdot (\delta M')^{M'/K}$ which will also achieve edit distance 0 (which of course suffices to achieve the edit distance required by the theorem statement).
Therefore we will assume that the conditions given in (\ref{eq:assumptions-initial}) hold throughout the rest of our proof of \Cref{thm:main}.

 \subsubsection{The high-level approach}

Our main algorithm \AR~makes essential use of one particular distinguished trace, which we denote $\bssy^\ast$ and refer to as the \emph{reference trace}, as well as $M$ other traces $\bssy^{(1)},\dots,\bssy^{(M)}$.  
The overall \AR~algorithm works by repeatedly executing two different subroutines.  Below we first give a high level description of what each of these subroutines does and then we present the overall algorithm and explain how it uses these subroutines.

{\bf First subroutine: Alignment.} The first subroutine is an alignment procedure which we call \Align. 
It takes as input the reference trace $\bssy^\ast$ and a pointer $\ell^\ast$ to a location in the reference trace, as well as the $M$ other traces $\bssy^{(1)},\dots,\bssy^{(M)}$.  
It outputs a list of $M$ pointers $(\ell^{(1)},\dots,\ell^{(M)})$ where each pointer $\ell^{(m)}$ specifies a location in the $m$-th trace $\bssy^{(m)}.$ 
Roughly speaking, \Align~uses the reference trace to ``align'' the other $M$ traces, i.e. to come up with a pointer into each trace so that most of the pointers agree (\Align does not change the location $\ell^\ast$ of the pointer into the reference trace). In more detail but still at a high level, the main guarantee of the \Align algorithm is that with high probability, a clear majority of the $M$ pointers all point to locations that came from the same bit $\bssx_i$ of the source string $\bssx$. 
(Another important guarantee is that with high probability this location $i \in [n]$ is ``not too far'' from the location in $\bssx$ that $\bssy^\ast_{\ell^\ast}$ came from; we give more details on this below.) 
Thus a successful run of \Align~results in a clear majority of the $M+1$ pointers (including the reference trace's pointer $\ell^\ast$) all being in agreement.  We refer to a specification of the pointer locations $(\ell^{(1)},\dots,\ell^{(M)})$ as a \emph{configuration}, and we say that a configuration for which there is a clear majority in agreement as described above is \emph{in consensus.} (We give a fully detailed definition in \Cref{subsec:align}, along with a detailed statement of the \Align algorithm's performance guarantee.)
We emphasize that the correctness of this subroutine, i.e., Alignment, crucially relies on the source string $\bx$ being uniform random.

{\bf Second subroutine: Bitwise Majority.} The second subroutine is a ``Bitwise Majority Alignment'' procedure, which we call \BMA. 
This procedure was first introduced in the work of \cite{BKKM04} and was further analyzed in the recent work \cite{CDLSS20lowdeletion}. 
(As we explain below, a crucial ingredient in our proof of \Cref{thm:main} is a new refined analysis of \BMA~that goes significantly beyond the results of \cite{CDLSS20lowdeletion}.)  
All of the output bits that our algorithm constructs are produced by \BMA.
The \BMA procedure takes as input the $M+1$ traces $\bssy^\ast,\bssy^{(1)},\dots,\bssy^{(M)}$ and the corresponding pointers $\ell^\ast,\ell^{(1)},\dots,\ell^{(M)}$.
The \BMA~algorithm is run for $R := (\delta M)^{-\Theta(M)}$ many stages to reconstruct $R$ output bits; in the course of its execution it updates the pointers into all $M+1$ of the traces $\bssy^\ast,\bssy^{(1)},\dots,\bssy^{(M)}$. 

To explain the performance guarantee of the \BMA~procedure we need the notion of a \emph{$k$-desert.} Roughly speaking, a binary string $z\in \{0,1\}^*$ is said to be a \emph{$k$-desert} if (i) it is sufficiently long, and (ii) it is a prefix of $s^{\infty}$ for some $s \in \{0,1\}^{\le k}$ (we give a precise definition in \Cref{subsec:BMA}).
The main guarantee of the \BMA~procedure is that if it is run on a configuration that is in consensus at some location $i$ in the source string $x$, then with high probability it produces a $R$-bit string that agrees with $(x_i,x_{i+1},\dots)$ up to the location (if any) where a $k$-desert of length $L := \Theta(M \log(1/(M \delta)))$ first appears, for some $k \leq L/2.$ We note  that unlike the first subroutine Alignment, the guarantee of \BMA~procedure is a worst-case guarantee.

{\bf The overall \AR~algorithm.} 
As stated earlier, the overall algorithm repeatedly runs \Align, then \BMA, then \Align, then \BMA, and so on.
We present {a slightly simplified version of the algorithm in \Cref{figg:MAIN-simplified} (see \Cref{sec:combine} for the formal algorithm; the version in \Cref{figg:MAIN-simplified} differs only in that some parameter settings have been slightly simplified for the sake of readability).

\begin{figure}[t!]
  \centering
\setstretch{1.2}
  \begin{algorithm}[H]
    \caption{{\MAIN} (slightly simplified)}\label{fig:BMA}
		\DontPrintSemicolon
		\SetNoFillComment
		\KwIn{A positive integer $n$ and $(M+1)$ traces {$\ssy^*,\ssy^{(1)},\ldots,\ssy^{(M)}$ for some $M \leq 1/(K \delta)$}  
		}
		\KwOut{
		A binary string $\ssw$}
		Set $\ell^*=1$ and $\ssw=\epsilon$ {(the empty string)}\\
		\While{$\ell^*\le |\ssy^*|$}{
		  Run $\Align\hspace{0.04cm}(\ell^*,\ssy^*,\ssy^{(1)},\ldots,\ssy^{(M)})$ to obtain a 
		  tuple of locations $(\ell^{(1)},\ldots,\ell^{(M)})$\\
		  Run $\smash{\BMA\hspace{0.04cm}(\ssy^*,\ssy^{(1)},\ldots,\ssy^{(M)};\ell^*,\ell^{(1)},\ldots,\ell^{(m)})}$ to obtain 
		    a binary string ($\eps$ or in $\{0,1\}^R$)\\
		  Concatenate the string returned by $\BMA$ to the end of $\ssw$\\
   		  Set $\ell^*$ to be the final pointer of $\ssy^*$ in the run of 
		  $\BMA$ above  and increment it\\		  
		  }
\Return $\ssw$.
\end{algorithm}\caption{A slightly simplified version of our main algorithm $\MAIN$ (the actual algorithm differs in some small details and is given in \Cref{figg:MAIN}).}\label{figg:MAIN-simplified}
\end{figure}

\ignore{


\begin{figure}[!t] 
\newcommand\mycommfont[1]{\small\ttfamily\textcolor{blue}{#1}}
\SetCommentSty{mycommfont}
\setstretch{1.2}
\begin{algorithm}[H]
  \caption{{\AR}}\label{alg:align}
	\DontPrintSemicolon
	\SetNoFillComment
  \KwIn{An allowed number of traces $M \leq 1/(K\delta)$ and access to $\Del_\delta(\bx),$ where $\bx \sim \zo^n$.}
  \KwOut{A hypothesis string $\widehat{x} \in \zo^\ast.$
}

\smallskip
Draw traces $\bssy^\ast,\bssy^{(1)},\dots,\bssy^{(M)} \sim \Del_\delta(\bssx)$. (We view all of these traces as padded with infinitely many 0-bits at the end.)

Set $\ell^*=1$ and set $\widehat{x}$ to be the empty string. 

  \While(){$\ell^*\le |\bssy^*|$}{ \label{alg:main-alg-line4}

	Run $\Align(\bssy^*,\ell^*,\bssy^{(1)},\ldots,\bssy^{(M)})$ to obtain a configuration $(\ell^{(1)},\ldots,\ell^{(M)})$.

	Run $\BMA(\bssy^*,\bssy^{(1)},\ldots,\bssy^{(M)};\ell^*,\ell^{(1)},\ldots,\ell^{(M)})$ for $R := (\delta M)^{-\Theta(M)}$ rounds (so $\BMA$ always returns a string of length precisely $R$).
	 
	 Concatenate the $R$-bit string returned by $\BMA$ to the end of $\widehat{x}$. 
	 
	 Increment $\ell^\ast$ by 1.

    }

    Return $\widehat{x}$.

\end{algorithm}
\caption{The main \AR algorithm.}
\label{fig:align}
\end{figure}

END IGNORE

}


The high level intuition for why the algorithm succeeds is as follows. Each run of \Align with high probability succeeds in putting the $M$ traces in consensus at a location ``not too far'' from the location in $\bssx$ corresponding to $\bssy^\ast_{\ell^\ast}$.  Given that this consensus has been achieved by \Align, then the subsequent run of \BMA with high probability succeeds in correctly reconstructing the next $R := (\delta M)^{-\Theta(M)}$ many bits of $\bssx$.  In the course of running \BMA the pointer $\ell^\ast$ is with high probability advanced to ``approximately the right location'' corresponding to the last-reconstructed bit of $\bssx$, so the next run of \Align again establishes consensus at approximately the right location.  Thus the overall output string $w$ of the algorithm is the concatenation of many length-$R$ strings, most of which correspond to subwords of $\bssx$ from approximately the right locations. From this it can be shown that the overall reconstructed string is not too far in edit distance from $\bssx$.

The above high-level explanation sketches an idealized version of the actual scenario and glosses over a number of technical difficulties. In more detail, there are many sources of error from different possible failure events and a careful analysis is required (and is provided in \Cref{sec:combine}) to keep the failure probabilities from all of these under control and not ``give away too much'' in the overall edit distance.  The issues that must be handled include the following:
\begin{flushleft}\begin{itemize}
\item \Align~may fail to align the traces to a consensus location, or may misalign the traces and achieve consensus at a  location that is far away from the location in $\bssx$ corresponding to $\bssy_{\ell^\ast}^\ast$.  Our analysis shows that this happens with small (but non-negligible probability), and bounds the cumulative error (edit distance) incurred by the runs of \Align~for which this happens.

\item Even when \Align aligns the traces at a location that is ``not too far'' from the correct location, the alignment location in $\bssx$ may not be exactly the location in $\bssx$ corresponding to $\bssy_{\ell^\ast}^\ast$. This contributes to the overall edit distance between $\widehat{x}$ and $\bssx$ even when the \Align algorithm succeeds.

\item On average the random string $\bssx$ will have a $k$-desert of length $L$ occuring roughly once every $2^{\Theta(L)}$ positions.  When a run of \BMA encounters such a location the resulting $R$-bit string that it produces may be badly off from the true corresponding portion of $\bssx$. This contributes to the overall edit distance between $\widehat{x}$ and $\bssx$ even when the algorithm succeeds.

\item Even when the portion of $\bssx$ that a given run of \BMA is operating on does not contain a $k$-desert of length $L$, the \BMA algorithm may fail to correctly reconstruct the relevant portion of \BMA with small (but non-negligible) probability. Our analysis bounds the overall error in the reconstructed string that comes from such ``failed runs'' of \BMA.  
\end{itemize}\end{flushleft}

\subsection{The \Align~procedure} \label{subsec:align}

As stated earlier, \Align takes as input the reference trace $\bssy^\ast$, a pointer $\ell^\ast$ to a location in the reference trace, and the $M$ other traces $\bssy^{(1)},\dots,\bssy^{(M)}$, and outputs a list of $M$ pointers $\ell^{(1)},\dots,\ell^{(M)}$ into the $M$ traces $\bssy^{(1)}, \dots, \bssy^{(M)}.$ 
In a successful run of \Align, it generates a list of pointers most of which point to locations that came from the same bit $\bssx_i$ of the source string $\bssx$.  

At a very high level, \Align~works in two stages.  The first stage, which performs an ``approximate alignment,'' consists of a sequence of iterative refinement steps; in each successive step of this stage, for each trace $\bssy^{(m)}$ \Align tries to identify successively smaller subwords of $\bssy^{(m)}$ that fairly closely match (as measured by edit distance) suitable successively smaller subwords, centered at $\ell^\ast$, of the reference trace $\bssy^\ast$.   At the end of a successful execution of the first stage, for each trace $\bssy^{(m)}$ a relatively small subword $\bssy^{(m)}_{Q^m_1}$ has been identified which contains the ``right location'' in $\bssy^{(m)}$ (informally, corresponding to the portion of $\bssx$ that $\bssy^\ast_{\ell^\ast}$ came from). In the second stage, \Align searches for a suitable subword that appears in at least 95\% of 
$\bssy^{(1)}_{Q^1_1},\dots,\bssy^{(M)}_{Q^M_1}$ and uses the location of this subword in each $\bssy^{(m)}$ to determine the exact final pointer location $\ell^{(m)}.$

Correctness of the second stage (given that a successful ``approximate alignment'' was indeed achieved in the first stage) is established using an elementary but careful analysis that we do not describe here but is given in \Cref{sec:second-stage}.
To gain intuition for the iterative approach employed in the first stage, it is useful to consider the following toy scenario: Fix an $a$-bit subword $\ssw$ of the reference trace $\bssy^\ast$ that is centered at location $\ell^\ast$. Intuitively, the deletion rate $\delta$ is  relatively low, so the subword $\ssw$ of $\bssy^{\ast}$ should have small edit distance from the corresponding subword of the source string $\bssx$, and, transitively, should also have small edit distance from the corresponding subwords of each of the $M$ traces $\bssy^{(1)},\dots,\bssy^{(M)}$. However, since $\bssx$ is uniform random (and hence each trace $\bssy^{(m)}$ is also uniform random), if $a$ is a ``small'' value that is $\ll \log n$, then it is very likely that $\ssw$ will occur as an $\ell$-bit subword of each $\bssy^{(m)}$ in many locations, and thus a simpleminded approach of just scanning all of $\bssy^{(m)}$ to try to find $\ssw$ (or a close match to it) will not succeed in uniquely identifying the correct location.  But if $a$ is a ``large'' value (actually, being just modestly larger than $\log n$ will do), then it is very likely that only one location in each $\bssy^{(m)}$ will be a close match to the $a$-bit string $\ssw$.  This reduces the problem of finding the right location in the $\approx n$-bit string $\bssy^{(m)}$ to the problem of finding the right location in the $\approx a$-bit subword of $\bssy^{(m)}$ that was just identified (by virtue of closely matching $\ssw$); and now we can iterate.

A more complete overview and explanation of \Align is given in \Cref{sec:align-overview}.  \Cref{sec:align} gives a  detailed proof of \Cref{thm:align-restated}, which is our main result about \Align; since the exact theorem statement is somewhat cumbersome (involving various specific parameter settings), we give an informal version here and defer the fully detailed statement to \Cref{sec:align}.
Informally, we say that the $\Align$ algorithm
  \emph{succeeds} on source string $\ssx$ with respect to a tuple of traces $ (\ssy^*,\ssy^{(1)},\ldots, \ssy^{(M)} )$ if 
  the following condition holds for ``almost all'' locations $\ell^*\in [|\ssy^*|]$: The output $(\ell^{(1)},\ldots,\ell^{(M)})$
  of $\Align(\ell^*,\ssy^*,\ssy^{(1)},\ldots,\ssy^{(M)})$ satisfies
(1) At least 90\% of $\source^{(m)}(\ell^{(m)})$, $m \in [M]$, agree on the same location $i \in [n]$, and 
(2) The consensus location $i$ is ``quite close'' to the location in $\bssx$ that $\bssy^\ast_{\ell^\ast}$ came from. 
Now we can state an informal version of \Cref{thm:align-restated},  which gives a performance guarantee on $\Align$:



\begin{theorem}[Main result about \Align, informal statement] \label{thm:align}
Let $\bssx\sim\{0,1\}^n$ and  let
  $\bssy^*,\bssy^{(1)},\ldots,\bssy^{(M)}$ be independent traces drawn from $\Del_\delta(\bssx)$, where $\delta$ and $M$ satisfy \Cref{eq:assumptions-initial}.
Then $\Align$ succeeds on $\bssx$ with respect to $(\bssy^*,\bssy^{(1)},\ldots,\bssy^{(M)})$
  with probability at least $1-1/\poly(n)$.
\end{theorem}


\subsection{The \BMA~procedure} \label{subsec:BMA}

The Bitwise Majority Alignment, or \BMA, procedure, operates in discrete time steps on a collection of $M$ independent traces. At each time step it outputs one bit of the hypothesis string that it is reconstructing. Throughout its execution, at each time step $t$, for each $m \in [M]$ the \BMA algorithm maintains a pointer\ignore{$\curr^{(m)}(t)$} into the $m$-th trace.  The idea of \BMA is that at each time step $t$, it should be the case that most of the pointers are correctly aligned, i.e.~the majority of the bits that they point to in their respective traces came from the same bit of the source string $\ssx$.  In the $t$-th time step the majority vote of the $M$ bits that are pointed to in the traces is the output bit \BMA~produces, and

\begin{flushleft}\begin{itemize}
\item For each trace in which the pointer points to a bit that agrees with the majority, the pointer is incremented by one location;

\item For each trace in which the pointer points to a bit that disagrees with the majority, the pointer stays in the same location.
\end{itemize}\end{flushleft}

A first analysis of \BMA, for deletion rate $\delta$ slightly less than $n^{-1/2}$, was originally given in \cite{BKKM04}, and more recently an analysis for deletion rate $\delta$ slightly less than $n^{-1/3}$ was given in \cite{CDLSS20lowdeletion}.  We give a significant extension of \cite{CDLSS20lowdeletion} by providing a much more refined analysis which yields a considerably stronger quantitative result.\footnote{This quantitative strengthening plays an essential role in our being able to obtain tight bounds  (recall the essentially matching \Cref{thm:main,thm:mainlower}) via our approach.}  In more detail, in the current work our analysis of \BMA handles deletion rates even as large as a (small) absolute constant independent of $n$, and indeed handling such deletion rates is essential for our overall results.

To state our main theorem about \BMA we require the following terminology:
Recall that a string is said to be a 
\emph{$k$-desert} for some $k\ge 1$ if it is the prefix of $\sss^{\infty}$ for some string $\sss\in \{0,1\}^k$.  We say a string is a \emph{long desert} if it is a $k$-desert of length $L$ for some $k\le L/2$. 

Our main result about \BMA says, roughly speaking, that if the source string does not contain any long desert then with high probability \BMA succeeds in exactly reconstructing the source string, and moreover does so with a ``clear majority'' in each round. Similar to \Align, the detailed theorem statement about \BMA involves various specific parameter settings, so we defer its exact statement until later (see \Cref{thm:bma-easy}) and here give an informal statement:

\begin{theorem} [Main result about \BMA, informal statement] \label{thm:bma}
Let $\tsx\in \{0,1\}^R$ be a string that does not contain any
  long desert.
Let $\bssz^{(1)},\ldots,\bssz^{(M)}\sim \Del_\delta(\tsx)$ be independent traces. For suitable settings of $R, M, k, L$ and $\delta$, with high probability 
  \BMA returns exactly $\tsx$, and~in every round $t\in [R]$ the majority is
  reached by at least $90\%$ of $\bssz^{(1)},\ldots,\bssz^{(M)}$. 

\end{theorem}


\subsection{Overview of our lower bound approach (\Cref{thm:mainlower})}

Our lower bound approach is informed by insights arising from the analysis of our algorithm.  Given the arguments sketched above for our algorithmic results, it is natural to pursue a lower bound based on the difficulty of reconstructing deserts. The high-level idea of our lower bound is that having access to only a limited number $M \leq \Theta(1/\delta)$ of traces imposes strong limitations on the ability of any reconstruction algorithm to accurately estimate the lengths of deserts, and this inability to accurately reconstruct deserts translates into an inability to perform overall high-accuracy approximate reconstruction. Guided by this general idea, it is natural to consider 1-deserts (runs of all 0's or all 1's) as potential sources of hardness, and indeed this is our approach.

In more detail, our lower bound proceeds in four conceptual stages.
\begin{flushleft}\begin{enumerate}
\item We first (\Cref{sec:atomic}) consider the following simple distribution distinguishing problem: an algorithm is
given $M$ draws which are guaranteed to come from one of two product distributions over $\N \times \N$: (a) the product distribution $\Bin(M,1-\delta) \times \Bin(M+1,1-\delta)$, or (b) the product distribution
$\Bin(M+1,1-\delta) \times \Bin(M,1-\delta)$, where both (a) and (b) are equally likely to be the target product distribution. We show that any algorithm for determining whether it is (a) or (b) must have failure probability at least $(\delta M)^{O(M)}.$

\item Next, in \Cref{sec:PRLP} we consider the algorithmic task of solving $B$ independent instances 
of the distinguishing problem described in (1) above; this may be viewed as the problem of inferring an unknown $B$-bit string that is uniform over $\{0,1\}^B$ given certain partial/noisy information about the string. Building on (1) above, we show that the expected edit distance from the output of any algorithm for this problem to the unknown uniform string in $\{0,1\}^B$ will be at least $B  (\delta M)^{O(M)}$.

\item We then (in \Cref{sec:embedding})  observe that a random string $\bssx$ can be viewed as containing, with high probability, $B = n / 2^{\Theta(M)}$ independent instances 
of the distribution distinguishing problem from (1).  Roughly speaking, this is because a random string $\bssx \sim \zo^n$ can be viewed as composed of $n/(2M+4)$ blocks of $2M+4$ bits each, and with high probability $\Theta(n/2^{2M+4})$ of these blocks will consist of either the string $\alpha = 0^M 1 0^{M+1} 11$ 
or the string $\beta = 0^{M+1} 1 0^{M} 11$, 
and these two strings are equally likely for each block. (The specific structure of these $\alpha$ and $\beta$ strings is chosen to ensure that they cannot overlap; this is useful for (4) below.)

\item Using (3), in \Cref{sec:embedding} we show that any algorithm that achieves a certain $n (\delta M)^{\Omega(M)}$ 
expected edit distance for reconstructing a random string from $M$ traces can be used to give an algorithm that solves $B = n / 2^{\Theta(M)}$ independent copies of the distinguishing problem described in (1) with 
an expected edit distance that is lower 
than can possibly be obtained, contradicting the lower bound from item (2) above. Establishing this reduction is the most intricate part of our lower bound.
\end{enumerate}\end{flushleft}

\subsection{Organization}

In \Cref{sec:preliminaries} we set up some preliminaries.
In \Cref{sec:align} we prove \Cref{thm:align}, our main result about the \Align algorithm.
In \Cref{sec:bma} we prove \Cref{thm:bma}, our main result about \BMA.
In \Cref{sec:combine} we use \Cref{thm:align,thm:bma} to prove \Cref{thm:main}.
Finally, \Cref{sec:lowerbound} proves our lower bound, \Cref{thm:mainlower}.


\section{Preliminaries} \label{sec:preliminaries}

\noindent {\bf Notation.}
Given a positive integer $n$, we write $[n]$ to denote $\{1,\ldots,n\}$.
Given two integers $a\le b$ we write $[a:b]$ to denote $\{a,\ldots,b\}$. 
We write $\ln$ to denote natural logarithm and $\log$ to denote logarithm to the base 2.
We denote the set of non-negative integers by $\Z_{\geq 0}$.
We  write ``$a=b\pm c$'' to indicate that $b-c\le a\le b+c$.

\medskip

\noindent {\bf Subwords.} 
It will be convenient for us to index a binary string~$\ssx \in \zo^n$
  using $[1:n]$ as~$\ssx=(\ssx_1,\dots,\ssx_{n})$. 
Given such a string $\ssx\in \{0,1\}^n$ and integers $1\leq i\le  j \leq n$, we write $\ssx_{[i:j]}$ to denote the \emph{subword} $(\ssx_i, \ssx_{i+1}, \dots, \ssx_j)$ of $\ssx$.
An  \emph{$\ell$-subword}   of $\ssx$ is a 
  subword of $\ssx$ of length $\ell$, given by $(\ssx_i, \ssx_{i+1}, \dots, \ssx_{i+\ell-1})$ 
  for some $i \in [1: n-\ell+1]$. 
 
\medskip

\noindent {\bf Distributions.}
When we use bold font such as $\bD, \bssy, \bz$, etc., it indicates that the entity in question is a random variable.
We write ``$\br \sim {\cal P}$'' to indicate that random variable~$\br$~is 
  distributed according to probability distribution ${\cal P}$.  If $S$ is a finite set we write ``$\br \sim S$'' to indicate that $\br$ is distributed uniformly over $S$.
 
\medskip

\noindent {\bf Deletion channel and traces.}
Throughout this paper the parameter $0 <\delta < 1$ denotes~the \emph{deletion probability}.  Given a string $\ssx \in \zo^n$, we write $\Del_\delta(\ssx)$ to denote the distribution of the string that results from passing  $\ssx$ through the $\delta$-deletion channel (so the distribution $\Del_\delta(\ssx)$ is supported on $\zo^{\leq n}$), and we refer to a string in the support of $\Del_\delta(\ssx)$ as a \emph{trace} of $\ssx$.  Recall that a random trace $\bssy \sim \Del_\delta(\ssx)$ is obtained by independently deleting each bit of $\ssx$ with probability $\delta$ and concatenating the surviving bits.\hspace{0.05cm}\footnote{For simplicity in this work we assume that the deletion probability $\delta$ is known to the reconstruction algorithm.  We~note that it is possible to obtain a high-accuracy estimate of $\delta$ simply by measuring the average length of traces received from the deletion channel.\ignore{ \red{Rocco: Do we need this? Maybe yes to get from \Cref{thm:main} to \Cref{thm:main2}?}}}

When a trace $\bssy$ is drawn from $\Del_\delta(\ssx)$ 
  we write $\bD$ to denote the set of \emph{locations 
  deleted} when $\ssx$ goes through the deletion channel, i.e., 
  $\bD$ is obtained by including each
  element of $[n]$ independently with probability $\delta$, and
  $\bssy$ is set to be $\smash{\ssx_{[n]\setminus \bD}}$.  (When the trace is denoted $\bssy^\ast$ or $\bssy^{(m)}$ we use $\bD^\ast$ or $\bD^{(m)}$ to denote the set of locations deleted.)

As discussed earlier, our algorithm uses a special reference trace $\bssy^\ast$ and $M$ additional traces $\bssy^{(m)}$, $m \in [M]$, and maintains pointers into each of these traces. We write $\ell^\ast$ to denote the pointer into $\bssy^{\ast}$ and $\ell^{(m)}$ to denote the pointer into $\bssy^{(m)}$ for $m \in [M]$.

\medskip

\noindent {\bf Edit distance and matchings.}
{It will be convenient for us to define the edit distance between two strings $\ssx,\ssx' \in \zo^\ast$ as
\[
\dedit(\ssx,\ssx') := |\ssx| + |\ssx'| - 2 \cdot |\mathrm{LCS}(\ssx,\ssx')|,
\]
where $|\mathrm{LCS}(\ssx,\ssx')|$ is the length of the longest common subsequence of $\ssx$ and $\ssx'$.  This is equivalent to viewing insertions and deletions of characters as being the only allowable ``atomic edits'' that can be used to transform $\ssx$ to $\ssx'$, and is easily seen to be equivalent to the standard definition (in which substitutions are also allowed) up to at most a factor of 2, since a substitution can be simulated by a deletion followed by an insertion.

A \emph{matching} $\mu$ between two strings $\ssx,\ssx' \in \zo^\ast$ is a list of pairs $(i_1,j_1), (i_2,j_2),\dots$ such that $i_1 \leq i_2 \leq \cdots$, $j_1 \leq j_2 \leq \cdots$, and for every $t$ we have $\ssx_{i_t} = \ssx'_{j_t}.$ The \emph{size} of a matching is the number of pairs.  We note that the largest matching between $\ssx$ and $\ssx'$ is of length 
$|\mathrm{LCS}(\ssx,\ssx')|$.

For two intervals $A=[a_1,a_2]$ and $B=[b_1,b_2]$ of equal length, we write ``$\mu(A)=B$'' to indicate that for every element $a_1+j \in A$, the pair $(a_1+j,b_1+j)$ is in the matching (note that this implies that the subwords $\ssx_{A}$ and $\ssx'_{B}$ are identical).
}
\ignore{
}

\medskip

\noindent {\bf Some notational conventions.}  To aid the reader we adopt the following conventions:
\begin{itemize}
\item {\bf Locations in strings of different types:} The letters $i,j$ are reserved for locations in the source string $\ssx$, so these variables refer to integers in the range $[1:n].$ We use capital letters $I,J$ to denote intervals of such locations.  The letters $p,q$ are reserved for locations in traces, so if $p$ is a location in a particular trace $\ssy$ then it refers to an integer in the range $[1:|\ssy|].$ We use capital letters $P,Q$ to denote intervals of such locations.  The letters $a,b$ are reserved for locations in other incidental strings that arise in our analysis, and intervals of such locations are denoted $A,B.$ 

\item {\bf Indexing multiple strings:} On a number of occasions we deal with collections of multiple strings (such as our $M$ traces). We index such collections with parenthesized superscripts, so for example our $M$ traces are denoted $\bssy^{(1)},\bssy^{(2)},\dots,\bssy^{(M)}$.

\item {\bf Correspondence between traces and source string $\bssx$.} Given a location $q \in |\bssy^{\ast}|$ in the reference trace $\bssy^\ast$, we write $\source^\ast(q)$ to denote the location $i \in [n]$ such that bit $\bssx_i$ gave rise to $\bssy^\ast_q$. For $m \in [M]$ we similarly write $\source^{(m)}(q)$ to denote the location $i \in [n]$ such that bit $\bssx_i$ gave rise to $\bssy^{(m)}_q$ in the trace $\bssy^{(m)}$.  
For an interval $Q=[q_1:q_2]$ of locations in $\ssy^{(m)}$, we write $\source^{(m)}(Q)$ to denote the set $\{\source^{(m)}(q): q\in [q_1:q_2]\}.$
We define $\overline{\source^{(m)}(Q)}$ to be the 
  interval $[\source^{(m)}(a):\source^{(m)}(b)]\subseteq [n]$.

 \ignore{\rnote{Should we define $\source^{(m)}(Q)$ this way or as a continuous interval, i.e we could define it as follows: for an interval $Q=[q_1,q_2] \subseteq [|\bssy^{(m)}|]$, we write $\source^{(m)}(Q)$ to denote the interval $[\source^{(m)}(q_1):\source^{(m)}(q_2)]$. Either is fine with me.}}

Given a location $i \in [n]$, if $i \notin {\bD^\ast}$ then $\image^\ast(i)$ denotes the element of $[|\bssy^\ast|]$ that $\bssx_i$ lands in (and if $i \in {\bD^\ast}$ then we define $\image^\ast(i)$ to be $\bot$). The notation $\image^{(m)}(i)$ is defined similarly with respect to trace $\bssy^{(m)}$, $m \in [M].$
We observe that if $I \subseteq \overline{\source^{(m)}(Q)}$ then $\image^{(m)}(I) \subseteq Q$.

\item {\bf Notation for bitstrings.} To help the reader differentiate between bits and the locations of bits in bitstrings, we use $\mathsf{sans \ serif \ font}$ to denote ``bit-valued objects.'' Hence the uniform random source string in $\zo^n$ is $\bssx$, the traces are $\bssy^\ast, \bssy^{(1)}$, etc., a generic fixed word in $\zo^\ast$ which is not a random variable would be denoted $\ssw$, a generic word in $\zo^\ast$ which is a random variable would be denoted $\bssw$, and so on.
\end{itemize} 
Finally we introduce some useful terminology:  We refer to a tuple of pointers $(\ell^{(1)},\dots,\ell^{(M)})$ into traces $\bssy^{(1)},\dots,\bssy^{(M)}$ (so each $\ell^{(m)}$ belongs to $[|\bssy^{(m)}|]$) as a \emph{configuration}.  We say the configuration $(\ell^{(1)},\dots,\ell^{(M)})$ is \emph{in consensus} if at least $0.9M$ of the values $m \in [M]$ all have $\source^{(m)}(\ell^{(m)})$ equal to the same location $i \in [n]$.

\subsection{Useful results}
We recall McDiarmid's basic ``method of bounded differences'' inequality, which we will use repeatedly in our analysis: 

\begin{theorem} [Theorem~3.1 of \cite{McDiarmid1998}] \label{thm:Mcdiarmid}
Let $\bX = (\bX_1,\dots,\bX_n)$ be a family of independent random variables where each $\bX_k$ takes values in a set $A_k$. Suppose that $f: A_1 \times \cdots \times A_n \to \R$ satisfies
\[
|f(x)-f(x')| \leq c_k
\]
whenever the vectors $x$ and $x'$ differ only in the $k$-th coordinate.  Then for any $t \geq 0$, writing $\mu$ for $\E[f(\bX)]$, we have
\[
\Pr[f(\bX) - \mu \geq t] \leq \exp\left(-2t^2 / \sum_{k=1}^n c_k^2\right).
\]
\end{theorem}

We use standard notation for the binary entropy function $H(p) = -p \log p -(1-p) \log(1-p)$, and we recall the standard upper bound on binomial coefficients in terms of this function, namely that  ${n \choose pn} \leq 2^{n H(p)}$ for any $0 < p < 1$.


\def\sz{\mathsf{z}}

\def\uth{{\bigskip \bigskip {\blue {\huge {\bf UP TO HERE}} }\bigskip \bigskip}}

\section{The \Align~algorithm and proof of \Cref{thm:align}} \label{sec:align}

\ignore{OLD BOUNDS:



END OLD BOUNDS}

Recall from \Cref{sec:simplifications} that the two parameters $\delta$ and $M$ satisfy
\begin{equation} 
\label{eq:assumptions}
{\frac 1 {n^2}} \leq \delta < {\frac 1 {KM}} < {\frac 1 K}
\quad \text{and} \quad
K^2 \leq M \leq {\frac 1 {K \delta}}, 
\end{equation}
where $K$ is some sufficiently large absolute constant. Let
\begin{equation} \label{eq:H}
H:=\frac{M}{K}\log \left(\frac{1}{\delta M}\right) \le 2 \log n,
\end{equation}
where the inequality is by \Cref{eq:assumptions-initial}. We observe that \Cref{eq:assumptions} also gives that $H \geq K$, and that $2^{-\Omega(H)}=(\delta M)^{\Omega(M)}$. Let
\begin{equation} \label{eq:gamma-tau}
\gamma=0.01, \quad \tau=5/\gamma = 500
\end{equation} be two constants that will be used in this section.

In this section we describe the (deterministic) \Align algorithm and prove \Cref{thm:align}
  about its performance.
Let $\ssx\in \{0,1\}^n$ be the source string and $\ssy^*,\ssy^{(1)},\ldots,\ssy^{(M)}$ be  traces 
  of $\ssx$ obtained with corresponding deletion sets $D^*,D^{(1)},\ldots,D^{(M)}\subseteq [n]$, respectively.
The algorithm $\Align$ takes $\ell^*,\ssy^*,\ssy^{(1)},\ldots,\ssy^{(M)}$ as inputs,
  where
  $\ell^*\in  [5\tau \log n:|\ssy^*|- 5\tau \log n]$, and
  returns a tuple of locations $(\ell^{(1)},\ldots,\ell^{(M)})$.


The following terminology will be useful: we say that the $\Align$ algorithm
  \emph{succeeds} on source string $\ssx$ with respect to a tuple of traces $ (\ssy^*,\ssy^{(1)},\ldots, \ssy^{(M)} )$ if 
  the following condition holds for all except at most $2^{-0.1 H} n$ many
  locations $\ell^*\in [5 \tau\log n : |\ssy^*|- 5\tau\log n ]$: The output $(\ell^{(1)},\ldots,\ell^{(M)})$
  of $\Align(\ell^*,\ssy^*,\ssy^{(1)},\ldots,\ssy^{(M)})$ satisfies
\begin{enumerate}
\item The configuration $(\ell^{(1)},\dots,\ell^{(M)})$ is in consensus, i.e.~at least 90\% of $\source^{(m)}(\ell^{(m)})$, $m \in [M]$, agree on the same location $i \in [n]$, and \vspace{-0.1 cm}
\item The consensus location $i$ satisfies 
\begin{equation}\label{consensus}
  \source^*(\ell^*) - 2H\le i \leq \source^*(\ell^*) .
\end{equation}
\end{enumerate} 

Now we can state the main result of this section which gives a performance guarantee on $\Align$:



\begin{theorem}[\Cref{thm:align}, detailed statement] \label{thm:align-restated}
Let $\bssx\sim\{0,1\}^n$ and  let
  $\bssy^*,\bssy^{(1)},\ldots,\bssy^{(M)}\sim \Del_\delta(\bssx)$ independently, where $\delta$ and $M$ satisfy \Cref{eq:assumptions,eq:H}.
$\Align$ succeeds on $\bssx$ with respect to $(\bssy^*,\bssy^{(1)},\ldots,\bssy^{(M)})$
  with probability at least $1-1/n^2$. 
\end{theorem}

\ignore{START IGNORE
\subsection{Put in Preliminaries}

For an interval $Q = [a:b]\subseteq [|\ssy|]$ of positions in a trace $\ssy$
  of a source string $\ssx$, recall $$\source(Q) = \big\{\source(q): q\in [a:b]\big\}.$$ \red{We define $\overline{\source(Q)}$ to be the 
  interval $[\source(a):\source(b)]\subseteq [n]$.}
  
\begin{claim} \label{claim:set}
   If $I \subseteq \overline{\source(Q)}$ then $\image(I) \subseteq Q$.
\end{claim}
\begin{proof}
The definition of $\overline{\source(Q)}$ was changed. I think we don't need a proof now.
\end{proof}

END IGNORE
}

\subsection{Overview} \label{sec:align-overview}

  

We first give a high-level overview of \Align.
Let $\ell^\ast$ be a location in the special reference trace $\ssy^\ast$ that 
  is not too close to the left and right ends of $\ssy^*$.
Let $Q_S^*\supset \cdots \supset Q_1^*$ be a sequence of nested
  intervals (of locations of $\ssy^*$) centered at $\ell^*$,
  with $|Q^\ast_s| =t_s$ for each $s\in [S]$, 
  \[
  t_1=2H+1, 
  \quad
  t_{s+1}=3t_s,
  \quad 
  \text{and~}t_S= \Theta(\log n)
  \]
   (note that hence $S = O(\log \log n)$).
Let $\smash{\ssw_s^*=\ssy^*_{Q^*_s}}$ for each $s\in [S]$.

The \Align algorithm consists of two stages. 
In this subsection we give some intuition behind each stage~and its analysis. 
In the intuitive  discussion below, we focus chiefly on understanding the probability that $\Align$ succeeds at 
  a particular location $\ell^*$;
  in the formal proof we need to apply the bounded difference inequality of McDiarmid 
  to argue that $\Align$ succeeds on all but except $2^{-\Omega(H)}n$ many locations with high probability.

For the rest of this section, we say that an event is
   \emph{most likely}
  to happen if it happens 
  with probability
   $1-2^{-\Omega(H)}$.

\paragraph{First stage --- locating a small neighborhood of $\source^\ast(\ell^\ast)$ in each trace
  $\ssy^{(m)}$.} \  

In the first stage, \Align works separately on each $\ssy^{(m)}$, $m\in [M]$.
It iteratively uses $\ssw^*_S,\ldots,\ssw^\ast_1$ (as templates) 
  to find a sequence of nested intervals $Q^m_S\supset \cdots \supset Q^m_1$ of locations of $\ssy^{(m)}$ such that
\begin{equation}\label{condition1}\dedit\left(\ssw^*_s,\ssy^{(m)}_{Q^m_s}\right)\le 2\gamma t_s,\quad\text{for 
  each $s=S,\ldots,1$.}\end{equation}
This is done by first finding  
$Q^m_S\subset [|\ssy^{(m)}|]$ that satisfies (\ref{condition1}) and
  then repeatedly finding 
  $Q^m_s\subset Q^m_{s+1}$ that satisfies (\ref{condition1}), for each $s=S-1,\ldots,1$.
When multiple $Q^m_s$ satisfy (\ref{condition1}), we pick
  one arbitrarily; when no interval $Q^m_s$ exists for some $s$ and some $m$, 
  $\Align$ fails and returns $\smash{\ell^{(1)}=\cdots=\ell^{(M)}=1}$.

In the analysis we show that when $\bssx\sim\{0,1\}^n$ and $\bssy^*,\bssy^{(1)},\ldots,
  \bssy^{(M)}\sim \Del_\delta(\bssx)$, 
  it is most likely that every $m\in [M]$ satisfies
\begin{equation}\label{condition2}
\left|\hspace{0.05cm}{\overline{\source^*(Q^*_s)}}\bigtriangleup {\overline{\source^{(m)}
  (Q^m_s)}}\hspace{0.05cm}\right|\le 4\gamma t_s,\quad\text{for
  each $s = S,\ldots,1$;}
\end{equation}
in words, this means that the interval $Q^m_s$ of $\bssy^{(m)}$
  (almost) comes
  from the subword of $\bssx$ whose image is $Q_s^*$ in $\bssy^*$. 
The proof proceeds by induction on $s=S,\ldots,1$ (see \Cref{lemma:step-1}).
Assume that (\ref{condition2}) holds for $s+1$:
\begin{equation}\label{condition3}
\left|\hspace{0.05cm}{\overline{\source^*(Q^*_{s+1})}}\bigtriangleup {\overline{\source^{(m)}
  (Q^m_{s+1})}}\hspace{0.05cm}\right|\le 4\gamma t_{s+1}.\end{equation}
Then most likely $I^*_s:=\overline{\source^*(Q^*_s)}$ is contained in 
  $\overline{\source^{(m)}
  (Q^m_{s+1})}$ given (\ref{condition3}) and that 
  $I^*_s$ is roughly the middle one-third of $\smash{\overline{\source^*(Q^*_{s+1})}}$
  (also recall that $\gamma =0.01$ is a small constant).
As a result, $\image^{(m)}(I^*_s)$ would most likely satisfy (\ref{condition1}) 
  as $Q^m_s$ 
  (using that $\delta$ is sufficiently smaller than $\gamma$ and thus,
  the number of bits deleted from $I^*_s$ in getting both $\bssy^*$ and $\bssy^{(m)}$ is 
  smaller than $\gamma t_s$).
On the other hand, let $Q^m_s$ be the interval actually picked by $\Align$.
To finish the proof of (\ref{condition2}), we show that
  when (\ref{condition3}) is violated, 
  since $\bssx \sim \zo^n$, the two subwords
$x_{\overline{\source^\ast(Q^\ast_s)}}$ and
$x_{\overline{\source^{(m)}(Q^m_s)}}$ most likely have large edit distance 
   (see  \Cref{claimhehe}),
  which in turn implies that $\bssy^\ast_{Q^\ast_s}$
  (i.e., the string $\ssw^*_s$) has large edit distance
  from $\bssy^{(m)}_{Q^m_s}$, which contradicts (\ref{condition1}).

\paragraph{Second stage --- determining a consensus location close to $\source^\ast(\ell^\ast)$.}\ 

In the second stage, \Align uses subwords $\smash{ \ssy^{(m)}_{Q^m_1}}$,
  $m\in [M]$, to determine the final locations $\ell^{(1)},\dots,\ell^{(M)}$.
This is done by first identifying a string $\ssw$ that (a) has length at least $0.9t_1$,
  and (b) appears as a subword in at least $95\%$ of $\smash{ \ssy^{(m)}_{Q^m_1}}$, $m\in [M]$.
(When multiple strings $\ssw$ satisfy the two conditions, $\Align$ picks one 
  arbitrarily; when no such $\ssw$ exists, $\Align$ fails and sets $\ell^{(m)}=1$ for all $m\in [M]$.)
Finally $\Align$ finds $\ssw$ in $\smash{\ssy^{(m)}_{Q^m_1}}$ and sets
  $\ell^{(m)}$ to be the location of the first symbol of $\ssw$ in $\smash{\ssy^{(m)}_{Q^m_1}}$, for each $m\in [M]$.
(When $\ssw$ appears in $\smash{\ssy^{(m)}_{Q^m_1}}$ at multiple locations, 
  $\Align$ picks one of them arbitrarily as $\ell^{(m)}$; when $\ssw$
  does not appear, $\Align$ sets $\ell^{(m)}=1$ by default.)


Using $\bssx\sim\{0,1\}^n$ and $\bssy^*,\bssy^{(1)},\ldots,
  \bssy^{(M)}\sim \Del_\delta(\bssx)$, we show that
  most likely $\ell^{(1)},\ldots,\ell^{(M)}$ satisfy the two desired conditions in the definition of ``success''
  (i.e. at least 95\% of $\source^{(m)}(\ell^{(m)})$, $m \in [M]$, agree on the 
  same location $i \in [n]$, and  this consensus location $i$ satisfies (\ref{consensus})).
To give some intuition behind the analysis, we first assume that every $m\in [M]$ satisfies (\ref{condition2}) 
  and in particular,
  \begin{equation}\label{condition4}
\left|\hspace{0.05cm}{\overline{\source^*(Q^*_1)}}\bigtriangleup {\overline{\source^{(m)}
  (Q^m_1)}}\hspace{0.05cm}\right|\le 4\gamma t_1.
\end{equation}
Since $|Q^*_1|=t_1$, most likely $I_1^*:=\overline{\source^*(Q^*_1)}$ has length close to $t_1$. 
Let $I_1^{**}$ denote the~interval obtained from $I_1^*$ by extending it in both
  directions by $4\gamma t_1$ (so $I_1^{**}$ also has length close to $t_1$ since
  $\gamma=0.01$).
Using our choice of $t_1=2H+1$, it follows from simple calculations that 
  most likely at least $95\%$ of $\bssy^{(m)}$, $m\in [M]$, are obtained from $\bssx$ with no deletions in 
  $I_1^{**}$.
Let $G\subseteq [M]$ be the set of such $m\in [M]$. 
It follows from (\ref{condition4}) that $\smash{\bssy^{(m)}_{Q^m_1}}$ comes from $\bssx_{I^m_1}$
  with no deletions for some interval $I^m_1$ such that 
  $|I^m_1\bigtriangleup I^*_1|\le 4\gamma t_1$, for each $m\in G$.

At this point it is clear that $\ssw=\bssx_{\cap_{m\in G}I^m_1}$ 
  would satisfy both conditions (a) and (b) (using $\gamma=0.01$).
On the other hand, if there is a string $\ssw$ that appears in at least $95\%$ of 
  all $x_{I^m_1}$, $m\in [M]$, then for at least $90\%$ of $m\in [M]$, $\ssw$ appears in $\smash{\bssy^{(m)}_{Q^m_1}}$ and 
  $m\in G$. 
Using the randomness of $\bssx \sim \{0,1\}^n$, one can argue that  most likely
  no string of length at least $0.9t_1$ can appear as a subword more than once
  in $\bssx_{I_1^{**}}$.
This implies that at least $90\%$ of $\ell^{(m)}$ returned are in consensus.
To see that the consensus location $i\in [n]$ satisfies (\ref{consensus}), we recall $t_1=2H+1$ and observe that 
  $\source^*(\ell^*)$ appears around the middle of $I_1^{**}$ but 
  $i$~(as the unique location where $\ssw$ appears as a subword in $\ssx_{I_1^{**}}$)
  lies close to its left end.

\subsection{Algorithm \Align} \label{sec:algorithm-align}

We now describe the algorithm \Align which takes as input $\ell^*,\ssy^*,\ssy^{(1)},\ldots,\ssy^{(M)}$ with 
  $\ell^*\in [5\tau\log n: |\ssy^*|- 5\tau\log n]$.  (See \Cref{alg:align} for a formal presentation of the algorithm.)
  
$\Align$ starts by computing a sequence of nested
  subwords of $\ssy^*$ centered at $\ell^*$ as follows.
Let~$t_1=$ $2H+1$ (recall that $H\le 2\log n$) and $t_s=3t_{s-1}$ for each $s\ge 1$, and let 
  $S$ be the smallest integer such that \[
  t_S\ge \tau\log n \quad \text{(and hence $t_S\le3\tau \log n$)}
  \]
  (recall that $\tau = 500$).
Given $\ssy^*$ and $\ell^*$,
we define the sequence of subwords $\ssw^\ast_1, \ldots, \ssw^\ast_S $, where
$\ssw^\ast_s$ is the $t_s$-bit subword $\ssy^\ast_{Q^\ast_s}$
  with $Q^\ast_s=[\ell^*-(t_s-1)/2: \ell^*+(t_s-1)/2]$ 
    centered at $\ell^\ast$ in $\ssy^\ast$.
(Given that $t_S\le 3\tau \log n$, we always have $Q_S^*\subset [|\ssy^*|]$; indeed 
  we have that there are more than $t_S$ elements to the left and to the
  right of $Q_S^*$ in $[|\ssy^*|]$, which  
  is the reason why we only consider $\ell^*$'s that are at least $5\tau \log n$ away
  from both ends of $\ssy^*$.)
  
  

\begin{figure}[!t] 
\newcommand\mycommfont[1]{\small\ttfamily\textcolor{blue}{#1}}
\SetCommentSty{mycommfont}
\setstretch{1.2}
\begin{algorithm}[H]
  \caption{{\Align}}\label{alg:align}
	\DontPrintSemicolon
	\SetNoFillComment
  \KwIn{A location $\ell^\ast$ and a tuple of $M+1$ strings $ \ssy^\ast, \ssy^{(1)}, \ldots, \ssy^{(M)} $}
  \KwOut{$M$ locations $\ell^{(1)}, \ldots, \ell^{(M)}$, where $\ell^{(i)} \in [|\ssy^{(i)}|]$ is a location in $\ssy^{(i)}$}

\smallskip
Compute $ \ssw^\ast_1, \ldots, \ssw^\ast_S $ from $\ssy^\ast$ as defined in 
  \Cref{sec:algorithm-align}. 
  Let $Q_{S+1}^m=[|\ssy^{(m)}|]$ for each $m\in [M]$.\\  

  \For(\ \ \ \ \tcp*[h]{First stage}){each $m \in [M]$}{ \label{alg:line2}
    \For{$s = S , \ldots, 1$}{ \label{alg:line3}
      Find any subword $\ssy^{(m)}_{Q^m_s}$  \label{alg:line4}
          in $\smash{\ssy^{(m)}_{Q^m_{s+1} }}$ (breaking ties arbitrarily)
        that has edit distance at most $2\gamma t_s$ from $\ssw_s^*$; 
     if such a subword does not exist \Return $\ell^{(1)}=\cdots =\ell^{(M)}=1$.}
    }
   
  {Find any string $\ssw$ of length at least $0.9t_1$ that appears 
    as a subword in at least $95\%$ of $\ssy^{(m)}_{{ Q^m_1 }}$, $m\in [M]$.
    If no such $\ssw$ exists, \Return $\ell^{(1)}=\cdots =\ell^{(M)}=1$.
  \ \ \ \ \tcp*[h]{Second stage} 
   \label{line:longest-common-subword}}\\
   \For{each $m \in [M]$}{
    If $\ssy^{(m)}_{Q_1^m}$ has $\ssw$ as a subword, 
      set $\ell^{(m)}$ to be any location such that $\ssy^{(m)}_{Q_1^m}$
      has $\ssw$ as a subword starting at $\ell^{(m)}$; 
    otherwise ($\smash{\ssy^{(m)}_{Q_1^m}}$ does not contain $\ssw$ as a subword\ignore{substring}), set $\ell^{(m)} = 1.$
  }
  \Return $\ell^{(1)}, \ldots, \ell^{(M)}$ 
\end{algorithm}
\caption{The $\Align$ algorithm.}
\label{fig:align}
\end{figure}


We divide the analysis of $\Align$ into two parts.
In the first part (\Cref{sec:prob1}) we begin by~describing some good events over the randomness of $\bD^\ast$, $\bssx$, and $\bD^{(m)}$, $m\in [M]$, where $\bD^\ast$ and $\bD^{(m)}$ are the sets of deleted locations that gave rise to  traces $\bssy^{\ast}$ and $\bssy^{(m)}$ of $\bssx$, respectively.
We then show that these events happen with probability at least $1 - 1/n^2$.
The second part of our analysis (\Cref{sec:det1}) will be entirely deterministic.
We show that $\Align$ succeeds on $\ssx$ with respect to 
  $(\ssy^*,\ssy^{(1)},\ldots,\ssy^{(M)})$ whenever
  $D^*,$ $\ssx$ and $D^{(m)}:m\in [M]$ satisfy all conditions described in the first part  (\Cref{sec:prob1}).

\subsection{Probabilistic Analysis}\label{sec:prob1}

Let $\calD$ denote the distribution over subsets of $[n]$ where $\bD\sim \calD$ is drawn by including
  each integer of $[n]$ independently with probability $\delta$.
We prove \Cref{thm:align-restated} in two steps. 
In this subsection we describe an event over $\bssx\sim \{0,1\}^n$ and $\bD^*,\bD^{(1)},
  \ldots,\bD^{(M)}\sim \calD$ (as deletions used to obtain $\bssy^*,\bssy^{(1)},\ldots,\bssy^{(M)}$~from $\bssx$)
  and show that it happens with probability at least $1-1/n^2$ (see \Cref{lemma:good-rvs}).
In \Cref{sec:det1}, we show that whenever the event occurs, 
  $\Align$ succeeds on $\bssx$ with respect to $(\bssy^*,\bssy^{(1)},\ldots,\bssy^{(M)})$.

We describe the event by imposing conditions on random variables in the 
  following order: first $\bD^\ast\sim \calD$, then $\bssx\sim \{0,1\}^n$ and finally $\bD^{(m)}\sim \calD$, $m\in [M]$.
We describe conditions on each random variable conditioning on the event that
  previous ones have already met conditions imposed on them.

We start with some preliminary claims.




\begin{claim}\label{simpleclaim}
Let $t$ be a positive integer and let
  $1\le i_1<\cdots<i_t\le n$ and $1\le j_1<\cdots<j_t\le n$ with $i_k\ne j_k$ for all $k\in [t]$. 
For $\bssx\sim \{0,1\}^n$, we have $\Pr[\bssx_{i_k} =\bssx_{j_k}$ for all $k\in [t]]=2^{-t}.$
\end{claim} 
\begin{proof}
The proof is by induction on $t$.
  The base case when $t = 1$ is trivial.
  For the inductive step, we assume that the statement holds for $t-1$. Using the induction hypothesis we have
  \begin{align*}
    \Pr \bigl[ \bssx_{i_k} = \bssx_{j_k} \text{ for all $k \in [t]$} \bigr]
    &= \Pr \bigl[ \bssx_{i_t} = \bssx_{j_t} \bigm\mid \bssx_{i_k} = \bssx_{j_k} \text{ for all $k \in [t-1]$} \bigr] \cdot 2^{-(t-1)} .
  \end{align*}
  Without loss of generality we assume that $i_t < j_t$.
  Note that $\bssx_{j_t}$ is still uniform when conditioned on values of $\bssx_{i_k} : k \le t$ and $\bssx_{j_k} : k < t$.
  Therefore the conditional probability on the right hand~side is $1/2$.
This finishes the induction step and the proof of the claim.
\end{proof}

\Cref{simpleclaim} has the following corollary which we will use later:

\begin{corollary}\label{corohehe}
Let $t$ be a positive integer, and let
  $I\ne I'\subseteq [n]$ be two distinct (but not necessarily disjoint) intervals of length $t$. 
For $\bssx\sim \{0,1\}^n$, we have $\Pr[\bssx_I =\bssx_{I'}]=2^{-t}.$
\end{corollary} 


\Cref{simpleclaim} lets us bound the edit distance between subwords of a random string as follows:

\begin{claim}\label{claimhehe}
Let $t$ be a positive integer. Let
$I,I'\subseteq [n]$ be two intervals that satisfy (1) $|I|\ge 25t$ and (2)
   $|I\bigtriangleup I'|\ge t$.
Then we have $\dedit(\bssx_I,\bssx_{I'})< t$ with 
  probability at most $2^{-5t }$ when $\bssx\sim \{0,1\}^n$.
\end{claim} 
\begin{proof}
Having $\dedit(\bssx_I,\bssx_{I'})<  t$ implies that there exist $J\subseteq I $ and
  $J'\subseteq I' $ such that $|J|+|J'|< t$, $|I\setminus J|=|I'\setminus J'|$
  and $\bssx_{I \setminus J}=\bssx_{I' \setminus J'}$.
Fixing such a pair $(J,J')$ and writing 
  $I\setminus J$ as $i_1<i_2<\cdots$ and $I'\setminus J'$ as $j_1<j_2<\cdots$,
  we claim that $i_k\ne j_k$ for all $k$.
To see this we note that having $i_k=j_k$ for some $k$ implies that we need to delete
  at least $|I\bigtriangleup I'|\ge t$ bits from $I$ and $I'$ even just to match the lengths of $I$ to the
  left and to the right of $k$ with those of $I'$, a contradiction with $|J|+|J'|< t$.

Therefore, it follows from \Cref{simpleclaim} that  
  $\bssx_{I\setminus J}=\bssx_{I'\setminus J'}$ with probability at most
$
2^{-(|I|+|I'|- t)/2}.
$
It follows by a union bound on all pairs $(J,J')$ that 
$\dedit(\bssx_I,\bssx_{I'})<  t$ with probability at most
$$
2^{-(|I|+|I'|-  t)/2}\cdot \sum_{k\le t}{|I| \choose k}
\cdot  \sum_{k\le  t}{|I'| \choose k}
\le 2^{-(|I|+|I'|- t)/2} \cdot 2^{H(0.04) (|I|+|I'|)} < 2^{-5t}.
$$
This finishes the proof of the claim.
\end{proof}

\subsubsection{Conditions on $\bD^\ast\sim \calD$} 

We start with conditions on $\bD^\ast\sim \calD$, i.e., the locations
  of bits deleted in $\bssy^*$.
Given an outcome $D^\ast\subseteq [n]$ of $\bD^\ast$, we write $ L^*$ to denote the interval
\[
L^* := [5\tau\log n: n-| D^*|-5\tau \log n].
\]
For each $\ell^*\in L^*$ and $s\in [S]$ we write $Q^*_{s,\ell^*}$ to denote~the interval
  of length $t_s$ that is centered at $\ell^*$.
Let $L_1^*$ denote the set of $\ell^*\in L^*$ such that $$\left| {\overline{\source^\ast(Q^\ast_{s,\ell^*})}} \right| \leq (1+\gamma)t_s,\quad\text{for all $s\in [S]$.}$$

\begin{claim} \label{claim:D-star}
  With probability at least $1-\exp(-n^{0.1})$ over $\bD^\ast$, 
   we have $|\bL^*\setminus \bL^*_1|\le 2^{-0.2H}n$.
\end{claim}
\begin{proof}
Given $D^*$, for each 
  $\ell^*\in L^*\setminus L_1^*$ 
  there is an $s\in [S]$ such that $| {\overline{\source^\ast(Q^\ast_{s,\ell^*})}}|\ge (1+\gamma)t_s$.~We can get from it an interval $I$ with $|I|=(1+\gamma)t_s$ and $|\image^*(I)|\le |Q_{s,\ell^*}^*|=t_s$        by deleting elements 
  from the right end of $\smash{\overline{\source^\ast(Q^\ast_{s,\ell^*})}}$.
We note that the intervals $I$ obtained from different $\ell^*\in L^*\setminus L_1^*$
  are different.
(To obtain the same interval, we must  use the same $s\in [S]$ because of the 
  length of $I$; on the other hand, sharing the same left end and the same $s$ 
  implies that the $\ell^*$ is the same as well.)
Therefore, $|L^*\setminus L_1^*|$ is at most the number of intervals
  $I\subseteq [n]$ such that $|I|=(1+\gamma)t_s$ for some $s\in [S]$ and 
  $|\image^*(I)|\le t_s$.
Below we upperbound the latter when $\bD^*\sim \calD$.

We apply the McDiarmid inequality (\Cref{thm:Mcdiarmid}). We draw $\bD^*$ by
  drawing $n$ independent random indicator variables
  $\bX_1,\ldots,\bX_n$  with $\bX_k=1$
  with probability $\delta$ (so $k\in \bD^*$ if $\bX_k=1$).
We use $f(\bX_1,\ldots,\bX_n)$ to denote the number of $I\subseteq [n]$
  such that $|I|=(1+\gamma)|t_s|$ for some $s\in [S]$ and $|\image^*(I)|\le t_s$. 
On the one hand,  
  the probability of an interval $I$ with $|I|=(1+\gamma)t_s$ satisfying $|\image^* (I)|\le t_s$
  is at most
$ 
2^{|I|}\cdot \delta^{\gamma t_s}\le \delta^{\gamma t_s/2} 
$, by using $\delta \le 1/K$ and making $K$ sufficiently large.
As a result,
$$
\bE\big[f\big]\le n\left(\sum_{s=1}^S \delta^{\gamma t_s/2}\right)= \delta^{\Omega(t_1)}n=\delta^{\Omega(H)}n. 
$$
On the other hand, flipping one variable $\bX_k$ can change
  $f$ by no more than $O(t_S)=O(\log n)$.
Thus it follows from the McDiarmid inequality that
$$
f(\bX_1,\ldots,\bX_n)\le \delta^{\Omega(H)}n +\tilde{O}(n^{0.55})
\le 2^{-0.2H} n,
$$ 
with probability at least $1-\exp(-n^{0.1})$, where we used 
  that $\delta$ is sufficiently small and $H\le 2 \log n$~in the last inequality. 
This finishes the proof of the claim.
\end{proof}




\subsubsection{Conditions on $\bssx\sim \{0,1\}^n$}
We fix a $D^*\subseteq [n]$ that satisfies \Cref{claim:D-star} when describing the conditions 
  for $\bssx$ and $\bD^{(m)}$ below.
As $D^*$ is fixed, $L^*,L^*_1$ and 
  $\smash{\overline{\source^\ast(Q^\ast_{s,\ell^*})}}$ for each $\ell^*\in L_1^*$ 
  are all fixed and are no longer  
  random variables. 
  For each $\ell^\ast \in L^\ast_1$, for brevity we write $I^\ast_{s,\ell^\ast}$ to denote 
  $\smash{\overline{\source^\ast(Q^\ast_{s,\ell^*})}}$, and we observe that for each 
  $\ell^*\in L^*_1$ we have 
\begin{equation}\label{boundsforI}
t_s\le |I^*_{s,\ell^*}|\le (1+\gamma)t_s.
\end{equation}
The conditions for $\bssx\sim \{0,1\}^n$ are given in the next three claims.

\begin{claim}\label{condx1}
With probability at least $1-1/n^3$ over $\bssx\sim \{0,1\}^n$, 
  every $\ell^*\in L_1^*$  and every interval $I\subseteq [n]$ with $|I\bigtriangleup I_{S,\ell^*}^*|\ge 4\gamma t_S$ satisfy $\dedit(\bssx_{I^*_{S,\ell^*}},\bssx_I)\ge 4\gamma t_S$.
\end{claim}
\begin{proof}
Recall that $|I^*_{S,\ell^*}|\ge t_S$ and $\gamma = 0.01$.
Fix  an $\ell^*\in L_1^*$ and an interval $I$ with $|I\bigtriangleup I^*_{S,\ell^*}|\ge 4\gamma t_S$. 
By \Cref{claimhehe} we have that $\smash{\dedit(\bssx_{I^*_{S,\ell^*}},\bssx_I)< 4\gamma t_S}$ 
  occurs with probability
  at most $\smash{2^{-20\gamma t_S}\le 1/n^{100}}$, 
    using $\smash{t_S\ge \tau\log n}$ and $\gamma \tau=5$.
The claim follows by a union bound over no more than $\smash{n^3}$ pairs of $\smash{\ell^*}$ and $I$.
\end{proof}

For the next claim we need the following notation.
Given $\ell^*\in L^*$ and $s\in [S]$, we let $N(I^*_{s,\ell^*})$ denote
  the interval obtained by adding $t_s$ elements to both ends of $I^*_{s,\ell^*}$
  (so $ N(I^*_{s,\ell^*}) $ has length $|I^*_{s,\ell^*}|+2t_s$).
Note that $N(I^*_{s,\ell^*})$ is an interval contained in $[n]$ given that 
  $\ell^*\in L^*$ is at least $5\tau \log n$ from both ends of $[n-|D^*|]$ (recall that $t_s \leq t_S \leq 3 \tau \log n$).

\begin{claim} \label{claim:x-ED-dist}
  With probability at least $1 - \exp(-n^{0.1})$ over $\bssx\sim \{0,1\}^n$,
  all but at most $2^{-0.2H}n$ many locations $\ell^* \in L^*_1$ 
  satisfy the following condition: 
  For any $s \in [2:S]$ and any interval $I \subseteq N(I^*_{s,\ell^*})$
             such that $|I\bigtriangleup I^*_{s-1,\ell^*}|\ge 4\gamma t_{s-1}$, we have
$ \dedit(\bssx_{I},
\bssx_{I^*_{s-1,\ell^*}}) \ge 4\gamma t_{s-1}.$ 
\end{claim}
\begin{proof}
Again we use the McDiarmid inequality. Let $f(\bssx_1,\ldots,\bssx_n)$ denote the 
  number of $\ell^*\in L^*_1$ that violates the condition.
We first upperbound the probability over $\bssx_1,\dots,\bssx_n$ of a fixed $\ell^*\in L_1^*$ violating the condition.

For each $s\in [2:S]$ and each interval $I\subseteq N(I^*_{s,\ell^*})$
  with $|I\bigtriangleup I_{s-1,\ell^*}^*|\ge 4\gamma t_{s-1}$ (note that $\smash{I^*_{s-1,\ell^*}}$ has 
  length at least $t_{s-1}$), by \Cref{claimhehe} the probability
  of $\smash{\dedit(\bssx_I,\bssx_{I^*_{s-1,\ell^*}})<4\gamma t_{s-1}}$ is at most $\smash{2^{-20\gamma t_{s-1}}}$. 
As
$$\left|N(I^*_{s,\ell^*})\right|\le |I^*_{s,\ell^*}|+2t_s\le (3+\gamma)t_s,$$
by (\ref{boundsforI}), 
  it follows from a union bound that each $\ell^*$ in $L^*_1$ violates the condition
  with probability at most
$$
\sum_{s\in [2:S]} O(t_s^2)\cdot 2^{-20\gamma t_{s-1}}\le 
\sum_{s\in [2:S]} 2^{-19\gamma t_{ s-1}}\le 2^{-18\gamma t_{ 1}}\le 
2^{-0.3H},
$$  
using $t_1=2H+1$ and $H$ is sufficiently large, and hence $\E[f] \leq 2^{-0.3H} n$.
Given that each variable $\bssx_i$ can change $f$ by no more than $O(\log n)$,
  the lemma follows from arguments similar to the proof of \Cref{claim:D-star}.  
\end{proof}

\begin{claim} \label{claim:unique-lcs}
  With probability at least $1-\exp(-n^{0.1} )$ over $\bssx\sim \{0,1\}^n$,
  all but at most $2^{-0.2H}n$
  many $\ell^*\in L^*_1$  
  are such that
  no two subwords of $\bssx_{N(I^*_{1,\ell^*})}$ of length $H$ are the same.
\end{claim}
\begin{proof}
The probability of an $\ell^*\in L_1^*$ violating the above condition is 
  at most $ O(t_1^2)\cdot 2^{-H}\le 2^{-0.5H}$  by \Cref{corohehe}.
The   proof follows from a similar application of McDiarmid inequality.
\end{proof}


\subsubsection{Conditions on $\bD^{(1)},\ldots,\bD^{(M)}\sim \calD$}
We fix an outcome $D^*$ that satisfies \Cref{claim:D-star} and 
  a string $\ssx\in \{0,1\}^n$ that satisfies \Cref{condx1},
  \Cref{claim:x-ED-dist}, and \Cref{claim:unique-lcs}. 

We now describe some useful conditions on $\bD^{(1)},\ldots,\bD^{(M)}$.
We start with two conditions for every $\bD^{(m)}$.

\begin{claim}\label{claimhaha}
With probability at least $1-1/n^3$ over $\bD\sim \calD$, every  interval $I\subseteq [n]$ of length
   at most $(1+3\gamma)t_S$ satisfies that $|\image(I)|\ge |I|-\gamma t_S$. 
\end{claim}
\begin{proof}
For each $I\subseteq [n]$ of length 
  at most $(1+3\gamma)t_S$, we have $|\image(I)|<|I|-\gamma t_S$ with probability at most
  $2^{|I|}\cdot  \delta^{\gamma t_S}\le \delta^{\gamma t_S/2}$, where the inequality uses that $\delta$ 
  is sufficiently small.
Using $t_S\ge \tau\log n$ we have that $\delta^{\gamma t_S/2} \le \delta^{5\log n/2}\le 1/n^5$
  (as $\delta$ is sufficiently small).
The claim then follows from a union bound.
\end{proof}

\begin{remark}\label{remark:1}
{We note that the event described in \Cref{claimhaha} implies that any interval $Q\subseteq [n-|\bD|]$ with
  length at most $(1+2\gamma)t_S$ must satisfy $|\overline{\source(Q)}|\le |Q|+\gamma t_S$.
To see this, let $I=\overline{\source(Q)}$ and assume for a contradiction that $|I|>|Q|+\gamma t_S$.
If $|I|\le (1+3\gamma)t_S$ then $I$ violates the event of \Cref{claimhaha};
if $|I|> (1+3\gamma)t_S$ then we can delete bits of $I$ from the beginning to obtain
  an interval $I'$ with $|I'|=(1+3\gamma)t_S$, which satisfies $|\image(I')|<|\image(I )|\le (1+2\gamma)t_S$ and thus,
$I'$ violates the condition of \Cref{claimhaha}.}\end{remark}

Let $L_2^*$ be the set of all $\ell^*\in L^*_1$ that satisfy
the conditions in \Cref{claim:x-ED-dist} and \Cref{claim:unique-lcs}. 
\begin{claim} \label{claim:D}
With probability at least $1-\exp(-n^{0.1})$ over $\bD\sim \calD$, all but at most
  $2^{-0.2H}n$ many $\ell^*\in L^*_2$
  satisfy the following condition: For every $s\in [2:S]$ and 
   every interval $I\subseteq N(I^*_{s,\ell^* }) $ of length at most $(1+3\gamma)t_{s-1}$,
   we have $|\image(I)|\ge |I|-\gamma t_{s-1}$.
\end{claim}
\begin{proof}
We upper bound the probability of an $\ell^*\in L_2^*$ violating the condition above,
  and then apply the McDiarmid inequality.
Fixing an $\ell^* \in L^*_2$, a value of $s\in [2:S]$, and any
   interval $I\subseteq N(I^*_{s,\ell^* }) $ of length at most $(1+3\gamma)t_{s-1}$,
   $I$ violates the condition with probability $\smash{2^{|I|}\cdot 
   \delta^{\gamma t_{s-1}}\le \delta^{\gamma t_{s-1}/2}}$.
By a union bound (over all possibilities for $s$ and $I$), the probability of $\ell^*$ violating the condition is at most
$$
\sum_{s\in [2:S]} O(t_{s }^2)\cdot \delta^{\gamma t_{s-1}/2},
$$ 
and hence the expected number of $\ell^\ast \in L^\ast_2$ that
violate the condition is at most $n$ times this, which is at most
$2^{-0.3H}n$ using that $t_{s-1} \geq t_1 = 2H+1$ and
$\delta$ is sufficiently small.
Finally, given that the outcome of each independent event (of whether an element in $[n]$ is included in $\bD$ or not) can change the number of $\ell^\ast$ that satisfy the condition by at most $O(\log n)$,
  the lemma follows from arguments similar to the proof of \Cref{claim:D-star}.  
\end{proof}

\begin{remark}\label{remark:2}
\Cref{remark:1} applies similarly: Whenever the condition holds for $\ell^*\in L_2^*$,
  any interval $Q\subseteq \image(N(I^*_{s,\ell^*}))$ for any $s$ with length at most $(1+2\gamma)t_{s-1}$ 
  satisfies $|\overline{\source(Q)}|\le |Q|+\gamma t_{s-1}$.
\end{remark}

The last condition considers $\bD^{(1)},\ldots,\bD^{(M)}$ together:



\begin{claim} \label{claim:Chernoff}
  With probability at least $1 - \exp(-n^{0.1})$ over $\bD^{(1)},\ldots,\bD^{(M)}\sim \calD$, 
  all but at most $2^{-0.2H}n$ many $\ell^*\in L^*_2$ satisfy the following condition:  
  At least $95\%$ of $m\in [M]$ satisfy $$N(I^*_{1,\ell^*}) \cap \bD^{(m)} = \varnothing,$$  i.e., no bit of the subword $ \bssx_{N(I^*_{1,\ell^*})}$ of $\bssx$ is deleted in 
  at least $95\%$ of the traces $\bssy^{(1)},\ldots,\bssy^{(M)}$.
\end{claim}
\begin{proof}
Consider drawing $\bD^{(1)},\ldots,\bD^{(M)}$ by drawing $nM$ independent indicator random 
  variables $\smash{\bX_k^{(m)}}$, $k\in [n]$ and $m\in [M]$, with $\smash{k\in \bD^{(m)}}$ if $\smash{\bX_{k}^{(m)}=1}$.
We write $f$ to denote the number of $\ell^*\in L^*_2$ such that at least $5\%$ of $m\in [M]$
  have $\smash{N(I^*_{1,\ell^*})\cap \bD^{(m)}\ne \varnothing}$.
On the one hand, fixing an outcome of $\ell^*$, the probability of $\smash{N(I^*_{1,\ell^*})\cap \bD^{(m)}= \varnothing}$
  is at most $7\delta H $ given that $\smash{|N(I^*_{1,\ell^*})|\le (3+\gamma)t_1\le 
  7H}$, and hence
  the probability of $\ell^*$ being one of the locations counted in $f$ is at most
  $2^M\cdot (7\delta H )^{0.05M}$. 
Recalling the constraint \Cref{eq:H} on $H$,  we have that
\begin{equation} \label{eq:bound-on-deltaH}
\delta H=\frac{1}{K}\cdot \delta M\log \left(\frac{1}{\delta M}\right)\le \frac{\sqrt{\delta M}}{K}
\end{equation}
where the inequality holds given that $\delta M$ is 
  sufficiently small (observe from \Cref{eq:assumptions} that $\delta M \leq 1/K$).
Hence the probability is at most
$$
2^M\cdot (7\delta H)^{0.05M}\le (\delta M)^{0.025 M}\le 2^{-0.3H}
$$
where the first inequality is by \Cref{eq:bound-on-deltaH} and the second uses \Cref{eq:H} and the fact that $K$ is sufficiently large.
Recalling that $H=O(\log n)$, the claim follows from the McDiarmid inequality using similar arguments to those given above and the fact that changing the outcome of any one of the $nM$ independent indicator random variables can only change $f$ by at most $O(H)$.
\end{proof}

\subsubsection{Conclusion of Probabilistic Analysis}

We summarize our probabilistic analysis with the following corollary, which combines all the claims from this subsection.

\begin{corollary} \label{lemma:good-rvs}
With probability at least $1-1/n^2$ over the randomness of $\bD^*,\bssx,\bD^{(1)},\ldots,\bD^{(M)}$, all of the following hold:
\begin{enumerate}
\item $\bD^*$ satisfies \Cref{claim:D-star};\vspace{-0.16cm}
\item $\bssx$ satisfies \Cref{condx1},
  \Cref{claim:x-ED-dist} and \Cref{claim:unique-lcs}; \vspace{-0.16cm}
\item Every $\bD^{(m)}$, $m\in [M]$, satisfies \Cref{claimhaha} and 
  \Cref{claim:D}, and \vspace{-0.16cm}
\item  $\bD^{(1)},\ldots,\bD^{(M)}$ together satisfy \Cref{claim:Chernoff}.
\end{enumerate}
%
\end{corollary}

\subsection{Deterministic Analysis}\label{sec:det1}

The rest of \Cref{sec:align} is dedicated to proving the following lemma, which 
  finishes the proof of \Cref{thm:align-restated} (and hence \Cref{thm:align}):
  
\begin{lemma} \label{maindetlemma} 
$\Align$ succeeds
  on $\ssx$ with respect to $\ssy^*,\ssy^{(1)},\ldots,\ssy^{(M)}$
  when they satisfy \Cref{lemma:good-rvs}.
\end{lemma}

Assume that $\ssx,D^*,D^{(1)},\ldots,D^{(M)}$ satisfy all conditions of \Cref{lemma:good-rvs}. 
Then we have that all but at most 
$O(M2^{-0.2H}n)\le 2^{-0.1H}n$
  {(where the $M$ comes from a union bound in item (vi) and the inequality follows from $H>M/K\ge \sqrt{M}$ using \Cref{eq:H} and $M\ge K^2$ from \Cref{eq:assumptions} and thus, $2^{0.1H}$ is enough to cover $O(M)$ when $M$ is sufficient large)}
many $\ell^*\in L^*$ satisfy the following 
  list of conditions (below we use $I^*_s$ to denote 
   $I^*_{s,\ell^*}$ for convenience given that $\ell^*$ is fixed in the rest of the proof):
\begin{flushleft}\begin{enumerate}
\item[(i)] $|I^*_{s }|< (1+\gamma)t_s$ for every $s\in [S]$ (\Cref{claim:D-star});
\item[(ii)] Every interval $I\subseteq [n]$ with $|I\bigtriangleup I^*_{S }|\ge 4\gamma t_S$,
  has $\dedit(\ssx_I,\ssx_{I^*_S})\ge 4\gamma t_S$ (\Cref{condx1});
\item[(iii)]
For all $s\in [2:S]$ and intervals $I\subseteq N(I^*_{s})$    
  with $|I\bigtriangleup I^*_{s-1}|\ge 4\gamma t_{s-1}$, it holds that
  $\dedit(\ssx_I,\ssx_{I^*_{s-1}})\ge 4\gamma t_{s-1}$ (\Cref{claim:x-ED-dist});
\item[(iv)] No two subwords of $\ssx_{N(I^*_{1})}$ of length $H$ are the same (\Cref{claim:unique-lcs});
\item[(v)] (\Cref{claimhaha} and \Cref{remark:1}) For all $m\in [M]$, {$|\image^{(m)}(I_S^*)|\ge |I_S^*|-\gamma t_S$} and
  every interval $Q^{m}\subseteq [|\ssy^{(m)}|]$
  of length at most  $(1+2\gamma)t_S$ satisfies
  $$\left|\overline{\source^{(m)}(Q^{m})}\right|\le |Q^{m}|+\gamma t_S;$$
\item[(vi)] (\Cref{claim:D} and \Cref{remark:2}) For all $m\in [M]$ and $s\in [2:S]$, $|\image^{(m)}(I_{s-1}^*)|\ge |I_{s-1}^*|-\gamma t_{s-1}$ and
  every interval $Q^{m}\subseteq \image(N(I^*_s))$ of length at most $(1+2\gamma)t_{s-1}$ satisfies 
  $$\left|\overline{\source^{(m)}(Q^{m})}\right|\le |Q^m|+\gamma t_{s-1};$$
 \item[(vii)] At least $95\%$ of $m\in [M]$ satisfy that $N(I_1^*)\cap D^{(m)}=\varnothing$ (\Cref{claim:Chernoff}).
\end{enumerate}\end{flushleft}

\subsubsection{First stage: Locating a small neighborhood of $\source^\ast(\ell^\ast)$ in each trace
  $\ssy^{(m)}$}

\ignore{







}

\usetikzlibrary{decorations.pathreplacing}
\usetikzlibrary{math}
\tikzmath{\yy=-30; \yx=0; \yys=30; \sh=55; \w=140; }
\begin{figure}[!t]
\centering
\begin{tikzpicture}[d/.style={draw=black,font=\scriptsize}, x=1mm, y=1mm, z=1mm]

  \draw[thick,dotted](0,\yys)--(\w,\yys) node at (-6,\yys)[left]{$\bssy^\ast$};
  \draw[thick,dotted](0,\yx)--(\w,\yx) node at (-6,\yx)[left]{$\bssx$};
  \draw[thick,dotted](0,\yy)--(\w,\yy) node at (-6,\yy)[left]{$\bssy^{(m)}$};

  \draw[d][decorate, decoration={brace}, yshift=+2ex]  (\sh-8,\yys) -- node[above=0.4ex] {$Q^\ast_{s-1}$}  (\sh+8,\yys);

  \fill[fill=red, opacity=0.2](\sh-24,\yys+1)--(\sh+24,\yys+1)--(\sh+24,\yys)--(\sh-24,\yys);
  \draw[d][decorate, decoration={brace,mirror}, yshift=-2ex]  (\sh-24,\yys) -- node[below=0.4ex] {$Q^\ast_{s}$}  (\sh+24,\yys);

  \draw[thick,red](\sh-8,\yys+1)--(\sh+8,\yys+1);


  \fill[fill=red, opacity=0.2](\sh-24,\yx+1)--(\sh+24,\yx+1)--(\sh+24,\yx)--(\sh-24,\yx);
  \draw[d][decorate, decoration={brace}, yshift=2ex]  (\sh-24,\yx) -- node[above=0.4ex] {$I^\ast_s = \overline{\source^\ast(Q^\ast_{s})}$}  (\sh+24,\yx);

  \fill[fill=blue, opacity=0.2](\sh-26,\yx)--(\sh+20,\yx)--(\sh+20,\yx-1)--(\sh-26,\yx-1);
  \draw[d][decorate, decoration={brace,mirror}, yshift=-2ex]  (\sh-26,\yx) -- node[below=0.4ex] {$\overline{\source^{(m)}(Q^m_s)}$}  (\sh+20,\yx);

  \fill[fill=green, opacity=0.2](\sh-40,\yx-1)--(\sh+40,\yx-1)--(\sh+40,\yx-2)--(\sh-40,\yx-2);
  \draw[d][decorate, decoration={brace,mirror}, yshift=-6ex]  (\sh-40,\yx) -- node[below=0.4ex] {$N(I^\ast_s)$}  (\sh+40,\yx);


  \draw[thick,red] (\sh-10,\yx+1)--(\sh+10,\yx+1);
  \draw[thick,blue] (\sh-8,\yx-1)--(\sh+12,\yx-1);




  \fill[fill=blue, opacity=0.2](\sh-24,\yy+1)--(\sh+24,\yy+1)--(\sh+24,\yy)--(\sh-24,\yy);
  \draw[d][decorate, decoration={brace}, yshift=2ex]  (\sh-24,\yy) -- node[above=0.4ex] {$Q^m_s$}  (\sh+24,\yy);

  \fill[fill=green, opacity=0.2](\sh-44,\yy)--(\sh+36,\yy)--(\sh+36,\yy-1)--(\sh-44,\yy-1);
  \draw[d][decorate, decoration={brace,mirror}, yshift=-2ex]  (\sh-44,\yy) -- node[below=0.4ex] {$\image^{(m)}(N(I^\ast_s))$}  (\sh+36,\yy);

  \draw[thick,red](\sh-12,\yy+1)--(\sh+6,\yy+1);

  \draw[thick,blue](\sh-12,\yy-1)--(\sh+8,\yy-1);

\end{tikzpicture}
\caption{First stage of \Align:  The 3 red lines are the subwords $\ssy^{\ast}_{Q^\ast_{s-1}}$, $\ssx_{I^\ast_{s-1} = \overline{\source^\ast(Q^\ast_{s-1})}}$, and $\ssy^{(m)}_{Q^m = \image^{(m)}(I^\ast_{s-1})}$.
  The blue lines are the subwords $\ssx_{\overline{\source^{(m)}(Q^m_{s-1})}}$, and $\ssy^{(m)}_{Q^m_{s-1}}$ found by \Align.
  In the completeness argument, we have $Q^m = \image^{(m)}(I^\ast_{s-1}) \subseteq Q^m_s$ because $I^\ast_{s-1} \subseteq I^\ast_s$ and $\abs{\overline{\source^{(m)}(Q^m_s)} \bigtriangleup I^\ast_s}$ is small.
  In the soundness argument, as $Q^m_{s-1} \subseteq Q^m_s$ and $\abs{\overline{\source^{(m)}(Q^m_s)} \bigtriangleup I^\ast_s}$ is small, we have $I := \source^{(m)}(Q^m_{s-1}) \subseteq N(I^\ast_s)$, which implies $Q^m_{s-1} \subseteq \image^{(m)}(N(I^\ast_s))$.
}
\end{figure}

We prove the following lemma for the first stage of $\Align(\ell^*,\ssy^*,\ssy^{(1)},\ldots,\ssy^{(M)})$ (\cref{alg:line2,alg:line3,alg:line4}):

\begin{lemma} \label{lemma:step-1}
For every $m\in [M]$, 
  the final interval $Q^m_1$ found by \Align in the first stage satisfies 
\begin{equation}\label{eq:last}
\left| \overline{\source^{(m)}( Q_1^m)} \bigtriangleup I_1^*\right|\le 4 \gamma t_1.
\end{equation}
\end{lemma}
\begin{proof}
We prove by induction on $s=S,\ldots,1$ that 
\begin{equation}\label{hehe1}
\left| \overline{\source^{(m)}( Q_s^m)}\bigtriangleup I^*_s\right|\le 4 \gamma t_s.
\end{equation}
We start with the base  case $s=S$.
First we establish completeness by showing that there is an interval $Q^{m}\subset [|\ssy^{(m)}|]$ such that 
  $\smash{\dedit(\ssw_S^*,\ssy^{(m)}_{Q^m})\le 2\gamma t_S}$.
To this end, let $Q^m=\image^{(m)}(I^*_S)$, and observe that this is an interval in $[|\ssy^{(m)}|]$. We have
$$
\dedit(\ssw_S^*,\ssy^{(m)}_{Q ^m})\le \dedit(\ssw_S^*,\ssx_{I^*_S})+
  \dedit(\ssy^{(m)}_{Q^m},\ssx_{I_S^*})< 2\gamma t_S,
$$
where we used (i) $|I_S^*|< (1+\gamma)t_S$ to upper bound the first edit distance by $\gamma t_S$,
  and (v) $|Q^m|\ge |I_S^*|-\gamma t_S$ to upper bound the second edit distance by $\gamma t_S$.
Next we establish soundness by showing that any interval $Q_S^m$ picked by $\Align$ satisfies  (\ref{hehe1}).
Let $\smash{I=\overline{\source^{(m)}(Q_S^m)}}$. Then we have
\begin{equation}\label{lastlast}
\dedit(\ssx_I,\ssx_{I_S^*})\le \dedit(\ssx_I,\ssy^{(m)}_{Q_S^m})+
\dedit(\ssy^{(m)}_{Q_S^m},\ssw_S^*)+\dedit(\ssw_S^*,\ssx_{I_S^*})< \gamma t_S+2\gamma t_S
+\gamma t_S=
4\gamma t_S.
\end{equation}
To see that $\smash{\dedit(\ssx_I,\ssy^{(m)}_{Q_S^m})\le \gamma t_S}$ we first notice that 
  $|Q_S^m|\le (1+2\gamma)t_S$ because $\smash{\dedit(\ssy^{(m)}_{Q_S^m},\ssw_S^*)\le 2\gamma t_S}$.
It then follows from item (v) that $\smash{|I|\le |Q_S^m|+\gamma t_S}$, which gives $\smash{\dedit(\ssx_I,\ssy^{(m)}_{Q_S^m})\le \gamma t_S}$.  For the last summand, as in the completeness argument (i) gives that $|I_S^*|< (1+\gamma)t_S$ and hence $\dedit(\ssw_S^*,\ssx_{I^*_S}) \leq \gamma t_S$.  The soundness part then follows from (\ref{lastlast}) and (ii).

With the base case in hand, assume for the inductive step that (\ref{hehe1}) holds for some $s\in [2:S]$. 
We use this to prove it for $s-1$.

The completeness proof is similar to the base case:
let $Q^m=\image^{(m)}(I_{s-1}^*)$. 
It follows from (\ref{hehe1}) on $s$ that $Q^m\subseteq Q_s^m$.
Then 
$$
\dedit(\ssw_{s-1}^*,\ssy^{(m)}_{Q^m})\le \dedit(\ssw_{s-1}^*,\ssx_{I_{s-1}^*})+
  \dedit(\ssy^{(m)}_{Q^m},\ssx_{I_{s-1}^*})\le 2\gamma t_{s-1},
$$
where we used (i) and (vi).

The soundness argument is also similar to the base case. Let $Q_{s-1}^m$ be the interval found by $\Align$  
  and~let $I=\smash{\overline{\source^{(m)}(Q_{s-1}^m)}}$. 
Then 
  $Q_{s-1}^m\subseteq Q_s^m$  and it follows from (\ref{hehe1}) (and the definition
  of $N(\cdot)$) that 
$$
I=\overline{\source^{(m)}(Q_{s-1}^m)}\subseteq \overline{\source^{(m)}(Q_s^m)}\subseteq N(I_s^*)
$$
and thus, $Q_{s-1}^m\subseteq \image(N(I_s^*))$. We also have 
  $|Q_{s-1}^m|\le (1+2\gamma)t_{s-1}$ given that $y^{(m)}_{Q_{s-1}^m}$ has edit distance at most $2\gamma t_{s-1}$
  from $\ssw_{s-1}^*$ (which is of length $t_{s-1}$).
As a result, analogous to (\ref{lastlast}), we have
\begin{equation} \label{eq:lastlastlast}
\dedit(\ssx_I,\ssx_{I_{s-1}^*})\le \dedit(\ssx_I,\ssy^{(m)}_{Q_{s-1}^m})+
\dedit(\ssy^{(m)}_{Q_{s-1}^m},\ssw_{s-1}^*)+\dedit(\ssw_{s-1}^*,\ssx_{I_{s-1}^*})<  
4\gamma t_{s-1},
\end{equation}
where $\dedit(\ssx_I,\ssy^{(m)}_{Q_{s-1}^m})\le \gamma t_{s-1}$ follows as in the base case but now using  (vi) rather than (v).
The soundness follows from (\ref{eq:lastlastlast}) and (iii), and the inductive step is completed.
\end{proof}

\subsubsection{Second stage: 
Determining a consensus location close to $\source^\ast(\ell^\ast)$.}  \label{sec:second-stage}

We finish the proof of \Cref{maindetlemma} with the following lemma for the second stage of $\Align$:

\begin{lemma}
Locations $\ell^{(1)},\ldots,\ell^{(M)}$ returned by $\Align$ satisfy the following two conditions:
(A) at least 90\% of $\source^{(m)}(\ell^{(m)})$, $m \in [M]$, agree on the same location $i \in [n]$ and
(B) the consensus location $i$ satisfies 
\begin{equation}\label{consensus2}
  \source^*(\ell^*) - 2H\le i \leq \source^*(\ell^*) .
\end{equation}
\end{lemma}


\tikzmath{\yyd=-45; \yyc=-35; \yyb=-25; \yya=-15; \yx=0; \yys=20; \sh=50; \w=140; }
\begin{figure}[!t]
\centering
\begin{tikzpicture}[d/.style={draw=black,font=\scriptsize}, x=1mm, y=1mm, z=1mm]

  \draw[thick,dotted](0,\yys)--(\w,\yys) node at (-6,\yys)[left]{$\bssy^\ast$};
  \draw[thick,dotted](0,\yx)--(\w,\yx) node at (-6,\yx)[left]{$\bssx$};
  \draw[thick,dotted](0,\yya)--(\w,\yya);
  \draw[thick,dotted](0,\yyb)--(\w,\yyb);
  \draw[thick,dotted](0,\yyc)--(\w,\yyc);
  \draw[decorate, decoration={brace,mirror}, xshift=-4ex]  (0,\yya+1) -- node[left=1ex] {$\bssy^{(m)}, m \in G$}  (0,\yyc-1);

  \draw[d][decorate, decoration={brace}, yshift=+1ex]  (\sh-20,\yys) -- node[above=0.4ex] {$Q^\ast_1$}  (\sh+20,\yys);
  \draw[thick,red](\sh-20,\yys+1)--(\sh+20,\yys+1);


  \fill[fill=red,opacity=0.2](\sh-30,\yx+1)--(\sh+24,\yx+1)--(\sh+24,\yx)--(\sh-30,\yx);
  \draw[d][decorate, decoration={brace}, yshift=+5ex]  (\sh-30,\yx) -- node[above=0.4ex] {$N(I^\ast_1)$}  (\sh+24,\yx);

  \draw[thick,red](\sh-24,\yx+1)--(\sh-14,\yx+1);
  \draw[thick,red](\sh-12,\yx+1)--(\sh-5,\yx+1);
  \draw[thick,red](\sh-4,\yx+1)--(\sh+4,\yx+1);
  \draw[thick,red](\sh+5,\yx+1)--(\sh+10,\yx+1);
  \draw[thick,red](\sh+11,\yx+1)--(\sh+18,\yx+1);
  \draw[d][decorate, decoration={brace}, yshift=+1ex]  (\sh-24,\yx) -- node[above=0.4ex] {$I^\ast_1 = \overline{\source^\ast(Q^\ast_1)}$}  (\sh+18,\yx);

  \draw[thick,blue](\sh-14,\yx-1)--(\sh+26,\yx-1);

  \draw[thick,green](\sh-26,\yx-2)--(\sh+14,\yx-2);

  \draw[thick,purple](\sh-20,\yx-3)--(\sh+20,\yx-3);

  \draw[thick,orange](\sh-14,\yx-4)--(\sh+14,\yx-4);


  \draw[thick,blue](\sh-28,\yya+1)--(\sh+12,\yya+1);

  \draw[thick,green](\sh-17,\yyb+1)--(\sh+23,\yyb+1);

  \draw[thick,purple](\sh-20,\yyc+1)--(\sh+20,\yyc+1);

\end{tikzpicture}
\caption{Second stage of \Align:
  The red line is the subword $\ssy^\ast_{Q^\ast_1} = \ssw^\ast_1$ and it appears as the disconnected red segments in $\ssx_{\source^\ast(Q^\ast_1)}$.
  The blue, green and purple lines are the subwords $\ssy^{(m)}_{Q^m_1}, m\in G$ and $\ssx_{I^m = \overline{\source^{(m)}(Q^m_1)}}, m \in G$.
  Since $m \in G$, we have that $N(I^\ast_1) \cap D^{(m)} = \emptyset$, and so $\smash{\ssy^{(m)}_{Q_1^m}=\ssx_{I^m}}$.
  The orange line is the common subword $\ssx_{\cap_{m\in G} I^m}$ that appears in the three subwords $\ssy^{(m)}_{Q^m_1}$.
}
\end{figure}

\begin{proof}
It follows from \Cref{lemma:step-1} that for every $m \in [M]$, the interval $Q_1^m$ satisfies (\ref{eq:last}).
Let $G$ be the set of $m\in [M]$ such that $N(I_1^*)\cap D^{(m)}=\varnothing$; by (vii) we have that $|G| \geq 0.95 M.$ 
Let $\smash{I^m =\overline{\source^{(m)}(Q^m_1)}}$. 
It follows from (\ref{eq:last}) that every $m\in [M]$ 
  satisfies $\smash{|I^m\bigtriangleup I_1^*|\le 4\gamma t_1}$ 
  and thus $I^m\subseteq N(I_1^*)$. 
As $N(I_1^*)\cap D^{(m)}=\varnothing$ 
  we have $\smash{\ssy^{(m)}_{Q_1^m}=\ssx_{I^m}}$ 
  (and $\source^{(m)}$ induces a bijection between $Q_1^m$ and $I^m$)
  for each $m\in G$.
Moreover, $$ {\left|\bigcap_{m\in G} I^m\right|\ge  
(1-2\cdot 4\gamma) t_1\ge 0.9t_1},$$
which implies the completeness part: $\ssw:= \ssx_{\cap_{m\in G} I^m}$
  appears as a subword
  in at least $95\%$ of $\smash{\ssy^{(m)}_{Q^m_1}}$ (i.e., every $m\in G$)
  and has length at least $0.9t_1$.
  
Finally, we prove the soundness part: Assume that $\ssw$ is a string of length at least $0.9t_1$
  and appears as a subword in at least $95\%$ of $\smash{\ssy^{(m)}_{Q^m_1}}$, $m=1,\dots,M$.
Then at least $90\%$ of $m\in [M]$ have~$m\in G$ and 
  contain~$\ssw$~as~a subword.
Let $G'\subseteq G$ denote the set of such $m\in G$.
It follows from (iv)  
  that $\source^{(m)}(\ell^{(m)})$ are the same for all $m\in G'$, which consists
  of at least $90\%$ of $[M]$.
To prove (\ref{consensus2}), we take any $m\in G'$ and have that $\source^{(m)}(\ell^{(m)})$ is at least $0.9t_1$
  away from
  the right end of $I^m$; on the other hand,
  $\source^{*}(\ell^*)$ is no more than $H+\gamma t_1$ away from the right end of $I_1^*$, using 
  $|I_1^*|< (1+\gamma)t_1$.
Given that the right ends of $I^m$ and $I_1^*$ differ by no more than $4\gamma t_1$,
  we have  
$
\source^{(m)}(\ell^{(m)})\le \source^*(\ell^*)
$. 
Similarly, $\source^{(m)}(\ell^{(m)})$ is at least as large as the left end of $I^m$
  but $\source^*(\ell^*)$ is similarly no more than $H+\gamma t_1$ away from the left end of $I_1^*$.
Given that their left ends differ by no more than $4\gamma t_1$, 
  we have that $\source^{(m)}(\ell^{(m)})-\source^*(\ell^*)\le 4\gamma t_1+H+\gamma t_1$, which is at most $ 2H$ by (\ref{eq:gamma-tau}) and the definition of $t_1 = 2H+1.$
\end{proof}

\section{The \BMA~algorithm and proof of \Cref{thm:bma}} \label{sec:bma}

The goal of this section is to prove \Cref{thm:bma} which is restated below in full detail.
As mentioned in the introduction, the \BMA~algorithm (which stands for Bitwise Majority Alignment) was first described and analyzed in \cite{BKKM04}. 
Recall the two parameters $\delta$ (deletion rate) and $M$ (number of traces)
  that satisfy (\ref{eq:assumptions}) for some sufficiently large 
  constant $K$, and the positive integer 
  $H\le 2\log n$ given in $(\ref{eq:H})$. 
Let us set some parameters: define
\begin{equation} \label{eq:LGR}
L:=8H, \quad G := L/2, \quad \text{~and~}R:=L2^{0.01L}.
\end{equation}
We prove in this section that \BMA reconstructs any source string 
  $\smash{\tsx\in \{0,1\}^R}$ exactly with $M$ traces from $\Del_\delta(\tsx)$ with high probability\footnote{We use $\tsx$ instead of $\mathsf{x}$ for
    the source string because later in the next section the role of $\tsx$ will be played by various different substrings of $\sx$ of length $R$.},
  when $\tsx$ does not contain any ``\emph{long deserts}.''
We now recall the definition of deserts from \cite{CDLSS20lowdeletion}.

\begin{definition}
A string is said to be a 
  \emph{$k$-desert} for some $k\ge 1$ if it is the prefix of $\sss^{\infty}$ for some string $\sss\in \{0,1\}^k$.
We say a string is a \emph{long desert} if it is a $k$-desert of length $L$ for some $k\le G$. 
\end{definition}



The algorithm \BMA is described in \Cref{alg:BMA},
  and its input consists of $M$ traces $\sz^{(1)},\ldots,\sz^{(M)}$ of 
  $\tsx$.\footnote{Note that in $\MAIN$, we run $\BMA$ on $M+1$ strings obtained
    from $\ssy^*,\ssy^{(1)},\ldots,\ssy^{(M)}$ instead of $M$ strings.
  In its analysis, however, we pretend the string from $\ssy^*$ is not present and 
    focus on what happens when running $\BMA$ on the other $M$ strings only.
  This is why we focus on analyzing $\BMA$ running on $M$ strings in this section.} 
We restate the main theorem of this section:

\begin{theorem} [Detailed statement of \Cref{thm:bma}]  \label{thm:bma-easy}
Let $\tsx\in \{0,1\}^R$ be a string that does not contain any
  long desert.
Let $\bssz^{(1)},\ldots,\bssz^{(M)}\sim \Del_\delta(\tsx)$ be independent traces.
With probability at least $1 - 2^{-H}$, 
  \BMA returns exactly $\tsx$ and~in every round $t\in [R]$, the majority is
  reached by at least $90\%$ of the $M$ bits that are pointed to in $\bssz^{(1)},\ldots,\bssz^{(M)}$. 
\end{theorem}

Similar to the previous section, we write $D^{(m)} \subseteq [R]$ to
  denote the set of positions deleted in $\tsx$ to obtain $\sz^{(m)}$,
  and use them to define $\source^{(m)}(\cdot)$; the only difference is that 
  we set $$\source^{(m)}\bigl(|\sz^{(m)}|+\ell\bigr)=R+\ell$$ for every $m\in [M]$ and $\ell\ge 1$
  since the pointer into $\ssz^{(m)}$ may move beyond $\sz^{(m)}$ into the padded $\ast$'s
  (this can be viewed as adding $\ast$'s to the 
  end of the unknown $\tsx$ which are never deleted and are where the $\ast$'s at the end of $\sz^{(m)}$ come from).
As introduced in \Cref{alg:BMA}, let $\curr^{(m)}(t)$ denote the location of the pointer of $\sz^{(m)}$ at the start of the $t$-th step of \BMA.
In addition, let
$$\pos^{(m)}(t) := \source^{(m)} \bigl(\curr^{(m)}(t)\bigr)\quad \text{and}\quad 
\dist^{(m)}(t) := \pos^{(m)}(t) - t.$$
Informally, $\dist^{(m)}(t)$ captures the number of positions in $\tsx$ that the pointer into $\ssz^{(m)}$ has gotten ``ahead of where it should be.’’

To prove \Cref{thm:bma-easy}, it suffices to show that with high probability,  for every $t^*\in [R]$ it holds that
  $\dist^{(m)}(t^*)=0$ for at least $90\%$ of $m\in [M]$.
We will analyze the behavior of $\{\dist^{(m)}(\cdot)\}_{m \in [M]}$  over~random traces $\smash{\bssz^{(1)},\ldots,\bssz^{(M)}\sim\Del_\delta(\tsx)}$. 
The high level goal of the analysis is to show that for each $m\in [M]$,
  the sequence of random variables $\dist^{(m)}(1),\ldots,\dist^{(m)}(R)$ are nonnegative
  and \emph{negatively drifted} (i.e., $\dist^{(m)}(t)$ tends to decrease as $t$ grows).  

Note that $\{\dist^{(m)}(\cdot)\}_{m \in [M]}$ are not independent 
  over $m$ due to correlations from the consensus $\bsw$ they produce together and thus this ensemble of random variables can be difficult to analyze. To ease the~analysis we introduce a new set of random variables denoted by $\smash{\idist^{(m)}(\cdot)}$ for each ${m \in [M]}$. They are identical to $\dist^{(m)}(\cdot)$
with one key difference: in Step~4 of $\smash{\BMA}$, we set $\sw_t = \tsx_t$ instead of the majority of the bits. Note that this makes $\smash{\{\idist^{(m)}(\cdot)\}_{m \in [M]}}$ independent over $m$ and in fact, identically distributed.~Hence, it suffices to analyze any one of them which we denote by $\idist(\cdot)$ over the draw of $\bssz\sim \Del_\delta(\tsx)$; we let $\bD$ denote the set of positions deleted in $\bssz$.
We define $\icurr(\cdot)$ and $\ipos(\cdot)$ similarly. 
  
The following is a key technical lemma.


\begin{lemma} \label{lemma:rw-easy}
For every $t^*\in [R]$, $ \idist (t^*) = 0$ with probability at least $1- 2\delta L $ over $\bssz\sim \Del_\delta(\tsx)$.
\end{lemma}

We first use \Cref{lemma:rw-easy} to prove \Cref{thm:bma-easy}.

\begin{proof}[Proof of \Cref{thm:bma-easy} using \Cref{lemma:rw-easy}.]
Let $\bssz^{(1)},\ldots,\bssz^{(M)}\sim \Del_\delta(\tsx)$.
For each $t \in [R]$, let $E_t$ be the event that 
%
$$\sum_{m\in [M]} \mathbf{1}\left[\idist^{(m)}(t)=0\right] \ge 0.9M$$
for each $t\in [R]$.
When the event $E_t$ holds for every $t\in [R]$, 
  we have $\bssw=\tsx$ by an induction on $t$.
This implies that the two sets of random variables ($\smash{\dist^{(m)}(\cdot)}$
  and $\smash{\dist^{(m)}_{\mathsf{ideal}}(\cdot)}$) are indeed identical,
  which in turn implies for every $t \in [R]$, 
  $\dist^{(m)}(t)=0$ for at least $90\%$ of $m\in [M]$.

As a result it suffices to understand the probability of $E_t$.
Given that these random variables are independent, it follows from \Cref{lemma:rw-easy} and \Cref{eq:H} that for every $t\in [R]$:
  \[
    \Pr\bigl[E_t\bigr] \ge 1 - 2^M\cdot \bigl(2\delta L \bigr)^{0.9M}
    \ge 1-\left(\frac{48\delta M}{K}\log \frac{1}{\delta M}\right)^{\frac{0.9HK}{\log (1/(\delta M))}}
    \ge 1-2^{-0.45 HK}.
  \]
where the last inequality uses $$\frac{48\delta M}{K}\log \left(\frac{1}{\delta M}\right)\le \sqrt{\delta M}$$ which holds when $K$ is sufficiently large.
  It follows from a union bound  (using $R=L2^{0.01L}$) that
$$
\Pr\bigl[E_t \text{~holds for all~}t\in [R]\bigr] \ge 1- R\cdot 2^{-0.45 HK}\ge 1-2^{-H} 
$$
by setting $K$ to be sufficiently large.
This finishes the proof of the theorem.
\end{proof}



\begin{figure}[t]
  \centering
\setstretch{1.2}
  \begin{algorithm}[H]
    \caption{Algorithm {\BMA}}\label{alg:BMA}
		\DontPrintSemicolon
		\SetNoFillComment
		\KwIn{A multiset
		$\{\sz^{(1)},\ldots,\sz^{(M)}\}$
of $M$ strings} 
		\KwOut{Either $\eps$, the empty string, or a string $\ssw\in \{0,1\}^{R}$}
		For each $m\in [M]$, concatenate $R$ many $\ast$'s to the end of $\sz^{(m)}$\\
		Set $t=1$ and $\current^{(m)}(t)=1$ for each $m\in [M]$\\
\While{$t\le R$}{
Set $\sw_t \in \{0,1,\ast\}$ to be the majority of the $M$ symbols $\smash{\sz^{(m)}_{\current^{(m)}(t)}}$, $m\in [M]$\\
For each $m\in [M]$, set 
\[
\current^{(m)}(t+1)=\begin{cases}
\current^{(m)}(t) + 1 & \text{\  if $\sz^{(m)}_{\current^{(m)}(t)} = \sw_t$}\\[-0.3ex]
\current^{(m)}(t) & \text{\ otherwise}\end{cases}
\]

Increment $t$
}  
\Return $\sw$ if $\sw$ does not contain any $*$ (so $\sw\in \{0,1\}^R$) or $\eps$ if $\sw$ contains at least one $*$
\end{algorithm}\caption{The Algorithm $\BMA$}
\label{figg:BMA}
\end{figure}

\subsection{Proof of \Cref{lemma:rw-easy}}

In the rest of the section we prove \Cref{lemma:rw-easy}.
We start with three simple claims (which may also be found in \cite{CDLSS20lowdeletion});
these claims hold for any trace $\sz$ (and deletions $D$):
\begin{claim} \label{claim:dist}
For each $t\in [R]$, letting $b$ be the \emph{(}${\ipos (t)}$\emph{)}-th
  bit of $\tsx$, we have
  \begin{flushleft}\begin{enumerate}
    \item If $\tsx_t\ne b$, then $\idist (t+1)=\idist (t)-1$.
    \item If $\tsx_{t}=b$, then $\idist (t+1)=\idist (t)+\ell$, where 
        $\ell$ is the nonnegative integer with 
\begin{equation}\label{eq:nextbit}
\ipos(t)+1, \ldots, \ipos(t)+\ell\in D\quad\text{and}\quad
        \ipos(t)+\ell+1\notin D.
\end{equation}
  \end{enumerate}\end{flushleft}
\end{claim}

\begin{proof}
The first item follows from the observation that
 $\icurr(t+1) = \icurr (t)$. 

For the second item, we have 
  $\smash{\icurr (t+1)=\icurr (t)+1}$. It
 now points to the next bit in $\sz^{(m)}$, which 
   is the bit of $\tsx$ indexed by $\ipos (t)+\ell+1$ with $\ell$ defined in (\ref{eq:nextbit}).
\end{proof}

We next have the following two easy observations:  

\begin{claim}\label{clm:dist-positivity}
We have $\idist(t) \ge 0$ and $\ipos(t)\le R+1$ for every $t\in [R]$. 
\end{claim}  
\begin{proof}
The proof of the first part is by induction. 
Using Claim~\ref{claim:dist}, $\idist (t+1) \ge \idist (t)-1$. So if  $\idist (t)>0$ then $\idist (t+1) \ge 0$. Otherwise, if $\idist (t)=0$, we can  apply the second item of \Cref{claim:dist} to conclude that $\idist (t+1)\ge 0$. 
For the second part note that once $\ipos(t)$ reaches $R+1$ (i.e., $\icurr(t)$ reaches
  $|\ssy^*|+1$),
  it cannot move further as $\ssx$ is a string in $\{0,1\}$ but 
  the current bit in $\ssy^*$ is $*$.
\end{proof}

\ignore{  \[
    \Pr[ G_\ideal(t) \bigm\mid \bigwedge_{s=1}^{t-1} G_\ideal(s) ] \ge 1 - n'^{2} \cdot (\delta M)^{\Theta(M)}
  \]
  We draw $M$ traces $\by_1, \ldots, \by_M$ from $\Del_\delta(x)$.
  Let $\dist^1(\cdot), \ldots, \dist^M(\cdot)$ denote the $\dist(\cdot)$ of $\by_1, \ldots, \by_M$, respectively.
  For each $t \in [n']$, let $G(t)$ be the indicator of $\sum_{i=1}^M \dist^i(t) > M/2$.
  We prove by induction on $t$ that $\bigwedge_{s=1}^t G(s)$ and (therefore $w_{[1:t]} = x_{[1:t]}$) with probability $1 - t \cdot n'^{2} \cdot (\delta M)^{\Theta(M)}$, from which the theorem follows immediately.
  The base case follows from the fact that $G(1)$ and $G_\ideal(1)$ are identical, as the random variables $\dist^{(j)}(1)$ (for any choice of $j$) are independent. 

  For the inductive step, we have
  \begin{align*}
    \Pr\Bigl[ \bigwedge_{s=1}^t G(s) \Bigr]
    &\ge \Pr\Bigl[ G(t) \bigm\mid \bigwedge_{s=1}^{t-1} G(s) \Bigr] \Pr\Bigl[ \bigwedge_{s=1}^{t-1} G(s) \Bigr] \\
    &\ge \Pr\Bigl[ G(t) \bigm\mid \bigwedge_{s=1}^{t-1} G(s) \Bigr] - \frac{t-1}{n^4} \\
    &= \Pr\Bigl[ G(t) \bigm\mid \bigwedge_{s=1}^{t-1} G(s) \wedge \big( w_{[1:t-1]} = x_{[1:t-1]} \bigr) \Bigr] - (t-1) \cdot n'^{2} \cdot (\delta M)^{\Theta(M)} \\
    &= \Pr\Bigl[ G_\ideal(t) \bigm\mid \bigwedge_{s=1}^{t-1} G_\ideal(s) \Bigr] -(t-1) \cdot n'^{2} \cdot (\delta M)^{\Theta(M)}\\
    &\ge 1 - t \cdot n'^{2} \cdot (\delta M)^{\Theta(M)}.
  \end{align*}
  The equalities are because when $w_{1:t-1} = x_{1:t}$, then each process $\{\dist^i(s)\}_{s=1}^t$ is identical to $\{\dist_\ideal^i(s)\}_{s=1}^t$.
  However, note that conditioning on $\bigwedge_{s=1}^{t-1} G_\ideal(s)$ removes the independence across the $\dist_\ideal^i(t)$. Now, finally plugging in the value of $n'$ and $t=n'$, we have that $$\Pr\Bigl[ \bigwedge_{s=1}^{n'} G(s) \Bigr] \ge 1 -  n'^{3} \cdot (\delta M)^{\Theta(M)} \ge 1 - (C\delta M)^{c M}. $$
  Now, note that if $\bigwedge_{s=1}^{n'} G(s)$, then it follows that for every $1 \le s \le n'$, $\sum_{i=1}^M \dist^{s} (i) \le M/2$. This immediately implies that $w_{[1:n']} = x_{[1:n']}$. 
\end{proof}}


\begin{claim} \label{claim:dist-desert}
  If $\idist(t) = \cdots = \idist(t+L) = k$ for some $k \in [G]$, then $\tsx$ has a 
  long desert.
\end{claim}
\begin{proof}
{We have that $\tilde{\ssx}_{t+\ell}$ is equal to the $\pos_{\ideal}(t+\ell)$-th bit of $\tilde{\ssx}$ for all $\ell \in [0:L-1]$, and hence} using \Cref{claim:dist}, we have  
$
\tsx_{t+\ell} = \tsx_{t+k+\ell}  
$
for all $\ell\in [0:L-1]$.
The claim follows. 
\end{proof}

We now start to prove \Cref{lemma:rw-easy}.
Let $t^*\in [R]$ be the round we consider in  \Cref{lemma:rw-easy}, with
  $t^*=t_0+ sL$ such that $t_0\in \{0, 1,$ $\ldots,L-1\}$ and $s\le 2^{0.01L}$.
We observe that 
\begin{equation} \label{eq:seq-of-dists}
\idist(t_0), \hspace{0.06cm}\idist(t_0+L),\hspace{0.06cm}\ldots,\hspace{0.06cm}\idist(t_0+ sL)
\end{equation} is a Markov process and 
  look at how $\idist(t_0+ (\ell+1)L)$ changes conditioned on $\idist(t_0+ \ell L)$.
  
To this end, let us condition on $\idist(t_0+ \ell L)=\Delta\ge 0$, so $\ipos(t_0+ \ell L)=t_0+ \ell L+\Delta$.
Conditioning on this, each bit  of $\tsx$
  after $t_0+\ell L+\Delta$ is deleted independently and added
  to $\bD$ with probability $\delta$.
Let $\bbeta$ be the nonnegative random variable such that \emph{either}
\begin{align*}
&t_0+ \ell L+\Delta +\bbeta \text{~is at most~}R \text{~and does not belong to~}\bD \quad\text{and}\\
&
\bigl\{t_0+ \ell L+\Delta+ 1,\ldots,t_0+ \ell L+\Delta+\bbeta\bigr\}\cap \overline{\bD}=  L ,
\end{align*}
i.e., $t_0+ \ell L+\Delta+\bbeta$ is the unique location in $\ssx$ that is the $L$-th undeleted position after $t_0+ \ell L+\Delta$,
\emph{or}
   $\bbeta$ is chosen using $t_0+\ell L+\Delta +\bbeta=R+1$ if no such $\bbeta$ exists.
Since $\ipos(\cdot)$ can move forward by at most $L$ undeleted positions in steps $t_0 + \ell L
+ 1, \dots, t_0 + (\ell + 1)L$, 
  it follows from the second part of \Cref{clm:dist-positivity} that
$$
\ipos\bigl(t_0+(\ell+1)L\bigr)\le t_0+\ell L+\Delta +\bbeta$$ and thus, we have
  the following upper bound:
$$
\idist\bigl(t_0+ (\ell+1)L\bigr)\le \idist\bigl(t_0+ \ell L\bigr) +\bbeta-L.
$$
Moreover, 
when $ \Delta\le G$ and $\bbeta\le L$ (including $\bbeta=L$), we have 
$$
\idist\bigl(t_0+(\ell+1)L\bigr)\le \max\Bigl(\idist\bigl(t_0+\ell L\bigr) -1,0\Bigr).
$$
It holds for $\bbeta=L$ because otherwise by \Cref{claim:dist-desert}, the subword of $\tsx$ in ${[t_0+\ell L:t_0+(\ell+1)L-1]}$ would be a long desert,
  contradicting with our assumption that $\tsx$ has no long deserts.

 {In the next portion of the analysis we relate this Markov process to a simpler one for which the transition probabilities are the same for all states.}  Note that 
  $\bbeta$ is distributed as $$\min\bigl(R+1-(t_0+\ell L+\Delta),\bbeta^*\bigr)$$
  where $\bbeta^*$ denotes the sum of $L$ i.i.d. geometric random variables with success probability $1-\delta$.\footnote{We use the version of the geometric distribution {for which the outcome of a draw is the total number of trials up to and including the first success (hence the support is $\{1,2,3,\dots\}$).}}~So $\bbeta$ is 
  stochastically dominated by $\bbeta^*$.\footnote{Recall that a random 
  variable $\bX$ is said to stochastically dominate $\bY$ if 
$\Pr[\bX \ge a] \ge \Pr[\bY \ge a]$ for all $a$.} 
Let $\alpha:=\delta L$, which can be made sufficiently small as $\delta L \le 8/K$ and $K$ can be made sufficiently large.
We use the following rough estimates for the probability of $\bbeta^*=L+c$ for $c=0,1,\ldots$:
When $c=0$ we have 
\begin{equation}\label{hehehehehe11}
\Pr\big[\bbeta^*=L\big]=(1-\delta)^L\ge 1-\delta L=1-\alpha> \frac{1-2\alpha}{1-\alpha}.
\end{equation}
(The reason for using the lower bound $(1-2\alpha)/(1-\alpha)$ in the second inequality will become clear soon.)
When $c\ge 1$,
\begin{equation}\label{hehehehehe22}
  \Pr\big[ \bbeta^* = L+c \big]
  = 
    \binom{L+c-1}{c} \delta^c (1-\delta)^L \le \big(L\delta \big)^c = \alpha^c.
\end{equation}
Inspired by these estimates, we introduce the following simpler Markov chain
  $\bX_0,\bX_1,\ldots,\bX_s\ge 0$, where (1) $\bX_0$ is distributed the same as $\idist(t_0)$, and (2)
  for each $\bX_{\ell+1}$, if $\bX_{\ell}> G$, then 
\begin{equation}\label{transition1}
\bX_{\ell+1}=\begin{cases}\bX_\ell & \text{with probability $(1-2\alpha)/(1-\alpha)$}\\
\bX_\ell+c & \text{with probability $\alpha^c$ for each $c\ge 1$}\end{cases};
\end{equation}
if $\bX_\ell\le G$, then
\begin{equation}\label{transition2}
\bX_{\ell+1}=\begin{cases}\max(\bX_\ell-1,0) & \text{with probability $(1-2\alpha)/(1-\alpha)$}\\
\bX_\ell+c & \text{with probability $\alpha^c$ for each $c\ge 1$}\end{cases}.
\end{equation}
Note that the use of $(1-2\alpha)/(1-\alpha)$ makes sure that the probabilities sum to $1$.
Below we will analyze $\bX_0, \bX_1, . . . , \bX_s$
in lieu of \Cref{eq:seq-of-dists}.

\Cref{lemma:rw-easy} follows directly by combining the following two claims:

\begin{claim}\label{claimhehe1}
 $\Pr\bigl[ \idist(t^*) \ge c \bigr] \ge \Pr\bigl[ \bX_s \ge c \bigr]$ for every $c$.
\end{claim}

\begin{claim}\label{claimhehe2}
 $\Pr\big[\bX_s=0\big]\ge 1-2\alpha.$
\end{claim}

\begin{proof}[Proof of \Cref{claimhehe1}]
We prove by induction that for every $\ell\in [0:s]$, $\bX_\ell$ stochastically dominates $\idist(t_0+\ell L)$.
The basis is trivial since $\bX_0$ has the same distribution as $\idist(t_0)$.

To prove the case with $\bX_{\ell+1}$ using $\bX_\ell$, we make two simple observations:
\begin{flushleft}\begin{enumerate}
\item First, for any $a \ge b$,
  the distribution of $\bX_{\ell+1}$ conditioned on $\bX_{\ell}=a$ stochastically dominates
    the distribution of $\bX_{\ell+1}$ conditioned on $\bX_{\ell}=b$.
\item Next for any $a$, it follows from \Cref{hehehehehe11,hehehehehe22} that the distribution of $\bX_{\ell+1}$ conditioned on $\bX_\ell=a$
  stochastically dominates that of $\idist(t_0+(\ell+1)L)$ conditioned on $\idist(t_0+\ell L)=a$.

\end{enumerate}\end{flushleft}
It follows from these two observations, as well as the inductive hypothesis on $\ell$,
  that $\bX_{\ell+1}$ stochastically dominates $\idist(t_0+(\ell+1)H)$.
This finishes the proof of the claim.
\end{proof}

\begin{proof} [Proof of \Cref{claimhehe2}]
We prove by induction on $\ell=0,1,\ldots,s$ that the distribution of $\bX_\ell$ satisfies
\begin{align}
\label{eqeqeq111} \Pr\bigl[\bX_\ell = c\bigr] &\le \frac{\alpha}{1-2\alpha} \cdot \left(\frac{2\alpha(1-\alpha)}{1-2\alpha}\right)^{c-1},\qquad\text{for every $c\in [G]$};\\[0.6ex]
\label{eqeqeq222} \Pr\bigl[\bX_\ell > G\bigr] &\le (\ell+1)\cdot  \frac{\alpha}{1-\alpha} \left(\frac{2\alpha(1-\alpha)}{1-2\alpha}\right)^{G}. 
\end{align}
Before working on the induction, we have from these two items that
  \begin{align*}
    \Pr\bigl[ \bX_s = 0\bigr]
    &\ge 1 - \frac{\alpha}{1-2\alpha}\cdot  \sum_{a \ge 0} \left( \frac{2 \alpha (1-\alpha)}{1 - 2\alpha} \right)^a - (s+1) \cdot \frac{\alpha}{1-\alpha} \left( \frac{2\alpha(1-\alpha)}{1-2\alpha} \right)^{G}  \\[0.2ex]
    &\ge 1 - \frac{\alpha}{1-4\alpha+2\alpha^2} - \frac{\alpha}{2} \\[0.5ex]
    &\ge 1-2\alpha.
  \end{align*}
  The second inequality 
    used $s\le 2^{0.01L}$, $G=L/2$ and the fact that $\alpha$ can be made sufficiently small.
  The last inequality also used that $\alpha$ is sufficiently small.
We work on the induction below.
  
For the base case $\bX_0$, recall that this random variable has the same distribution as $\idist(t_0)$, and hence we have for each $c\ge 1$,
  \[
    \Pr\bigl[\idist (t_0) = c\bigr]
    \le \sum_{a\ge c} \binom{t_0+c-1}{a} \delta^a 
    \le \sum_{a\ge c} (L\delta)^a
    \le \frac{\alpha^c}{1-2\alpha} \cdot \left(\frac{2(1-\alpha)}{1-2\alpha}\right)^{c-1} ,
  \]
 (where the second inequality above uses $t_0 \leq L-1$ and the third uses that $\alpha$ is sufficiently small) and also
  \[
    \Pr\bigl[\idist (t_0) > G\bigr]
    \le \sum_{c\ge G+1}\sum_{a\ge c} \binom{t_0+c-1}{a} \delta^a  
    \le \frac{\alpha^{G+1}}{(1-\alpha)^2}
    \le \frac{\alpha}{1-\alpha} \left(\frac{2\alpha(1-\alpha)}{1-2\alpha}\right)^{G} .
  \]

For the induction, we assume the statement holds for $\ell$ and use it to prove the case with $\ell+1$.
  Using \Cref{eqeqeq111}, for every $c\in[G-1]$, we have
  \begin{align}\nonumber
    \Pr\bigl[\bX_{\ell+1}=c\bigr]
    &\le \sum_{a=0}^{c-1} \alpha^{c-a}\cdot \Pr\bigl[\bX_\ell=a\bigr] + \frac{1-2\alpha}{1-\alpha}\cdot \Pr\bigl[\bX_\ell=c+1\bigr]
    \tag{by \Cref{transition2}}
    \\ 
\nonumber    &\le \alpha^c + \sum_{a=1}^{c-1} \alpha^{c-a} \frac{\alpha}{1-2\alpha} \left(\frac{2\alpha(1-\alpha)}{1-2\alpha}\right)^{a-1} + \frac{1-2\alpha}{1-\alpha} \frac{\alpha}{1-2\alpha} \left(\frac{2\alpha(1-\alpha)}{1-2\alpha}\right)^c \\
\label{hehehe333}    &= \alpha^c + \frac{\alpha^c}{1-2\alpha} \sum_{a=1}^{c-1} \left( \frac{2 (1-\alpha)}{1 - 2\alpha} \right)^{a-1} + \frac{\alpha}{1-\alpha} \left(\frac{2\alpha(1-\alpha)}{1-2\alpha}\right)^c \\
\nonumber    &= \alpha^c \left(1 + \frac{1}{1-2\alpha} \frac{ \left(\frac{2(1-\alpha)}{1-2\alpha}\right)^{c-1} - 1}{ \frac{2(1-\alpha)}{1-2\alpha} - 1} + \frac{\alpha}{1-\alpha} \left( \frac{2 (1-\alpha)}{1 - 2\alpha} \right)^c \right) \\
\label{hehehe444}   &= \alpha^c \left( \left(\frac{2 (1-\alpha)}{1-2\alpha}\right)^{c-1} + \frac{\alpha}{1-\alpha} \left(\frac{2(1-\alpha)}{1-2\alpha}\right)^c \right) \\
\nonumber    &= \frac{\alpha}{1-2\alpha} \left(\frac{2\alpha(1-\alpha)}{1-2\alpha}\right)^{c-1} . 
  \end{align}
The proof for $c=G$ is similar except that we do not have the term for $\bX_\ell=G+1$.

Finally using \Cref{eqeqeq111,eqeqeq222} we have  
  \begin{align}
\Pr\bigl[\bX_{\ell+1}>G\bigr] 
\nonumber    &\le \Pr\bigl[\bX_{\ell}>G\bigr] + \sum_{b>G} \sum_{a=0}^{G} \alpha^{b-a} \cdot \Pr\bigl[\bX_\ell=a\bigr] 
\tag{by definition of $\bX_{\cdot}$}\\
\nonumber    &\le \Pr\bigl[\bX_{\ell}>G\bigr] + \frac{1}{1-\alpha} \sum_{a=0}^{G} \alpha^{G+1-a}\cdot  \Pr\bigl[\bX_\ell=a\bigr] \\
\nonumber    &\le \Pr\bigl[\bX_{\ell}>G\bigr] + \frac{\alpha^{G+1}}{1-\alpha} \left(1+\frac{1}{1-2\alpha}\sum_{a=1}^{G} 
    \left(\frac{2 (1-\alpha)}{1-2\alpha}\right)^{a-1}\right) \tag{by \Cref{eqeqeq111}}\\[0.1ex]
\label{hehehe555}    &= \Pr\bigl[\bX_{\ell}>G\bigr] + \frac{\alpha}{1-\alpha} \left(\frac{2\alpha(1-\alpha)}{1-2\alpha}\right)^{G}
\\[0.3ex]
\nonumber &\le (\ell+2)\cdot \frac{\alpha}{1-\alpha}\left(\frac{2\alpha(1-\alpha)}{1-2\alpha}\right)^{G},
\tag{by \Cref{eqeqeq222}}
  \end{align}
where \Cref{hehehe555} used similar derivation between \Cref{hehehe333,hehehe444} earlier.
This finishes the induction and the proof of the claim. \end{proof} 

\ignore{
\begin{proof}[Proof of \Cref{claim:distvsdistupper}]
  For $1 \le c \le K-1$,
  \begin{align*}
    \Pi^1(c)
    &= \sum_{j \ge c} \pi^1(j) \\
    &\le \sum_{j \ge c} \left( \sum_{i=0}^{j-1} \pi^0(i) \Pr[\bG = j-i] + \pi^0(j+1) \Pr[\bG = 0] \right) \\
    &= \sum_{i=0}^{c-1} \pi^0(i) \sum_{j \ge c} \Pr[\bG = j-i] 
      + \pi^0(c) \sum_{j \ge c+1} \Pr[\bG = j-i] 
      + \sum_{c+1 \le i \le k-1} \pi^0(i) \sum_{j \ge c+1} \Pr[\bG = j-i] 
      + \sum_{i \ge k} \pi^0(i) \\
    &= \sum_{i=0}^{c-1} \pi^0(i) \sum_{j \ge c} \Pr[\bG = j-i]
      + \pi^0(c) \sum_{j \ge c+1} \Pr[\bG = j-i]  + \sum_{i \ge c+1} \pi^0(i)  \\
    &\le \sum_{i=0}^{c-1} \pi^0(i) \Bigl( \sum_{j \ge c+1} \alpha^{j-i} \Bigr) + \pi^0(c) \sum_{j \ge c+1} \alpha^{j-i} + \sum_{i \ge c+1} \pi^0(i) \\
    &= \sum_{i=0}^{c-1} \frac{\alpha^{c - i}}{1 - \alpha} \pi^0(i) + \frac{\alpha}{1-\alpha} \pi^0(c) + \sum_{i \ge c+1} \pi^0(i) \\
    &\le \sum_{i=0}^{c-1} \frac{\alpha^{c-i} - \alpha^{c-i+1}}{1-\alpha} \Pi^0(i) + \left(1 - \frac{\alpha}{1-\alpha}\right) \Pi^0(c+1) \\
    &\le \sum_{i=0}^{c-1} \frac{\alpha^{c-i} - \alpha^{c-i+1}}{1-\alpha} \Pi_\upper^0(i) + \left(1 - \frac{\alpha}{1-\alpha}\right) \Pi_\upper^0(c+1) \\
    &= \sum_{i=0}^{c-1} \pi_\upper^0(i) \sum_{j \ge c+1} \alpha^{j-i} + \pi_\upper^0(c) \sum_{j \ge c+1} \alpha^{j-i} + \sum_{i \ge c+1} \pi_\upper^0(i) \\
    &=  \Pi_\upper^1(c)
  \end{align*}
  By a similar calculation, we have $\Pi^1(K) \le \Pi_\upper^1(K)$.
\end{proof} }


\section{Main Algorithm}\label{sec:combine}


\begin{figure}[t!]
  \centering
\setstretch{1.2}
  \begin{algorithm}[H]
    \caption{{\MAIN}}\label{fig:BMA}
		\DontPrintSemicolon
		\SetNoFillComment
		\KwIn{A positive integer $n$ and a tuple of $(M+1)$ strings $\ssy^*,\ssy^{(1)},\ldots,\ssy^{(M)}$  
		}
		\KwOut{
		A binary string $\ssw$}
		Set $\ell^*=5\tau\log n$ and $\ssw=\epsilon$\\
		\While{$\ell^*\le |\ssy^*|-R$ and $\ell^*\le |\ssy^*|-5\tau\log n$}{
		  Run $\Align\hspace{0.04cm}(\ell^*,\ssy^*,\ssy^{(1)},\ldots,\ssy^{(M)})$ to obtain a 
		  tuple of locations $(\ell^{(1)},\ldots,\ell^{(M)})$\\
		  Run $\smash{\BMA\hspace{0.04cm}(\ssy^*,\ssy^{(1)},\ldots,\ssy^{(M)};\ell^*,\ell^{(1)},\ldots,\ell^{(m)})}$ to obtain 
		    a binary string ($\eps$ or in $\{0,1\}^R$)\\
		  Concatenate the string returned by $\BMA$ to the end of $\ssw$\\
   		  Set $\ell^*$ to be the final pointer of $\ssy^*$ in the run of 
		  $\BMA$ above  and increment it\\		  
		  }
\Return $\ssw$.
\end{algorithm}\caption{The Algorithm $\MAIN$}\label{figg:MAIN}
\end{figure}

Let $\delta$ and $M$ be two parameters that satisfy \Cref{eq:assumptions,eq:H}.
Recall the following parameters used in our analysis of $\Align$ and $\BMA$:
$$
\tau=500,\quad
H =\frac{M}{K}\log \left(\frac{1}{\delta M}\right)\le {2\log n},\quad L =8H,\quad G =L/2=4H\quad\text{and}\quad
R =L2^{0.01L}.
$$
(Note that $L \leq 16 \log n,$ $G \leq 8 \log n$ and $R \le O(n^{0.16} \log n).$)
Our main (deterministic) algorithm \MAIN\ is described in \Cref{figg:MAIN}, where we 
  use
$$
\BMA\big(\ssy^*,\ssy^{(1)},\ldots,\ssy^{(M)};\ell^*,\ell^{(1)},\ldots,\ell^{(M)}\big)
$$
to denote running $\BMA$ on the suffix of $\ssy^*$ starting at location $\ell^*$ and the suffix of 
  each $\ssy^{(m)}$ starting at location $\ell^{(m)}$.
Recall that $\BMA$ either returns the empty string $\eps$ or a string $\ssw\in \{0,1\}^R$.

We prove the following theorem about the performance of $\MAIN$:

\begin{theorem}\label{theo:mainReconstruct}
Let $\bssx\sim \{0,1\}^n$ and $\bssy^*,\bssy^{(1)},\ldots,\bssy^{(M)}\sim \Del_\delta(x)$.
With probability at least $1-1/n$, $\MAIN$ on $\bssy^*,\bssy^{(1)},\ldots,\bssy^{(M)}$
  returns a string $\bssw$ with edit distance at most $2^{-0.01 H}n$ from $\bssx$.
\end{theorem}

Similar to the analysis of $\Align$ in \Cref{sec:align}, we divide the analysis of 
  $\MAIN$ into two parts:
In \Cref{mainanalysispart1}, we begin by describing some good events over the randomness of $\bssx$,
  $\bD^*$ and $\bD^{(m)}: m \in [M]$,
  and show that these events happen with probability at least $1-1/n$.
The rest of the analysis in \Cref{mainanalysispart2} will be entirely deterministic.
We show that $\MAIN$ on $\ssy^*,\ssy^{(1)},\ldots,\ssy^{(M)}$ must 
  return a string $\ssw$ with small edit distance from $\ssx$ when all the events described
  in \Cref{mainanalysispart1} hold.

\subsection{Probabilistic Analysis}\label{mainanalysispart1}

We start by showing that for $\bssx\sim \{0,1\}^n$, the number of length-${(2R)}$ subwords
of $\bssx$ that 
  contain at least one long desert is small. 

\begin{lemma}\label{lem:fewlongdeserts}
With probability at least $1-\exp(-n^{0.1})$ over $\bssx\sim \{0,1\}^n$,
the number of $i\in [n- {2R}+1]$ such that $\bssx_{[i:i+ {2R}-1]}$ has at least one long desert
  is at most $ 2^{-0.13H}n$.
\end{lemma}
\begin{proof}
Fix any $k\le G$.
The probability of a random length-$L$ string being a $k$-desert is at most
$$
2^k\cdot \frac{1}{2^L}.
$$
As a result, the probability of a random length-$L$ string being a long desert is at most
$$
\sum_{k=1}^{G} \frac{2^k}{2^L}\le 2^{-0.4 L}.
$$
Consider the number of length-$L$ subwords in $\bssx\sim \{0,1\}^n$ that are long deserts.
Given that changing each bit can only change the number by no more than $L$,
  it follows from McDiarmid's inequality (\Cref{thm:Mcdiarmid}) that with probability at least $1-\exp(-n^{0.1})$,
  this number is at most
$$
2^{-0.4L}n+O(n^{0.55}L)\le 2^{-0.22 H}n
$$
using $H\le 2\log n$.
When this happens, the number of indices $i$ we care about in the statement of the lemma
is at most $ {(4R-L)} \cdot 2^{-0.22H}n\le 2^{-0.13H}n$ using $R=L2^{0.01L}$ and $L=8H$.
\end{proof}

We modify the definition of $\image^{(m)}$ so that it is well-defined for every $i\in [n]$: $\image^{(m)}(i)$ is the smallest location
$\ell\in [|\ssy^{(m)}|]$ such that $\source^{(m)}(\ell)\ge i$, or set $\image^{(m)}(i)=|\ssy^{(m)}|+1$ if 
  no such $\ell$ exists. 
(Note that when the latter happens, the suffix of $\ssy^{(m)}$ starting at $|\ssy^{(m)}|+1$ is the empty string $\eps$.)
We say $\BMA$ \emph{succeeds on $\ssy^{(1)},\ldots,\ssy^{(M)}$ at location $i\in [n-R+1]$ of $\ssx$}
  if running $\BMA$ on $\ssy^{(1)},\ldots,\ssy^{(M)}$~starting at $\image^{(1)}(i),\ldots,
  \image^{(M)}(i)$ returns exactly the $R$-bit string $\ssx_{[i:i+R-1]}$ and moreover,
  the consensus is achieved by at least $90\%$ of strings in every round of \BMA's execution.

The following lemma is a direct corollary of \Cref{lem:fewlongdeserts} and \Cref{thm:bma-easy}:
  
\begin{lemma}\label{keykeyBMAlemma}
Let $\bssx\sim \{0,1\}^n$ and $\bssy^{(1)},\ldots,\bssy^{(M)}\sim \Del_\delta(\bssx)$.
With probability at least $1-2\exp(-n^{0.1})$, $\BMA$ succeeds on all but at most 
  $2^{-0.1H}n$ locations $i\in [n-R+1]$ in $\bssx$.
\end{lemma}
\begin{proof}
It follows from \Cref{lem:fewlongdeserts} that with probability at least $1-\exp(-n^{0.1})$, a random
$\bssx\sim \{0,1\}^n$ has no more than $2^{-0.13H}n$ length-${(2R)}$ subwords that contain at least
  one long desert.
Let $\ssx$ be such a string and
fix any location $i$ such that $\smash{x_{[i:i+R-1]}}$ contains no long deserts.
If the conclusion of \Cref{thm:bma-easy} holds on $\tsx:=\smash{\ssx_{[i:i+R-1]}}$ over subwords of $\bssy^{(m)}$ that originate from $\smash{\ssx_{[i:i+R-1]}}$, then
  $\BMA$ must succeed at location $i$.
It follows from \Cref{thm:bma-easy} that this happens with probability at least $1-2^{-H}$.

We will apply McDiarmid's inequality.
Note that each of the $nM$ independent random variables (each of which indicates whether or not a bit of $\ssx$ is included in $\bD^{(m)}$)
  can only change the number we care about (i.e. the number of locations $i$ such that 
  the conclusion of \Cref{thm:bma-easy} holds  on $\tsx:=\ssx_{[i:i+R-1]}$  over subwords of $\smash{\bssy^{(m)}}$ that originate from $\ssx_{[i:i+R-1]}$)
  by no more than $R$.~It~follows that with probability at least $\smash{1-\exp(-n^{0.1})}$,
  $\BMA$ succeeds on all except $\smash{2^{-0.1H}n}$ many locations $i$ in $\ssx$ such that $\ssx_{[i:i+R-1]}$ has no long deserts.
The lemma follows.
\end{proof}

\begin{lemma}\label{skmplem}
  $\bD^*\sim \calD$  satisfies the following two properties with 
  probability at least $1-1/n^2$: \textup{(}note that $\abs{\bssy^\ast} = n - \abs{\bD^\ast}$\textup{)}
\begin{flushleft}\begin{enumerate}
\item $\source^*(5\tau \log n)=O(\log n)$ and 
  $$\min\big(\source^*(\abs{\bssy^\ast}-R),\, \source^*(\abs{\bssy^\ast}-5\tau \log n)\big)\ge n-(2R+O(\log n));$$
\item There are at most  $ 2^{-0.2H}n$ values of $i\in [n]$ such that $\source^*(\ell)<i<\source^*(\ell+1)$
  for some $\ell\in [\abs{\bssy^\ast}-1]$ for which 
  $\source^*(\ell+1)-\source^*(\ell)\ge 2H$. 
\end{enumerate}
\end{flushleft}
\end{lemma}
\begin{proof}
  The first part follows from a Chernoff bound.

For the second part, note that for $i$ to be counted, it must be the case that 
  either $\smash{[i- H+1:i]}$ $\smash{\subseteq \bD^*}$ or $\smash{[i:i+ H-1]\subseteq \bD^*}$, which occurs
  with probability at most $2\cdot \delta^{H}\le 2^{-H}$.
The second part then follows from an application of McDiarmid's inequality.
\end{proof}

To describe our final condition on $\bssx\sim \{0,1\}^n$ and $\bD^*\sim \calD$,
  we introduce a procedure which we call $\SBMA$. $\SBMA $ takes as input a string $\ssy^*$, a location $\ell^*$ such that $|\ssy^*|\ge \ell^*+R$,  
  and 
  a reference string $\ssz\in \{0,1\}^R$.
  Let $\curr^*(1)=\ell^*$.
The $\SBMA$ procedure repeats the following for $R$ rounds:
In the $t$-th round, $t\in [R]$,
  $\SBMA$ compares $\ssz_t$ with $\smash{\ssy^*_{\curr^*(t)}}$ and 
  set $\curr^*(t+1)=\curr^*(t)+1$ if they match and 
  {$\curr^*(t+1)=\curr^*(t)$ if they do not match.}
  {For each $t$ we also define $\pos^*(t)= \source^* \bigl(\curr^*(t)\bigr).$}
  \ignore{
}
After the final ($R$-th) round $\SBMA $ outputs 
\[
\pos^*(R+1) :=
\source^* \bigl(\curr^*(R+1)\bigr).
\]

 {Intuitively, $\SBMA$ outputs the final location of the pointer $\ell^\ast$ after a successful run of \BMA.}
Now we state the final condition on $( {\bssx,\bD^*})$.
We say $\BMA^*$ \emph{succeeds on a source string $\ssx\in \{0,1\}^n$ with 
  respect to $\ssy^*$} if for all but at most $2^{-0.1H}n$ many $\ell^*\in [|\ssy^*|-R ]$, we have 
\begin{equation}\label{haheeq}
\SBMA \big(\ssy^*,\ell^*,\ssx_{[i:i+R-1]}\big) \le  i+R +G,\quad\text{for every $i\in [\source^*(\ell^*)-2H:\source^*(\ell^*)]$.}
\end{equation}
Note that from an argument similar to the proof of \Cref{clm:dist-positivity} we always have 
\begin{equation}\label{hahehaheeq}
\SBMA \big(\ssy^*,\ell^*,\ssx_{[i:i+R-1]}\big)\ge i+R,
\end{equation}
so intuitively, $\SBMA$ succeeds on $\ssx$ with respect to $\ssy^*$ if for almost every $\ell^\ast$, the output of $\SBMA$ on input 
$(\ssy^*,\ell^*,\ssx_{[i:i+R-1]})$ is close to the ``right value'' $i + R$.

\begin{lemma}\label{difflemma}
With probability at least $1-2\exp(-n^{0.1})$ over $\bssx\sim \{0,1\}^n$ and $\bD^*\sim \calD$, $\BMA^*$ succeeds on $\bssx$ with respect
 to $\bssy^*$.
\end{lemma}
\begin{proof}
We will show that with high probability, the number of pairs $\ell^*$ and $i\in [\source^*(\ell^*)-2H:\source^*(\ell^*)]$
  that violate \Cref{haheeq} is at most $2^{-0.1H}n.$
First, note that it follows from \Cref{lem:fewlongdeserts} that with probability at least $\smash{1-\exp(-n^{0.1})}$,
the number of length-$ {(2R)}$ subwords of $\bssx\sim \{0,1\}^n$ that have at least one long desert is 
  at most $\smash{2^{-0.13H}n}$.
Fixing such an $\ssx\in \{0,1\}^n$ in the rest of the proof, we show that with
  probability at least $1-\exp(-n^{0.1})$ over $\bD^*\sim \calD$,
  $\BMA^*$ succeeds on $\ssx$ with respect to $\bssy^*$, from which the lemma follows.


To this end, we consider an $\ell^*$ and an $i$ in the window such that 
  $\ssx_{[i:i+2R-1]}$ has no long deserts. (By doing this we skipped no more than $2H\cdot 2^{-0.13H}n$ many pairs, which is much smaller than our target of $2^{-0.1H}n$.)
The idea of the argument is to upper-bound the probability that~\Cref{haheeq} is violated at $\ell^*$ and $i$ over $\bD^*\sim \calD$, and then apply McDiarmid's inequality to finish the proof.

We note that whether \Cref{haheeq} holds or not only depends
  on the $[i:i+R+G]$ window of $\ssx$, because if $\pos^*(t)$ ever becomes
  larger than $i+R+G$ then \Cref{haheeq} is already violated.
Since $R+G<2R$, this subword $\ssx_{[i:i+R+G]}$ of $\ssx$ has no long deserts.
Similar to the analysis of $\BMA$, we define $\dist^*(t)$ for each $t= 1,\ldots,R,R+1$ as
\[
  \dist^*(t)
  := \pos^*(t) -t - (i-1),
\]
where $\pos^*(t)$ is from the execution of $\SBMA$.
Note that 
\[
  \dist^\ast(1) = \pos^\ast(1) - i = \source^\ast(\ell^\ast) - i \in [0,2H] .
\]
So the condition (\ref{haheeq}) can be restated as $\dist^*(R+1)\le G$.


Recall the random variables $\{\idist(t)\}_{t\in[R+1]}$ defined in \Cref{sec:bma}.
We claim that the random variable $\idist(t) + (\source^\ast(\ell^*) - i)$ stochastically dominates $\dist^\ast(t)$ for every $t \in [R+1]$.
To see this, observe that for every $c \ge 0$,
\[
  \Pr\bigl[ \idist(1) \le c\bigr]
  \le \Pr\bigl[ \dist^\ast(1) \le c + (\source^\ast(\ell^\ast) - i) \bigr]
  = 1 .
\]
Moreover, conditioned on $\idist(t) + (\source^\ast(\ell^*) - i) = \dist^\ast(t)$, the random variable $\idist(t+1) + (\source^\ast(\ell^*) - i)$ stochastically dominates $\dist^\ast(t+1)$. (They are identical when $\idist(t) \ne 0$.)
Also, for any $a \ge b$, we have that $\dist^\ast(t+1)$ conditioned on $\dist^\ast(t) = a$ stochastically dominates $\dist^\ast(t+1)$ conditioned on $\dist^\ast(t) = b$.

It follows from \Cref{claimhehe1} and \Cref{eqeqeq222} that
\begin{align*}
  \Pr[\dist^\ast(R+1) > G]
  &\le \Pr[\idist(R+1) > G - (\source^\ast(\ell^\ast) - i)] \\
  &\le \Pr[\idist(R+1) > G - 2H] \\
  &\le (s+1) \cdot \frac{\alpha}{1-\alpha} \left(\frac{2\alpha(1-\alpha)}{1-2\alpha}\right)^{2H} \\
  &\le 2^{-H} ,
\end{align*}
where we used $s = R/L = 2^{0.01L}$, and $G = 4H$, and $\alpha$ sufficiently small.

\ignore{

Similar to the analysis of $\BMA$ in \Cref{sec:bma},
  we let $\bX_0,\ldots,\bX_{s}$ be a Markov chain, with $s=R/L$, defined as follows:
$\bX_0=\source^*(\ell^*)-i\le 2H$ and the transition from $\bX_t$ to $\bX_{t+1}$ for each $t$ is 
  given by \Cref{transition1} and \Cref{transition2}.
It follows from the proof of \Cref{claimhehe1} that 
  the probability of $\dist^*(R+1)>G$ is at most that of $\bX_s> G$ 
  so we focus on the probability of the latter.\xnote{I feel that there must be better way to do what I did below.} 
Also note for the latter to happen, there must exist $t_1<t_2$ such that the following event $\calE_{t_1,t_2}$ holds:
  $$\bX_{t_1}\le 2H,\quad \bX_{t_2}>G=4H \quad\text{and}\rnote{Why is it ``$0 <$'' rather than ``$0 \leq$''  in what follows - couldn't the sequence of values go down from the initial $\bX_{t_1}$ value to 0 and then make it up to $G+1$?}\quad \red{0<}\bX_{t_1+1},\ldots,\bX_{t_2-1}\le G .$$
Below we consider each of these $s^2$ pairs of $t_1$ and $t_2$, upper-bound
  the probability of $\calE_{t_1,t_2}$ and then apply a union bound.

Fix $t_1<t_2$.
We let $\bY_1,\ldots,\bY_s$ denote the following random variables derived from the Markov chain $\bX_0,\bX_1,\ldots,\bX_s$:
  $\bY_\ell=-1$ if the transition from $\bX_{ \ell-1}$ to $\bX_{ \ell}$ follows the first line
  (of \Cref{transition1} or \Cref{transition2}) and $\bY_\ell=c$ if the transition follows the second line
   with $c\ge 1$.
We~note that $\bY_1,\ldots,\bY_s$ are i.i.d. random variables.
Also $\calE_{t_1,t_2}$ implies that $\sum_{\ell=t_1+1}^{t_2} \bY_\ell>2H$ so~it suffices to upper-bound
  the probability that $\sum_{\ell=t_1+1}^{t_2} \bY_\ell>2H$.
For notational convenience we focus on $\sum_{\ell\in [k]} \bY_\ell$ with $k=t_2-t_1$.

For this purpose, let $\bZ_1,\ldots,\bZ_k$ be the following i.i.d. random variables:
$$
\bZ_\ell=\begin{cases} -1 & \text{with probability $1-2\alpha$}\\
c & \text{with probability $(2\alpha)^c(1-2\alpha)$, for each $c \geq 1$}\end{cases}
$$
Then each $\bY_\ell$ is dominated by $\bZ_\ell$ (as $\alpha$ is sufficiently small so that $(2\alpha)^c(1-2\alpha)\ge \alpha^c$).
Furthermore, $\bZ_\ell$ is dominated by $2\bW_\ell-3$, where $\bW_\ell$
  follows the geometric distribution with $\bW_\ell=c+1$ with probability
  $(2\alpha)^c (1-2\alpha)$.
So it suffices to upper-bound the probability of 
 $\sum_{\ell\in [k]}  \bW_\ell> H+1.5k$.
 
Finally, the probability of $\sum_{\ell\in [k]}\bW_\ell>H+1.5k$ can be upper-bounded by
$$
2^{H+1.5k}\cdot (2\alpha)^{H+0.5k}\le 2^{-H}
$$
when $\alpha$ is sufficiently small.
By a union bound over $s^2$ many $\calE_{t_1,t_2}$ (recall that $s=R/L=2^{0.08H}$), we have that \Cref{haheeq} is violated
  with probability at most $2^{-0.8H}$.
}

We now apply McDiarmid's inequality.
The expected number of pairs $\ell^*$ and $i$ such that $\ssx$~in ${[i:i+2R-1]}$ has no
  long deserts and  \Cref{haheeq} is violated is at most
  $2^{-H}\cdot 2Hn$.
Since \Cref{haheeq} only depends on deletions in the window of $[i:i+R+4H]$,
  each random variable can only change the number we care about above by $O(R)$.
It follows from McDiarmid's inequality that with probability at least $1-\exp(-n^{0.1})$,
  the number of such pairs is at most $2^{-0.15H}n$.
When this happens, the total number of pairs of $\ell^*$ and $i$ that violate \Cref{haheeq}
  (including those $i$'s in \Cref{lem:fewlongdeserts} that have at least one long desert in $[i:i+2R]$ of $\ssx$)
  is at most 
  $$(2H {+1})\cdot 2^{-0.13H}n +2^{-0.15H}n\le 2^{-0.1H}n.$$
This finishes the proof of the lemma.
\end{proof}

We are now ready to present the list of conditions on $\bssx,\bssy^*,\bssy^{(1)},\ldots,\bssy^{(M)}$:

\begin{corollary}\label{corocoro}
With probability at least $1-1/n$, $\bssx\sim \{0,1\}^n$ and $\bssy^*,\bssy^{(1)},\ldots,\bssy^{(M)}\sim\Del_\delta(\bssx)$
  satisfy all conditions stated in \Cref{thm:align-restated}, \Cref{keykeyBMAlemma}, \Cref{skmplem} and \Cref{difflemma}.
\end{corollary}

\subsection{Deterministic Analysis}\label{mainanalysispart2}

We prove that when $\ssx,\ssy^*,\ssy^{(1)},\ldots,\ssy^{(M)}$ satisfy all conditions 
  in \Cref{thm:align-restated}, \Cref{keykeyBMAlemma}, \Cref{skmplem}, and \Cref{difflemma}, 
  the string $\ssw$ that $\MAIN$ returns on $\ssy^*,\ssy^{(1)},\ldots,\ssy^{(M)}$ must have edit distance at most $2^{-0.01H}n$ from $\ssx$. 
This, together with \Cref{corocoro}, finishes the proof of \Cref{theo:mainReconstruct}.

\def\goodset{\textsf{Good}}
\def\badset{\textsf{Bad}}




We start the proof with some notation. Let $P$ be the following interval of locations of $\ssy^*$:
$$
P=\Big[5\tau \log n: |\ssy^*|-\max\big(5\tau \log n,R\big)\Big].
$$
Let $Q$ be the set of locations $\ell^*\in P$ such that all three conditions below hold:
\begin{flushleft}\begin{enumerate}
\item[(i)] The output $(\ell^{(1)},\ldots,\ell^{(M)})$
  of $\Align\hspace{0.04cm}(\ell^*,\ssy^*,\ssy^{(1)},\ldots,\ssy^{(M)})$ satisfies
\begin{enumerate}
\item At least 90\% of $\source^{(m)}(\ell^{(m)})$, $m \in [M]$, agree on the same location $i^* \in [n]$, and \vspace{-0.1 cm}
\item Their consensus location $i^*$ satisfies 
  $i^*\in [\source^*(\ell^*) - 2H: \source^*(\ell^*) ]$;
\end{enumerate}
\item[(ii)] Running $\BMA$ on $ {\ssy^\ast,} \ssy^{(1)},\ldots,\ssy^{(M)}$~starting at $\image^\ast(\ssy^\ast), \image^{(1)}(i^*),\ldots, \image^{(M)}(i^*)$ returns $\ssx_{[i^*:i^*+R-1]}$ and 
  the consensus is achieved by at least $90\%$ of strings in every round;
\item[(iii)] The pair $\ell^*$ and $i^*$ satisfies $\SBMA (\ssy^*,\ell^*,\ssx_{[i^*:i^*+R-1]}) \le  i^*+R +G$.
\end{enumerate}\end{flushleft}
Using conditions from \Cref{thm:align-restated} (for (i)) \Cref{keykeyBMAlemma} (for (ii)), and \Cref{difflemma} (for (iii)),
  we have $|P\setminus Q|\le O(H2^{-0.1H}n)$.
If $\MAIN$ uses {a value $\ell^*$ that belongs to $Q$} in an execution of the main loop, the string concatenated
  to $\ssw$ during this execution of the loop must be $\ssx_{[i^*:i^*+R-1]}$ for some 
  $i^*\in [\source^*(\ell^*) - 2H: \source^*(\ell^*) ]$.
Furthermore, before incrementing $\ell^*$ at the end of this loop,
  we have from (iii) that $\source^*(\ell^*)\le i^*+R+G$.
In the rest of the proof, we write
$\ell^*_1,\ell^*_2,\ldots \in P$ to denote the 
  locations of $\ssy^*$ that are used in each execution of the main loop of $\MAIN$. 
We use $\textsf{Good}$ to denote the set of $k$ such that $\ell^*_k\in Q$ and
  $\textsf{Bad}$ to denote the set of $k$ such that $\ell^*_k\notin Q$.
For each $k$, we write $\ssw^k$ to denote the binary string concatenated
  to the end of $\ssw$ in Step~5 of the $k$-th execution of the loop,
so the output string of $\MAIN$ is $\ssw=\ssw^1\ssw^2\cdots$.
Our analysis bounding the edit distance between $\ssw$ and $\ssx$ will
  proceed in three steps.


\smallskip

{\bf First step:} In the first step we delete from $\ssw$ every $\ssw^k$ with $\ell_k^*\in \badset$.
Let $\ssw'$ denote the concatenation of all and only the $\ssw^k$ for which $\ell^*_k\in \goodset$.
We have that the edit distance between $\ssw$ and $\ssw'$ is at most 
\begin{equation*} 
|\badset|\cdot R\le |P\setminus Q|\cdot R\le  O(H 2^{-0.1H}n)
\cdot R\le  {2^{-0.011H}n},
\end{equation*}
where for the last step we recall that $R=8H2^{0.08H}.$

\smallskip 

{\bf Second step:}
 {After the first step $\ssw'$ is a concatenation of subwords of $\ssx$ but because these subwords are not necessarily disjoint $\ssw'$ is not necessarily a subsequence of $\ssx$ yet.}
In the second step we delete some bits of $\ssw'$ to obtain a subsequence of $\ssx$.
For each $k\in \goodset$, recall that
$i^*_k\in [\source^*(\ell_k^*)-2H: \source^*(\ell_k^*)]$ and
$\ssw^k=\ssx_{[i^*_k:i^*_k+R-1]}$, and that $\ssw'$ is the concatenation of $\ssw^k$ across all $k \in \goodset.$
For any two consecutive $k'<k$ in $\goodset$, we also have 
  $\source^*(\ell^*_{k})> i^*_{k'}+R$ using \Cref{hahehaheeq} and so
the windows $[
 {\source^\ast(\ell^\ast_k)}:i^*_k+R-1]$ are disjoint and thus, defining $\ssw''$ to be 
the concatenation of the subwords $\ssx_{[
 {\source^\ast(\ell^\ast_k)}:i^*_k+R-1]}$ of $\ssx$, we get that $\ssw''$ is a subsequence of $\ssx$.
To bound the edit distance between $\ssw'$ and $\ssw''$ we note that each 
new window $[
\source^\ast(\ell^\ast_k):i^*_k+R-1]$ can be obtained from $[i^*_k:i^*_k+R-1]$ by deleting
  no more than $2H$ indices at the beginning.
Using
$$
n=|\ssx|\ge |\ssw''|\ge|\ssw'|-|\goodset|\cdot 2H=|\goodset|\left(R-2H\right)
\quad \text{and}\quad 2H<R/2,
$$
we have that $|\goodset|\le 2n/R$ and thus,
the edit distance between $\ssw'$ and $\ssw''$ is at most
\begin{equation*} 
|\goodset|\cdot 2H \le \frac{2n}{R}\cdot 2H\le{2^{-0.07H}n}.
\end{equation*}

\smallskip

{\bf Third step:}
Given that $\ssw''$ (the concatenation of subwords of $\ssx$ in ${[i_k^{**}:i_k^*+R-1]}$ for 
  each $k\in \goodset$) is a subsequence of $\ssx$, to bound its edit distance
  from $\ssx$ it suffices to bound the number of $j\in [n]$ such that 
  $j\notin [i_k^{**}:i_k^*+R-1]$ for any $k\in \goodset$.
The following two cases cover every such $j\in [n]$:
\begin{flushleft}\begin{enumerate}
\item $j<\source^*(5\tau \log n)$ or $j>\source^*(|\ssy^*|-\max(5\tau\log n,R))$. There are only $O(R+\log n)$ many such $j$ by the first part of \Cref{skmplem}.
\item Otherwise, there is a unique $k$-th loop such that $\source^*(\ell^*_k) \le j<\source^*(\ell^*_{k+1})$. 
\end{enumerate}\end{flushleft}
We split the second case further into two cases: $k\in \goodset$ or $k\in \badset$.

We start with the case when $k\in \badset$ and bound the total number of $\ssx_i$ skipped.
In this loop, $\BMA$ starts with location $\ell_k^*$ of $\ssy^*$ and ends at location $\ell^*_{k+1}-1$.
Given that $\BMA$ only has $R$ rounds we have 
  $\ell^*_{k+1}-1\le \ell_k^*+R$.
The number of $j$'s skipped because of these loops is thus captured by 
\begin{align*}
  \sum_{k\in \badset} \bigl( \source^*(\ell^*_{k+1})-\source^*(\ell^*_k) \bigr)
  &= \sum_{k\in \badset}\sum_{\ell=\ell^*_k}^{\ell^*_{k+1}-1} \bigl( \source^*(\ell+1)-\source^*(\ell) \bigr) \\
  &{\le |\badset|R+ \sum_{k\in \badset}\sum_{\ell=\ell^*_k}^{\ell^*_{k+1}-1} \bigl( \source^*(\ell+1)-\source^*(\ell)-1 \bigr) }.
\end{align*} 

Using the second part of \Cref{skmplem}, the above can be upperbounded by
\[
   {|\badset| R+ |\badset| R\cdot 2H+2^{-0.2H}n}
   \le 2^{-0.011H}n.
\]

We finish with the case when $k\in \goodset$ and bound the total number of $j$'s skipped
  because of some $k\in \goodset$.
Given that every bit of $\ssx$ in the window of 
  $[ {
\source^\ast(\ell^\ast_k)}:i^*_k+R-1]$ is included,
  the number of $j$'s skipped is captured by
\begin{align*}
  \source^*(\ell_{k+1}^*) &- (i_k^*+R)\\ 
                          &=\source^*(\ell_{k+1}^*)-\source^*(\ell_{k+1}^*-1) 
                          +\source^*(\ell_{k+1}^*-1)-(i_k^*+R).
\end{align*}
Note that from (iii) we have $\source^*(\ell_{k+1}^*-1)-(i_k^*+R)\le G$.
As a result, the total number of $j$'s skipped is at most (again using the second part of \Cref{skmplem} and a similar argument as above)
$$
|\goodset| G +\sum_{k\in \goodset} \bigl( \source^*(\ell_{k+1}^*)-\source^*(\ell_{k+1}^*-1)\gray{-1} \bigr)
\le |\goodset|G+2^{-0.2H}n+|\goodset|\cdot 2H.
$$
Using $|\goodset|\le 2n/R$, the above is at most $2^{-0.07H}n$.

To summarize, the edit distance between $\ssw$ and $\ssx$ is at most 
$$
\stackrel{\text{first step}}{\overbrace{2^{-0.011H}n}}+
\stackrel{\text{second step}}{\overbrace{2^{-0.07H}n}}+
\stackrel{\text{third step, case 1}}{\overbrace{\vphantom{2^{-0.07H}}O(R+\log n)}}+
\stackrel{\stackrel{\text{\scriptsize{third step, case 2,}}}{k \in \goodset}}{\overbrace{2^{-0.011H}n}}+
\stackrel{\stackrel{\text{\scriptsize{third step, case 2,}}}{k \in \badset}}{\overbrace{2^{-0.07H}n}}
\le 2^{-0.01H}n.
$$
This finishes the proof of \Cref{theo:mainReconstruct}.
\qed


\section{Proof of \Cref{thm:mainlower}: Lower bound on expected edit distance from few traces}
\label{sec:lowerbound}

In this section we prove \Cref{thm:mainlower}. Recall that $\bssx \sim \{0,1\}^n$, $M \leq \Theta(1/\delta)$,
$\bssy^{(1)},\dots,\bssy^{(M)} \sim \Del_\delta(\bssx)$,  and {\tt A} is an arbitrary algorithm which, on input
$\delta$ and $\bssy^{(1)},\dots,\bssy^{(M)}$, outputs a hypothesis string $\widehat{\ssx}$ for $\bssx$. We prove \Cref{thm:mainlower} by
showing that $\E[\dedit(\widehat{\ssx},\bssx)] \geq n\cdot (\delta M)^{O(M)}$.

The main idea is to reduce the approximate trace reconstruction problem to the problem of computing the exact length of many runs of $0$'s ($1$-deserts) 
that are either of length $M$ or of length $M+1$.
We consider many instances of (a slight variation of) the following atomic problem: distinguish between a run of length $M$ and a run of length $M+1$, given $M$ ``traces" of the run at deletion rate $\delta$. 
This problem is equivalent to that of distinguishing between $\bX \sim \Bin(M, 1-\delta)$ and $\bX' \sim \Bin(M + 1, 1-\delta)$ given samples from a distribution that is either $\bX$ or $\bX'$. 
 (For technical reasons the actual atomic problem we work with, described in the next subsection, is the problem of distinguishing between two product distributions over non-negative integers which are closely related to these binomial distributions.)


\subsection{The atomic problem} \label{sec:atomic}

Consider the following two product distributions over pairs of non-negative integers:
\[
	\calD_0 := \Bin(M, 1-\delta) \times \Bin(M+1, 1-\delta); \quad \calD_1 := \Bin(M+1, 1-\delta) \times \Bin(M, 1-\delta).
\]


We consider a uniform prior distribution $\calP$ over $\{\calD_0, \calD_1\}$. In this subsection we prove the following lemma:

\begin{lemma} \label{lem:optdist}
Let ${\cal D}_{\bb} \sim \calP$, and let $p$ be the optimal (minimal) failure probability of any algorithm which is given a sample $\bS_M$ consisting of $M$ independent draws from ${\cal D}_{\bb}$ and aims to identify whether $\bb=0$ or $\bb=1$.  Then $p \geq (\delta M)^{c M}$ for some absolute constant $c>0.$
\end{lemma}

\begin{proof}
The optimal failure probability is achieved by the Bayes optimal predictor, which outputs 
0 if $\Pr[\calD_0 | \bS_M] \geq \Pr[\calD_1 | \bS_M]$ and outputs 1 if $\Pr[\calD_0 | \bS_M] < \Pr[\calD_1 | \bS_M]$. By Bayes' theorem, for $b \in \{0,1\}$ we have
\[
\Pr[\calD_b | \bS_M] = {\frac {\Pr[\bS_M  | \calD_b] \Pr[\calD_b]}{\Pr[\bS_M]}} =
 {\frac {\Pr[\bS_M  | \calD_b]}{2\Pr[\bS_M]}}
\]
so for any fixed outcome $S_M$ of the random variable $\bS_M$, the Bayes optimal predictor outputs 0 on $S_M$ if and only if 
\[
\Pr_{\bS_M \sim ({\cal D}_0)^M}[\bS_M=S_M] \geq
\Pr_{\bS_M \sim ({\cal D}_1)^M}[\bS_M=S_M].
\]

Consider the particular outcome of the $M$ draws which is

\[
\overbrace{(M,M-1), \dots, (M,M-1)}^{M\text{~pairs}},
\]
i.e.,~in each draw the outcome of the first coordinate is $M$ and the outcome of the second coordinate is $M-1$. It is clear that the Bayes optimal predictor, on this input, will output 1; to see this rigorously, the probability of this outcome under $\calD_1$ is 
\begin{equation} \label{eq:pair1}
\left({M + 1 \choose M} (1-\delta)^M \delta \cdot
{M \choose M-1} (1-\delta)^{M-1} \delta
\right)^M = 
\left((M+1)M (1-\delta)^{2M-1}\delta^2\right)^M
\end{equation}
while its probability under $\calD_0$ is 
\begin{equation} \label{eq:pair0}
\left({M \choose M} (1-\delta)^M  \cdot
{M+1 \choose M-1} (1-\delta)^{M-1} \delta^2
\right)^M = {\frac 1 {2^M}} \cdot \eqref{eq:pair1}
\end{equation}
But the probability of this outcome when the source distribution is $\calD_0$ is (recalling that $M \leq \Theta(1/\delta)$)
\[
\eqref{eq:pair0} = {\frac 1 {2^M}}\cdot  \eqref{eq:pair1} = \left(\Theta(M^2 \delta^2)\right)^M=
 \left(\Theta(M \delta)\right)^{\Theta(M)}.
\]
Since the probability that the source distribution is $\calD_0$ is $1/2$, it follows that the optimal failure probability for any $M$-sample algorithm for this distinguishing problem is at least
\[
{\frac 1 2}  \cdot \left(\Theta(M \delta)\right)^{\Theta(M)} = (M \delta)^{\Theta(M)}. \qedhere
\]
\end{proof}

\subsection{Direct sum (Paired Run Length Problem)} \label{sec:PRLP}

We define the \textbf{Paired Run Length Problem (PRLP)} as follows. Fix $B \in \N$. An instance of the PRLP is specified by a binary vector $\ssz = (\ssz_1, \ssz_2, \ldots, \ssz_B) \in \zo^B$. 
For an instance $\ssz$ of the PRLP, an algorithm is given as input samples of $B$-tuples of $M$ pairs from the product distribution $\calD_{\ssz} = \calD_{\ssz_1} \times \calD_{\ssz_2} \times \cdots \times \calD_{\ssz_B}$. 
It then must return some $\hat{\ssz} \in \zo^*$, 
with the objective of minimizing $\E[\dedit(\ssz, \hat{\ssz})]$.

We begin by recording a warmup lemma which states that the PRLP cannot be solved \emph{exactly} with success probability better than that obtained by solving each instance independently.

\begin{lemma} \label{lem:PRLP_exact}
Let $p = p(M, \delta) < 1/2$ be the optimal failure probability of any algorithm for the atomic problem from \Cref{lem:optdist}.
For a uniform $\bssz \sim \zo^B$, let {\tt A}$_{\mathrm{PRLP}}$ be any algorithm for the PRLP that is given $M$ samples from $\calD_{\bssz}$, and let $\hat{\bssz}$ be its output. Then $\Pr[\hat{\bssz} = \bssz] \leq (1-p)^B$.
\end{lemma}

\begin{proof}
This is a consequence of the independence of the distributions $\calD_{\bssz_{i}}, i \in [B]$; we now give details.
	We have
	\[
    \Pr[\hat{\bssz} = \bssz] = \prod_{i=1}^{B} \Pr\bigl[\bssz_{i} = \hat{\bssz}_{i} \mid \bssz_{k} = \hat{\bssz}_{k} \text{ for all } k \in [i-1] \bigr].
	\]
	Suppose there exists an algorithm {\tt A}$_{\mathrm{PRLP}}$ for the PRLP which outputs $\hat\ssz$ such that $\Pr[\hat{\bssz} = \bssz] > (1-p)^B$. Then there exists an $i^* \in [B]$ such that
	\begin{equation} \label{eq:PLRP-too-good-success-prob}
		\Pr \bigl[ \bssz_{i^*} = \hat{\bssz}_{i^*} | \bssz_{k} = \hat{\bssz}_{k} \text{ for all } k \in [i^*-1] \bigr] > 1 - p.
	\end{equation}
	We use this to construct a ``too good to be true" algorithm ${\tt A}$ for the atomic problem. Given $M$ samples of the distribution $\calD_{\bb}$ for a uniform (unknown) $\bb \sim \zo$, ${\tt A}$ ``embeds" the problem into the PRLP problem. Specifically, it draws $\bssz' \sim \zo^{B-1}$ and simulates $M$ samples of $\calD_{\bssz'_i}, i \in [B-1]$. Let $\bssz \in \zo^B$ be defined by
	\[
		\bssz_i = \begin{cases}
		\bssz'_i, & i < i^*\\
		\bb, & i = i^*\\
		\bssz'_{i-1}, & i > i^*.
		\end{cases}
	\]
	Clearly, $\bssz$ is uniformly random. ${\tt A}$ generates $M$ samples of $\calD_{\bssz}$ by appropriately concatenating the $M$ samples of $\calD_{\bb}$ and the simulated samples of $\calD_{\bssz'_i}, i \in [B-1]$. It then invokes {\tt A}$_{\mathrm{PRLP}}$ on the generated samples, 
	receives output $\hat{\bssz} \in \zo^B$, and checks whether $\bssz_{k} = \hat{\bssz}_{k}$ for all $k \in [i^*-1]$; if this is the case then it returns $\hat{\bb} := \hat{\bssz}_{i^*}$, and if it is not the case then it tries again by drawing a fresh independent $\bssz' \sim \zo^{B-1}$, repeating this until it is the case that 
	$\bssz_{k} = \hat{\bssz}_{k}$. 
	By \Cref{eq:PLRP-too-good-success-prob}, $\hat{\bb} = \bb$ with probability more than $1-p$, which contradicts the definition of $p$. \qedhere
\end{proof}

We use \Cref{lem:PRLP_exact} to give a lower bound on the expected edit distance for any algorithm for the PRLP:

\begin{lemma} \label{lem:PRLP}
Let $p = p(M, \delta) < 1/2$ be the optimal failure probability of any algorithm for the atomic problem from \Cref{lem:optdist}.
For a uniform $\bssz \sim \zo^B$, let {\tt A}$_{\mathrm{PRLP}}$ be any algorithm for the PRLP that is given $M$ samples from $\calD_{\bssz}$, and let $\hat{\bssz}$ be its output. Then $\E[\dedit(\bssz,\hat{\bssz})] \geq c' B \cdot p / \log (1/p)$, where $c' > 0$ is an absolute constant.
\end{lemma}

\begin{proof}
		Let $p' < p$ be a parameter to be specified later, and let $\calE$ be the event that $\dedit(\bssz,\hat{\bssz}) \leq Bp'$. (Since {\tt A}$_{\mathrm{PRLP}}$ can without loss of generality be taken to be deterministic, ${\cal E}$ is over the uniform random draw of $\bssz$ and the $M$ samples from ${\cal D}_{\bssz}.$)
    Fix a potential matching $\mu = (i_t, j_t)_{t\in [(1-p')B]}$ between $\bssz$ and $\hat{\bssz}$ that is of size $(1-p')B$. Let $\calE_\mu$ be the event that $\bssz_{i_t}$ is actually equal to $\hat{\bssz}_{j_t}$ for all $t$ (i.e.,~the potential matching actually is a matching between $\bssz$ and $\hat{\bssz}$). By \Cref{lem:PRLP_exact}, we have that
\ignore{
\rnote{This was
		
		``We first show that $\Pr[\calE_\mu] \leq (1-p)^{|\mu|}$. We have

\[
\Pr[\calE_\mu] = \prod_{t=1}^{(1-p')B} \Pr[\bssz_{i_t} = \hat{\bssz}_{j_t} | \bssz_{i_k} = \hat{\bssz}_{j_k} \text{ for all } k \in [t-1]].
\]

\ignore{
Note that $\bssz_{i_t}$ is still uniform when conditioned on values of $\bssz_{i_k}, k < t$. Let the samples from the underlying distribution $\calD_{\bssz}$ that $\calA$ is given be $\bs^{(1)}, \cdots, \bs^{(M)}$, where $\bs^{(m)} = (\bs^{(m)}_1, \cdots, \bs^{(m)}_B)$ and $\bs^{(m)} \sim \calD_{\bssz}$. Then the components $\bs^{(m)}_1, \cdots, \bs^{(m)}_B$ of the samples are mutually independent, as $\calD_{\bssz}$ is a product distribution, and for a fixed component $q \in [B]$, the set of samples $\bs^{(1)}_q, \cdots, \bs^{(M)}_q$ are independent of $\bssz_{q'}$ for all $q' \neq q$. Hence, we have
\[
	\Pr[\bssz_{i_t} = \hat{\bssz}_{j_t} | \bssz_{i_k} = \hat{\bssz}_{j_k} \text{ for all } k \in [t-1]] = \Pr[\bssz_{i_t} = \hat{\bssz}_{j_t}] \leq
	\begin{cases}
	1 - p, & i_t = j_t\\
	1/2, & i_t \neq j_t.
	\end{cases}
\]

As $p < 1/2$, we have $\Pr[\calE_\mu] \leq (1-p)^{|\mu|}$. Using the lower bound on $|\mu|$ and $1-p' \geq 1-p > 1/2$, we obtain
\[
		\Pr[\calE_\mu] \leq (1-p)^{|\mu|} \leq (1-p)^{(1-p')B} 
		 \leq \exp(-p B/2).
\]
}

By arguments similar to those in the proof of \Cref{lem:PRLP_exact}, we have
\[
	\Pr[\bssz_{i_t} = \hat{\bssz}_{j_t} | \bssz_{i_k} = \hat{\bssz}_{j_k} \text{ for all } k \in [t-1]] \leq
	\begin{cases}
	1 - p, & i_t = j_t\\
	1/2, & i_t \neq j_t.
	\end{cases}
\]

Using the fact that $1-p' \geq 1-p > 1/2$, we obtain''

Can we just go with ``By \Cref{lem:PRLP_exact}'' as above instead?}
}

\[
		\Pr[\calE_\mu] \leq (1-p)^{(1-p')B} 
		 \leq \exp(-p B/2),
\]
where we used $1-p'>1/2$ for the second inequality.

Now, by a union bound over all potential matchings of size $(1-p')B$, we get that
\begin{align*}
		\Pr[\calE] &\leq \sum_\mu \Pr[\calE_\mu] 
		\leq \ignore{(B p')} \binom{B}{Bp'}^2 \cdot \exp\left(-\frac{p B}{2}\right) 
		\leq \ignore{(B p')} \exp\left(B \left(2p' \log \frac{e}{p'} - \frac{p}{2}\right)\right).
\end{align*}
Choosing $p' = c_1 p / \log (1/p)$ for some small enough $c_1 > 0$, 
we have $\Pr[\calE] \leq 2^{- c_2 B p}$.  Hence we have $\E[\dedit(\bssz, \hat{\bssz})] \geq (1 - 2^{- c_2 B p}) \cdot (Bp') = c' B \cdot p / \log (1/p)$ for some absolute constant $c' > 0$, 
and the lemma is proved.\qedhere
\end{proof}

\subsection{Embedding and proof of \Cref{thm:mainlower}} \label{sec:embedding}

In this subsection we relate the PRLP to the average-case approximate trace reconstruction problem and prove \Cref{thm:mainlower}.  To explain the connection between average-case approximate trace reconstruction and the PRLP, let us define two subwords
\[
\alpha := 0^M 1 0^{M+1} 1 1, \quad \quad \beta := 0^{M+1} 1 0^M 1 1.
\]
Observe that for any string $\ssx \in \zo^n$ and any pair of distinct intervals $I, J \subset [n]$ such that $\ssx_I, \ssx_J \in \{\alpha, \beta\}$, $I$ and $J$ must be disjoint. 
Note that in a uniform random string $\bssx$, each of these subwords occurs with expected frequency $q=2^{-N}$, where $N := |\alpha| = |\beta| = 2M + 4$. 
Intuitively, given $M$ traces from $\Del_{\delta}(\bssx)$, determining whether a segment of $\bssx$ is in fact $\alpha$ or $\beta$ corresponds to a single instance of the atomic problem from \Cref{sec:atomic}, and determining this for $B$ disjoint segments corresponds to an instance of the PRLP. In the rest of this subsection we make this correspondence precise and show how a high-accuracy algorithm for $M$-sample average-case approximate trace reconstruction yields a high-accuracy algorithm for the PRLP; combining this with the lower bound on the PRLP from \Cref{sec:PRLP} gives \Cref{thm:mainlower}.

\ignore{
\gray{

We begin with a lemma showing that distinct occurrences of $\alpha$ or $\beta$ in any string can never overlap:

\begin{lemma} \label{lem-disjoint-special-blocks}
	Let $\alpha, \beta$ be as described above, and fix $x \in \zo^n$. If $I, J \subset [n]$ are distinct intervals such that $x_I, x_J \in \{\alpha, \beta\}$, 
	then $I$ and $J$ are disjoint.
\end{lemma}

\begin{proof}
	We will prove this by contradiction. Let $I = [i, i + N - 1], J = [j, j + N - 1]$ be distinct intervals such that $x_I, x_J \in \{\alpha, \beta\}$. Assume $i < j < i + N$ without loss of generality. Note that both $\alpha$ and $\beta$ begin with $010$ and end with $011$. We consider several cases:
	\begin{enumerate}
		\item Suppose $|I \cap J| = 1$. Then $j$ is the first position in $J$, so we must have $x_j = 0$. But $j$ is also the last position in $I$, so we must have $x_j = 1$. This leads to a contradiction.
		\item Suppose $|I \cap J| = 2$. Then, using similar arguments, we must have that $x_{[j-1, j]} = 01$ and $x_{[j-1, j]} = 11$, which is impossible.
		\item Suppose $|I \cap J| \geq 3$. Then $i < j \leq i + N - 3$, and we must have $x_{[j, j+2]} = 010$. However, $010$ does not appear as a subword of $\alpha$ or $\beta$ except as a prefix. \qedhere
	\end{enumerate}
\end{proof}

}
}

Fix a sufficiently small absolute constant $c<1$, and let {\tt A} be an algorithm for average-case approximate trace reconstruction which, given $M \leq c/\delta$ \ignore{
  }traces of a random string $\bssx \in \zo^n$ (and the value of $\delta$), returns $\hat{\bssx} \in \zo^*$ such that $\E[\dedit(\bssx, \hat{\bssx})] \leq n (\delta M)^{C M}$ for some constant $C$. 
Building on {\tt A}, in \Cref{fig:ALG_PRLP}
we provide a ``too good to be true'' (given \Cref{lem:PRLP}) algorithm {\tt A}$_{\mathrm{PRLP}}$ for the Paired Run Length Problem.

We give a description of the algorithm {\tt A}$_{\mathrm{PRLP}}$. Consider a uniformly random string $\bssx' \in \zo^n$, where $n := N \cdot 2^N \cdot B$ is chosen such that $\bssx'$ contains at least $B$ occurrences of $\alpha$ or $\beta$ with high probability.
Given $\bssx'$ and $\bssz \sim \zo^B$, we define another string $\bssx \in \zo^n$ by replacing the $b$-th occurrence of $\alpha$ or $\beta$ in $\bssx'$ with $\alpha$ if $\bssz_b = 0$ and with $\beta$ if $\bssz_b = 1$ as long as $b \le B$.
As any pair of occurrences of $\alpha$ or $\beta$ are disjoint, the above procedure is well-defined. We show in \Cref{lem-randomness-preservation} that $\bssx$ is uniformly random, and hence a set of $M$ traces from $\bssx$ is a legitimate input to {\tt A}.

The algorithm {\tt A}$_{\mathrm{PRLP}}$ for the PRLP of course does not have access to $\bssz$, but only to samples $\bs^{(1)}, \ldots, \bs^{(M)} \sim \calD_{\bssz}$.
Hence, it cannot generate $\bssx$ explicitly. However, we show that it can \emph{simulate} $M$ independent traces from $\bssx$ by generating $\bssx' \sim \zo^n$, followed by generating traces from the segments in $\bssx'$ that are disjoint from the occurrences of $\alpha$ or $\beta$, and then concatenating them appropriately with the traces of $\alpha$ or $\beta$ that are generated using the samples $\bs^{(1)}, \ldots, \bs^{(M)}$ (see the loop spanning lines~\ref{item:simulation_start}--\ref{item:simulation_end} 
of the algorithm). It then invokes {\tt A} on these traces to obtain a string $\hat{\bssx} \in \zo^*$, extracts from $\hat{\bssx}$ a binary vector $\hat{\bssz} \in \zo^*$ based on the first (at most) $B$ occurrences of $\alpha$ or $\beta$ in $\hat{\bssx}$, and returns it.

\begin{figure}[t!]
  \centering
\setstretch{1.2}
  \begin{algorithm}[H]
    \caption{{\tt A}$_{\mathrm{PRLP}}$}
		\DontPrintSemicolon
		\SetNoFillComment
    \KwIn{A list of $M$ samples $\bs^{(1)}, \ldots, \bs^{(M)}$ 
	from the product distribution $\calD_{\bssz}$ for an (unknown) uniformly random $\bssz \sim \zo^B$ (and the value of $\delta \in (0,1)$).}

		\KwOut{A string $\hat{\bssz} \in \zo^{\leq B}$.}
		
		Set $N = 2M + 4$ and $n = N \cdot 2^N \cdot B$.\\
		
		Generate a uniformly random string $\bssx' \in \zo^n$.\\
		
		
		Let $B' = \min\{B,$ number of occurrences of $\alpha$ or $\beta$ in $\bssx'\}.$\\
		
		For $b \in [B']$, let $I^{(b)} = [i^{(b)}, i^{(b)} + N - 1]$ be the $b$-th interval in $\bssx'$ such that $\bssx'_{I^{(b)}} \in \{\alpha, \beta\}$. Also set $i^{(B'+1)} = n + 1$.\label{item:generate_ATR_input}\\
		
		\For{$m \in [M]$\label{item:simulation_start}}{
		
			Set $\bssy^{(m)} \sim \Del_\delta\left(\bssx'_{[1, i^{(1)} - 1]}\right)$.\\
			
      Let $\bs^{(m)} = \left(\bs^{(m)}_{b,1}, \bs^{(m)}_{b,2}\right)_{b \in [B']}$, where $\left(\bs^{(m)}_{b,1}, \bs^{(m)}_{b,2}\right) \sim \calD_{\bssz_b}$.\\
			
			\For{$b \in [B']$}{
			
		Set $\bssy'^{(m)}_{b} \sim 0^{\bs_{b,1}^{(m)}} \circ \Del_\delta(1) \circ 0^{\bs_{b,2}^{(m)}} \circ \Del_\delta(11) \in \zo^{\le N}$.
				
				Set $\bar{\bssy}^{(m)}_{b} \sim \Del_\delta\left(\bssx'_{[i^{(b)} + N, i^{(b+1)} - 1]}\right)$.\\
				
        Append $\bssy'^{(m)}_{b} \circ \bar{\bssy}^{(m)}_{b}$ to the end of $\bssy^{(m)}$.\label{item:simulation_end}\\
			}
		}
		
		Run {\tt A} on $\delta$ and the $M$ strings $\bssy^{(1)}, \ldots, \bssy^{(M)}$ to obtain $\hat{\bssx} \in \zo^*$.\\
		
		
		Let $\hat{B} = \min\{B,$ number of occurrences of $\alpha$ or $\beta$ in $\hat{\bssx}\}.$\\
		
		\For{$b \in [\hat{B}]$}{
		
			Let $J^{(b)}$ be the $b$-th interval in $\hat{\bssx}$ such that $\hat{\bssx}_{J^{(b)}} \in \{\alpha, \beta\}$.\label{item:generate_output}\\
			
			Set $\hat{\bssz}_b = 0$ if $\hat{\bssx}_{J^{(b)}} = \alpha$, and $\hat{\bssz}_b = 1$ if $\hat{\bssx}_{J^{(b)}} = \beta$.\\
		}
\Return $\hat{\bssz} := \left(\hat{\bssz}_1, \hat{\bssz}_2, \ldots, \hat{\bssz}_{\hat{B}}\right)$.\\
\end{algorithm}
\caption{Algorithm {\tt A}$_{\mathrm{PRLP}}$ for the PRLP problem, given an algorithm {\tt A} for average-case approximate trace reconstruction.}
\label{fig:ALG_PRLP}
\end{figure}

We note that the different intervals $I$ in \cref{item:generate_ATR_input} are disjoint from each other, and likewise for the different intervals $J$ in \cref{item:generate_output}.

\begin{lemma} \label{lem-randomness-preservation}
	Let $\bssz \in \zo^B$ be uniformly random, let $\bssx'$ be uniform random over $\zo^n$. Let $B'$ be the minimum of $B$ and the number of occurrences of $\alpha$ or $\beta$ in $\bssx'$, and let $\bssx$ be obtained from $\bssx'$ by replacing the $b$-th occurrence of $\alpha$ or $\beta$ in $\bssx'$ with $\alpha$ if $\bssz_b = 0$ and with $\beta$ if $\bssz_b = 1$ for all $b \in [B']$. 
	Then $\bssx$ is uniformly random over $\zo^n$.
\end{lemma}

\begin{proof}
Fix any possible outcome $\ssx \in \zo^n$ of $\bssx$ and let $j = \min\{B,$ number of occurrences of $\alpha$ or $\beta$ in $\ssx\}.$ There are precisely $2^{j}$ outcomes $\ssx' \in \zo^n$ of $\bssx'$ for which it is possible that $\bssx'=\ssx'$ could give rise to $\bssx=\ssx$ (these are precisely the $2^{j}$ strings obtained by replacing the first $j$ occurrences of $\alpha$ or $\beta$ in $\ssx$ by $\alpha$ or $\beta$ in all possible ways). Each of these outcomes has probability $1/2^n$ under $\bssx'$ because $\bssx'$ is uniform random and for each such outcome there is a $1/2^{j}$ chance that the replacement yields $\bssx=\ssx$ from $\bssx'=\ssx'.$ Hence $\Pr[\bssx=\ssx]=2^{j} \cdot (1/2^n) \cdot (1/2^{j})=1/2^n$.
\ignore{
	For $j \in \N \cup \{0\}$, let $\calE_j$ be the event that there are $j$ occurrences of $\alpha$ or $\beta$ in $\bssx'$. Using \Cref{lem-disjoint-special-blocks}, we have that the occurrences of $\alpha$ or $\beta$ in $\bssx'$ are disjoint. For any interval $I$ of length $N$, $\bssx_I \in \{\alpha, \beta\}$ if and only if $\bssx'_I \in \{\alpha, \beta\}$. Thus, the number and position of occurrences of $\alpha$ or $\beta$ in $\bssx'$ are the same as those in $\bssx$.
	
	Further, conditioned on a subword in $\bssx'$ being $\alpha$ or $\beta$, it is equally likely to be $\alpha$ and $\beta$. In $\bssx$, this property is preserved. To see this, let $w$ be the $i$-th subword in $\bssx'$ that is $\alpha$ or $\beta$, for some $i \leq B$. It is replaced by $\alpha$ if $\bssz_i = 0$ and $\beta$ if $\bssz_i = 1$, and each $\bssz_i$ is uniformly random, independent of all $\bssz_j, j \neq i$.
	}
	 \qedhere
\end{proof}

\begin{lemma} \label{lem-many-special-blocks}
	$\bssx'$ (and hence $\bssx$) contains at least $B$ 
	disjoint occurrences of $\alpha$ or $\beta$ with probability at least $1 - \exp(-\Omega(B))$. 
\end{lemma}

\begin{proof}
	As $\bssx'$ is a uniformly random string, for any position $i \in [n - N + 1]$, we have that
	\[
		\Pr\bigl[\bssx'_{[i:i+N-1]} \in \{\alpha, \beta\}\bigr] = \frac{2}{2^N} = \frac{1}{2^{N-1}}.
	\]
	Let $q = 2^{-(N-1)}$. We divide $\bssx'$ into $s := 2^N \cdot B$ disjoint segments $\bssx'^{(j)}$ of length $N$. For $j \in [s]$, let $\calF_j$ be the indicator of the event that $\bssx'^{(j)} \in \{\alpha$, $\beta$\}. Then $\calF_j \sim \text{Ber}(q)$, and $\calF_j, j \in [s]$ are independent.
	
Let $\calF = \sum_j \calF_j$ denote the number of segments $\bssx'^{(j)}$ that are $\alpha$ or $\beta$. Note that $\calF$ is a lower bound on the overall number of disjoint occurrences of $\alpha$ or $\beta$, because we are only considering segments ending at positions which are integral multiples of $N$. Clearly, 
	\[
    \E[\calF] = \sum_{j\in[s]} \E[\calF_j] = s q = 2 B.
	\]
	By the Chernoff Bound, we have $\calF < B$
	with probability at most $\exp(-\Omega(B))$. Along with the observations above and the fact that $\bssx_I \in \{\alpha, \beta\}$ if and only if $\bssx'_I \in \{\alpha, \beta\}$ for an interval $I$ of length $N$, this concludes the proof.\qedhere
\end{proof}



Next we state and prove a crucial lemma which implies that if {\tt A} is a good algorithm for average-case approximate trace reconstruction, then {\tt A}$_{\mathrm{PRLP}}$ is a good algorithm for the PRLP problem:

\begin{lemma} \label{lem-edit-distance-upper-bound}
	If $\bssx'$ (and hence $\bssx$) contains at least $B$ 
	disjoint occurrences of $\alpha$ or $\beta$, 
	then $\dedit(\bssz, \hat{\bssz}) \leq 2 \cdot \dedit(\bssx, \hat{\bssx})$.
\end{lemma}

\begin{proof}
  Let $I^{(1)}, I^{(2)}, \ldots, I^{(S)}$ 
	be intervals of length $N$ corresponding to the 
	occurrences of $\alpha$ or $\beta$ in $\bssx$ (so by the assumption of the lemma, we have $S \geq B$).
  Similarly, let $J^{(1)}, J^{(2)}, \ldots, J^{(T)}$ 
	be intervals of length $N$ corresponding to the 
	occurrences of $\alpha$ or $\beta$ in $\hat{\bssx}$. Fix an optimal matching $\mu$ between $\bssx$ and $\hat{\bssx}$ corresponding to any longest common subsequence between those strings (if there is more than one choice for $\mu$ it can be selected from the optimal matchings arbitrarily).

	
    Consider the matching $\tau$ on $[S] \times [T]$, where $(s,t) \in \tau$ if and only if $\mu(I^{(s)}) = J^{(t)}$.
    Let $\tau'$ be the induced matching obtained by restricting $\tau$ to $[B] \times [\abs{\hat{\bssz}}]$.
    Note that $\tau'$ corresponds to a longest common subsequence of $\bssz$ and $\hat\bssz$.  
    We consider two cases.

    \begin{enumerate}

      \item Suppose $(s,t) \in \tau$ for some $s \le B$ and $t > B$.
    Note that this implies there are at least $t > B$ occurences of $\alpha$ or $\beta$ in $\hat{\bssx}$, and so we have $\abs{\hat\bssz} = \abs{\bssz}$.
    We claim that for every $t' \le B$, if $(s',t') \in \tau$ for some $s'$ then $s' \le B$; this is because otherwise we have $s, t' \le B$ and $s', t > B$, but $(s,t),(s', t') \in \tau$, contradicting our definition of matching.

    Therefore, for every $t \in [T], t \le B$ that is not matched to an element of $[S]$ in $\tau'$, it is also not matched to an element of $[S]$ in $\tau$, and hence either (1) some element in $J^{(t)}$ is not matched in $\mu$, or (2) the indices in $\bssx$ that are matched to $J^{(t)}$ do not form an interval.
    Either case contributes at least $1$ deletion in $\bssx$ or $\hat\bssx$ to $\dedit(\bssx,\hat\bssx)$.
    Since there are $\dedit(\bssz,\hat\bssz)/2$ such $t$'s, we have $\dedit(\bssz,\hat\bssz) \le 2 \, \dedit(\bssx,\hat\bssx)$.

      \item Otherwise, for every $s \le B$, if $(s,t) \in \tau$ then it must be that $t \le B$.
        Moreover, for every $s \le B$ that is not in $\tau'$, it is also not in $\tau$, and therefore either (1) some element in $I^{(s)}$ is not in $\mu$, or (2) the indices in $\hat\bssx$ that are matched to $I^{(s)}$ do not form an interval.
    Either case contributes at least $1$ deletion in $\bssx$ or $\hat\bssx$ to $\dedit(\bssx,\hat\bssx)$.
    Since $\abs{\bssz} \ge \abs{\hat\bssz}$, we have at least $\dedit(\bssz,\hat\bssz)/2$ such $s$', and so $\dedit(\bssz,\hat\bssz) \le 2 \, \dedit(\bssx,\hat\bssx)$.  \qedhere
  \end{enumerate}

\ignore{
	
	\gray{
	We first define $\bssz'$, a subsequence of $\bssz$ corresponding to the subset of the first $B$ occurrences of $\alpha$ or $\beta$ in $\bssx$ that were preserved exactly in the matching $\mu$.
	Formally, let $1 \leq i_1 < i_2 < \cdots < i_t \leq B$ and $1 \leq j_1 < j_2 < \cdots < j_t$ be a maximal pair of sequences such that
	\[
	\mu(I^{(i_k)}) = J^{(j_k)}
	\]
	for all $k \in [t]$. Clearly, $t \leq B$. 
	We define $\bssz' \in \zo^{t}$ by setting $\bssz'_{k} := \bssz_{i_k}$ for all $k \in [t]$.\footnote{We emphasize that this vector $\bssz'$ is only for the purpose of analysis and is not used in the algorithm.}
	
	\ignore{
	\begin{itemize}
		\item Initialize $b = 1, k = 1, \bssz' = \epsilon$.
		\item While $b \leq B$:
		\begin{itemize}
			\item If $\mu(I^{(b)}) = J^{(b')}$ for some $b'$, set $\bssz'_k = \bssz_b$ and increment $k$.
			\item Increment $b$.
		\end{itemize}
		\item Return $\bssz'$.
	\end{itemize}
	}
	Note that while $i_k \leq B$, it is possible that $j_k > B$ for some $k$. So, some of the occurrences $\hat{\bssx}_{J^{(j_k)}}$ of $\alpha$ or $\beta$ may not be among the first $B$ occurrences in $\hat{\bssx}$, and hence may not be reflected in the output $\hat{\bssz}$, because of earlier ``rogue" occurrences of $\alpha$ or $\beta$ in $\hat{\bssx}$ that have no exact analogue in $\bssx$.\footnote{A precise definition of a ``rogue" occurrence is given later.} We now show the following:
	\begin{enumerate}
		\item $\dedit(\bssz, \bssz') \leq \dedit(\bssx, \hat{\bssx})$.
		\item $\dedit(\bssz', \hat{\bssz}) \leq 2 \, \dedit(\bssx, \hat{\bssx})$.
	\end{enumerate}
	Clearly, these claims, along with the triangle inequality, suffice to prove the lemma.
	
	\begin{enumerate}
    \item \label{step:edit_distance_proof_1} As $\bssz'$ is a subsequence of $\bssz$, the edit distance $\dedit(\bssz, \bssz')$ must be achieved by $|\bssz| - |\bssz'|$ deletions in only $\bssz$ but not $\bssz'$.
    Fix one of the $\abs{\bssz} - \abs{\hat{\bssz}}$ indices $b \in [B] \setminus \{i_k: k \in [t]\}$. 
    We have that either some element in the interval $I^{(b)}$ is not in the matching $\mu$ (and so some bit is deleted in $\bssx_{I^{(b)}}$), or $I^{(b)}$ appears in $\mu$ but is not matched to any interval $J^{(b')}$.
    In either case, 
    the subword $\bssx_{I^{(b)}}$ must undergo at least $1$ insertion or deletion. 
Moreover, each such edit can only affect one block $I^{(b)}$, as they are disjoint. By optimality of $\mu$, we conclude that $\dedit(\bssz, \bssz') \leq \dedit(\bssx, \hat{\bssx})$.
		
		\item Recall that $\hat{\bssz}$ is determined entirely by the first $B$ occurrences of $\alpha$ or $\beta$ in $\hat{\bssx}$. We consider two cases:
		\begin{itemize}
			\item Suppose $j_k \leq B$ for all $k \leq t$. In other words, each of the first $B$ occurrences of $\alpha$ or $\beta$ in $\bssx$ is either matched to \emph{one of the first} $B$ occurrences of $\alpha$ or $\beta$ in $\hat{\bssx}$, or not matched to any occurrence of $\alpha$ or $\beta$ in $\hat{\bssx}$. In this case, clearly $\bssz'$ is a subsequence of $\hat{\bssz}$.
			By arguments similar to those in \Cref{step:edit_distance_proof_1} with $\hat{\bssz}$ playing the role of $\bssz$, we have $\dedit(\bssz', \hat{\bssz}) \leq \dedit(\bssx, \hat{\bssx})$.
			
    \item Otherwise, $j_k > B$ for some $k \in [t]$. In this case, for any $b' \le B$, we cannot have $\mu(I^{(b)}) = J^{(b')}$ for any $b > B$, due to the non-crossing nature of the matching $\mu$. So, each of the first $B$ occurrences of $\alpha$ or $\beta$ in $\hat{\bssx}$ is either (1) matched exactly with \emph{one of the first} $B$ occurrences of $\alpha$ or $\beta$ in $\bssx$, or (2) not matched with \emph{any} occurrence of $\alpha$ or $\beta$ in $\bssx$.

        For $b \le B$, we define an interval $J^{(b')}$ to be a ``rogue occurrence" if $b' \neq j_k$ for any $k \in [t]$.
        It follows by (2) that any rogue occurrence $J^{(b')}$ is one of the first $B$ occurrences of $\alpha$ or $\beta$ in $\hat{\bssx}$ with no exact analogue in $\bssx$ with respect to $\mu$.
			
			Let $\Delta := |\{k \in [t] : j_k > B\}|$. We show that $\dedit(\bssx, \hat{\bssx}) \geq \Delta$, and $\dedit(\bssz', \hat{\bssz}) \leq 2 \Delta$, proving the claim in this case.
			
      To see that $\dedit(\bssx, \hat{\bssx}) \geq \Delta$, note that by our definition of $\Delta$, there are $\Delta$ occurrences of $\alpha$ or $\beta$ among the first $B$ occurrences in $\bssx$ (and hence reflected in $\bssz'$), which are not mapped to the first $B$ occurrences of $\alpha$ or $\beta$ in $\hat{\bssx}$ (and hence not reflected in $\hat{\bssz}$). So, there must be $\Delta$ rogue occurrences $J^{(q_\ell)}$ in $\hat{\bssx}$, indexed by $1 \leq q_1 < q_2 < \cdots <q_\Delta \leq B$. For each such rogue occurrence $J^{(q_\ell)}$, by (2) there must be an insertion or deletion in the interval in $\bssx$ corresponding to $J^{(q_\ell)}$. Again, as the intervals $J^{(q_\ell)}, \ell \in [\Delta]$ are disjoint, the intervals in $\bssx$ corresponding to them are also disjoint, and so each such edit can only affect one interval. So $\dedit(\bssx, \hat{\bssx}) \geq \Delta$.
			
      Now we show that $\dedit(\bssz', \hat{\bssz}) \leq 2 \Delta$. Let $\gamma$ be the subsequence of $\bssz'$ obtained by deleting the $\Delta$ many $\bssz'_k = \bssz_{i_k}$ such that $j_k > B$. Let $\zeta$ be the subsequence of $\hat{\bssz}$ obtained by deleting the $\Delta$ bits corresponding to all rogue occurrences $J^{(q_\ell)}, \ell \in [\Delta]$.
      It is easy to observe that $\gamma = \zeta$ is a common subsequence of $\bssz'$ and $\hat{\bssz}$.
      Since we delete $\Delta$ deletions from each of $\bssz'$ and $\hat{\bssz}$ to obtain this subsequence, we have $\dedit(\bssz', \hat{\bssz}) \leq 2 \Delta$.
			\ignore{
			Otherwise, there exists some $b \in [B]$ such that $\mu(I^{(b)}) = J^{(b')}$ for some $b' > B$, and so $\bssz'$ is not a subsequence of $\hat{\bssz}$. Let $\dedit(\bssz', \hat{\bssz}) = \Delta$. For this to happen, there must be at least $\Delta$ ``rogue" occurrences of $\alpha$ or $\beta$ in $\hat{\bssx}$, where a ``rogue" occurrence is defined by an interval $J$ such that $\hat{\bssx}_J \in \{\alpha, \beta\}$ and $\mu^{(-1)}(J) \neq I^{(b)}$ for any $b \in B$. For each such occurrence, there must be an edit operation in $\bssx$ outside the intervals $I^{(b)}, b \in [B]$. Again, each edit gives rise to at most one such occurrence by \Cref{lem-disjoint-special-blocks}. 
			This shows that $\dedit(\bssz', \hat{\bssz}) = \Delta \leq \dedit(\bssx, \hat{\bssx})$.
			}
		\qedhere
		\end{itemize}
	\end{enumerate}
}

}
\end{proof}

\begin{proofof}{\Cref{thm:mainlower}}
	\Cref{lem:optdist} and \Cref{lem:PRLP} imply that for any algorithm {\tt A}$_{\mathrm{PRLP}}$ that solves the PRLP given $M$ samples from $\calD_{\bssz}$, its output $\hat{\bssz}$ satisfies
	\begin{equation} \label{eq:lower_bound_final}
	\E\bigl[\dedit(\bssz, \hat{\bssz})\bigr] \geq B \cdot (\delta M)^{c M}
	\end{equation}
	for some absolute constant $c > 0$.
	
	Now, let {\tt A} be any algorithm which, given $\delta$ and traces $\bssy^{(1)}, \ldots, \bssy^{(M)}$ from a random string $\bssx \in \zo^n$ as input, outputs a hypothesis string $\hat{\bssx}$ for $\bssx$ such that $\E[\dedit(\bssx, \hat{\bssx})] < n \cdot (\delta M)^{C M}$. 
	Consider the algorithm {\tt A}$_{\mathrm{PRLP}}$ described in \Cref{fig:ALG_PRLP}. By \Cref{lem-many-special-blocks}, $\bssx'$ (and hence $\bssx$) has at least $B$ disjoint occurrences of $\alpha$ or $\beta$ with probability at least $1 - e^{-\Omega(B)}$, in which case we have $\dedit(\bssz, \hat{\bssz}) \leq 2 \cdot \dedit(\bssx, \hat{\bssx})$ by \Cref{lem-edit-distance-upper-bound}. If $\bssx$ has fewer than $B$ disjoint occurrences of $\alpha$ or $\beta$, we have $\dedit(\bssz, \hat{\bssz}) \leq 2B$ as $\bssz \in \zo^B$ and $\hat{\bssz} \in \zo^{\leq B}$. So, we obtain
	\begin{align} \label{eq:lower_bound_contradiction}
    \E[\dedit(\bssz, \hat{\bssz})]
    &\leq e^{-\Omega(B)} \cdot 2B + \bigl(1 - e^{-\Omega(B)}\bigr) \cdot 2 \E\bigl[\dedit(\bssx, \hat{\bssx})\bigr]  \nonumber \\
    &\leq 4n \cdot (\delta M)^{C M} \nonumber \\
    &\leq B \cdot (\delta M)^{C' M}
	\end{align}
	for some suitable constant $C' > 0$. \Cref{eq:lower_bound_final,eq:lower_bound_contradiction} lead to the desired contradiction for $C' > c$, which concludes the proof of \Cref{thm:mainlower}.
	
\end{proofof}

\begin{flushleft}
\bibliography{allrefs}{}
\bibliographystyle{alpha}
\end{flushleft}

\end{document}